\documentclass[usenames,dvipsnames]{article}
\usepackage{blindtext,inputenc,graphicx,overpic}
\usepackage{amsmath,amssymb,amsthm,bm,mathrsfs,mathtools,enumitem}

\usepackage{multicol}
\usepackage{esint,braket}
\usepackage{fullpage,setspace}
\usepackage{authblk}
\usepackage{verbatim} 
\usepackage{caption}
\usepackage{subcaption}

\usepackage[
    backend=biber,
    backref=true,
    style=alphabetic,
    maxbibnames=99,
  ]{biblatex}
\DefineBibliographyStrings{english}{
  backrefpage={p.},
  backrefpages={p.}
}

\usepackage{tikz}

\usepackage{hyperref}
\hypersetup{
    colorlinks=true,
    linkcolor=BrickRed,
    filecolor=red,      
    urlcolor=Periwinkle,
    citecolor=teal,
}
\theoremstyle{definition}
\newtheorem{defn}{Definition}[section]
\newtheorem{theorem}{Theorem}[section]
\newtheorem{lemma}{Lemma}[section]
\newtheorem{remark}{Remark}[section]
\newtheorem{prop}{Proposition}[section]
\newtheorem{cor}{Corollary}[section]
\newtheorem{conjecture}{Conjecture}
\newtheorem{ex}{Example}[section]

\DeclareMathOperator{\tr}{tr}

\DeclareMathOperator{\CC}{\mathbb{C}}

\DeclareMathOperator{\PP}{\mathbb{P}}
\DeclareMathOperator{\ZZ}{\mathbb{Z}}

\DeclareMathOperator{\OO}{\mathcal{O}}

\DeclareMathOperator*{\Ress}{Res}
\newcommand{\Res}{\displaystyle\Ress}

\numberwithin{equation}{section}
\numberwithin{figure}{section}

\usepackage{todonotes}
\usepackage{lineno}

\usepackage{titlesec}
\titleformat{\section}{\centering\normalfont\scshape}{\thesection .}{1em}{}
\titleformat{\subsection}{\normalfont\scshape}{\thesubsection .}{1em}{}
\titleformat{\subsubsection}{\normalfont\scshape}{\thesubsubsection .}{1em}{}
\usepackage{abstract}

\title{\bf The Ising Model Coupled to 2D Gravity: Higher-order Painlev\'{e} Equations/The $(3,4)$ String Equation}
\author[1]{\scshape Nathan Hayford\thanks{\href{mailto:nhayford@kth.se}{nhayford@kth.se}}}
\affil[1]{\small\textit{Department of Mathematics, Royal Institute of Technology, Lindstedtsv\"{a}gen 25, SE 10044, Stockholm, Sweden}}
\date{\today}

\begin{document}

\maketitle
\vspace{10mm}
\begin{abstract}
    We study a higher-order Painlev\'{e}-type equation, arising as a string equation of the $3^{rd}$ order 
    reduction of the KP hierarchy.  This equation appears at the multi-critical point of the $2$-matrix model with quartic interactions, and describes the 
    Ising phase transition coupled to 2D gravity, cf. \cite{CGM,Douglas1}, and the forthcoming \cite{DHL1,DHL3}. We characterize this equation in terms of the isomonodromic deformations of a particular rational connection on $\PP^{1}$. We also identify the (nonautonomous) Hamiltonian structure associated to 
    this equation, and write a suitable $\tau$-differential for this system. This $\tau$-differential can be extended to the canonical coordinates of the associated Hamiltonian system, allowing us to verify Conjectures 1. and 2. of \cite{IP}. We also present a fairly general formula for the $\tau$-differential of a special class of resonant connections, which is somewhat simpler than that of \cite{BM}.
\end{abstract}

\tableofcontents

\section{Introduction.}
In this work, we mainly study the following pair of equations for two functions $U = U(t_{5},t_{2},x), V= V(t_{5},t_{2},x)$:
    \begin{equation} \label{string-equation}
        \begin{cases}
             0 = \frac{1}{2}V'' - \frac{3}{2}UV + \frac{5}{2}t_{5} V + t_{2},\\
             0 = \frac{1}{12} U^{(4)} -\frac{3}{4}U''U -\frac{3}{8}(U')^2+\frac{3}{2}V^2 + \frac{1}{2}U^3 - \frac{5}{12}t_{5}\left(3U^2 - U''\right) + x.
        \end{cases}
    \end{equation}
Here (and throughout the present work), $' = \frac{\partial}{\partial x}$. We will also sometimes instead write $x:=t_1$, as this notation is more convenient
in certain instances. We have tried to keep our notations for this equation close to those of \cite{DFGZJ}. The above, along with a collection of equations specifying the dependence of $U,V$ on $t_5,t_2$ (see Equations \eqref{U_mu}--\eqref{V_eta}), is known as the $(3,4)$ \textit{string equation}, and appears in the study of the Ising model coupled to 2D gravity, as we shall now make apparent.

This work is partially motivated by the recent work \cite{DHL1}, in which we set up a Riemann-Hilbert analysis of the $2$-matrix model with
quartic interactions, corresponding to the Ising model on random quadrangulations. In \cite{DHL1}, we replicated the results of \cite{Kazakov1,Kazakov2} for the genus-zero partition function, thus providing a fully rigorous proof of their formula. The next task in our program is to investigate the multi-critical point of this model, which corresponds to the Ising $(3,4)$ minimal model of conformal field theory coupled to $2D$-gravity. At the level of the steepest descent analysis, this amounts to finding the ``right'' model Riemann-Hilbert problem at the turning points, for which the matching condition is satisfied. Such a parametrix is presently absent from the literature; the current work aims to fill this gap. 

In finding such a parametrix, we are not completely in the dark; as usual, physicists have already provided us the foundations. An equation characterizing this critical point was first derived in \cite{Kazakov3,CGM,Douglas1}, and recognized to be a string equation to a $3^{rd}$ order reduction of the KP hierarchy. This equation is precisely \eqref{string-equation}. More generally, it is conjectured that \textit{all} critical points of the $2$-matrix model
are characterized by the so-called $(q,p)$-string equations (see the discussion in \S2), which arise as symmetry constraints of the KP hierarchy. We will not make any general statements about these string equations here. We continue with a more detailed description of the connection of Equation 
\eqref{string-equation} with the $2$-matrix model. 

We also are motivated from the works of Okamoto \cite{Okamoto,Okamoto2,Okamoto3} on the Hamiltonian structure of equations of Painlev\'{e} type. Indeed, one
can easily check that any meromorphic solution to \eqref{string-equation}, considered as an ODE in $x$, has no finite branch points, and possibly moveable
poles as its singularity set, and so this equation is of Painlev\'{e} type. It thus should admit a Hamiltonian representation, which should be 
consistent with its formulation in terms of an isomonodromy problem. Part of the aim of this paper is to make this statement precise.

\subsection{Connection to the $2$-matrix model.}
As previously mentioned, the above equation arises when studying the triple scaling limit of the $2$-matrix model with quartic interactions. The partition function
for this model is
    \begin{equation}\label{2-matrix-partition-function}
        Z_n(\tau,t,H;N) := \iint \exp\left\{N\tr\left[\tau X Y -\frac{1}{2}X^2-\frac{t\,e^H}{4}X^4 -\frac{1}{2}Y^2-\frac{t\, e^{-H}}{4}Y^4\right]\right\}dXdY,
    \end{equation}
where the integration here is carried out over the Cartesian product of the space of $n\times n$ Hermitian matrices with itself. This matrix model can be 
identified with the Ising model on random quadrangulations \cite{Kazakov1,Kazakov2}. The multicritical point of this model, which characterizes the Ising 
spin-ordering transition coupled to gravity, occurs at
    \begin{equation}
        t = t_c = -\frac{5}{72}, \qquad\qquad \tau = \tau_c = \frac{1}{4}, \qquad\qquad H = H_c = 0.
    \end{equation}
Evidently, since $t_c<0$, the matrix integral \eqref{2-matrix-partition-function} is non-convergent. We must therefore make an appropriate analytic continuation of
this integral in order to make sense of the multicritical point. This construction is demonstrated in \cite{DHL1}: here, by a slight abuse of notation, we shall
denote both the partition function and its analytic continuation by $Z_n(\tau,t,H;N)$. From \cite{DHL1}, it follows that this multicritical point arises from a special degeneration of the spectral curve, see in particular Figure 2.3, as well as the discussion in Section 2.2. Define
    \begin{align*}
        \delta t &:= -\frac{5}{108}t_5N^{-2/7}+\frac{25}{648}t_5^2N^{-4/7} + \frac{125}{1944}t_5^3N^{-6/7} -\frac{9}{164}x N^{-6/7},\\
        \delta H &:= \frac{2}{3}N^{-5/7}t_2,\qquad \qquad \delta \tau := -\frac{5}{12}t_5N^{-2/7}+\frac{1}{164}x N^{-6/7},
    \end{align*}
and put $\varkappa:=\frac{n}{N}$. The coupling parameters $x,t_2,t_5$ characterize deviations along the normal ($x$), and tangential ($t_2,t_5$) directions to the multicritical point, and come with their own scalings. In \cite{CGM,Douglas1} (cf. the earlier work \cite{Kazakov3} for the model without
the external field or temperature parameters), the following triple scaling limit is introduced:
    \begin{equation}
        \varkappa = 1-\frac{\delta t}{t_c}, \qquad H = H_c + \delta H,\qquad \tau = \tau_c + \delta\tau.
    \end{equation}
After scaling (and appropriate normalization), one finds that the partition function converges to (see the works \cite{GGPZ,DFGZJ}), as $n\to \infty$,
    \begin{equation}
        C^2\frac{d^2}{dx^2}\log Z_n(\tau,t,H;N) \to -U(t_5,t_2,x),
    \end{equation}
for some constant $C^2 >0$. Here, $U(t_5,t_2,x)$ is a solution to the string equation \eqref{string-equation}. This suggests that the multicritical 
partition function for this model is in fact a $\tau$-function of equation \eqref{string-equation}. Note that the representation of the partition function for the $2$-matrix model (resp. $1$-matrix model) as an isomonodromic $\tau$-function has been mathematically established \cite{Bertola-Marchal} (resp. \cite{BEH}). However, the fact that this statement remains valid after the multi-scaling limit is indeed nontrivial. For the equivalent statement for the $1$-matrix model, see \cite{BleherDeano}. It is the purpose of Part III of this series
of works \cite{DHL3} to make rigorous sense of this scaling limit in the case of the $2$-matrix model. In this work, we study the limiting object, i.e. the equation that results after performing
this scaling limit. 

One of the shortcomings of the work in the physics literature is that one is only able to identify that the multicritical partition function solves a particular
integrable equation; there is no indication from this analysis \textit{which} solution one has convergence to, or what properties the resulting solution has.
An important consequence of our analysis is that one can identify the particular solution of \eqref{string-equation} arising from the triple scaling limit of the 
2-matrix model, and, since we furnish a Riemann-Hilbert formulation of the equation, this solution is amenable to asymptotic analysis.

\subsection{Outline and Statement of Results.}
The remainder of this work is organized as follows. In \S2, we write the string equation
as a $3 + 3$ dimensional Hamiltonian system, in which the coordinates $(Q_U,Q_V,Q_W;P_U,P_V,P_W)$ are canonical. The induced flows
along the $t_{5},t_{2},$ and $x :=t_{1}$ directions are generated by (nonautonomous) Hamiltonians $H_{5}$, $H_{2}$, $H_{1}$, which pairwise 
commute with respect to the following Poisson bracket: if $f,g$ are functions of the variables $(Q_U,Q_V,Q_W;P_U,P_V,P_W)$, we define
    \begin{equation}\label{Poisson-Bracket}
        \{f,g\} := \sum_{a \in \{U,V,W\}} \left( \frac{\partial f}{\partial Q_{a}}\frac{\partial g}{\partial P_{a}} - \frac{\partial f}{\partial P_{a}}\frac{\partial g}{\partial Q_{a}}\right).
    \end{equation}
This is the essence of our next Proposition, which we will prove in \S2:
\begin{theorem}\label{Hamiltonian-Prop}
    Given a solution $U = U(t_{5},t_{2},x), V = V(t_{5},t_{2},x)$ of the equations \eqref{string-equation}, 
    \eqref{U_mu}--\eqref{V_eta}, define functions
        \begin{align}
            Q_U :&= U - \frac{4}{3}t_{5}, \qquad\qquad\qquad\qquad\qquad\qquad\, Q_V := V, \qquad\qquad  Q_W := U', \label{Darboux-Q}\\
            P_U :&= \frac{1}{4}\left( 3UU' - \frac{1}{3}U''' - \frac{7}{3}t_{5} U'\right), \qquad\quad\,\, P_V := V', \qquad\quad\,\,\,\,\, P_W := \frac{1}{12}U''-\frac{1}{6}t_{5} U + \frac{7}{18}t_{5}^2.\label{Darboux-P}
        \end{align}
    (Here, we recall that $x:= t_1$). Then, there exist functions $H_{5}$, $H_{2}$, $H_{1}$, polynomially
    dependent on $(Q_U,Q_V,Q_W,P_U,P_V,P_W)$, and on $t_{5},t_{2},t_{1}$, such that
        \begin{equation} \label{Hamilton-eq}
            \frac{\partial Q_{a}}{\partial t_k } = \frac{ \partial H_k }{\partial P_{a}}, \qquad\qquad 
            \frac{\partial P_{a}}{\partial t_k } = -\frac{ \partial H_k }{\partial Q_{a}},
        \end{equation}
    for $a \in \{U,V,W\}$, $k=1,2,5$. These functions are defined up to the addition of an explicit function of
    the variables $t_{5},t_{2},t_{1}$; these ``integration constants'' can be chosen so that\footnote{Some caution must be taken here; the symbol $\frac{\partial}{\partial t_k}$ is taken to mean $\frac{\partial}{\partial t_k}\bigg|_{P_{a},Q_{a} = const.}$ here.}
        \begin{equation}
            \left\{H_k,H_j \right\} + \frac{\partial H_k}{\partial t_j} - \frac{\partial H_j}{\partial t_k} = 0,
        \end{equation}
    for $k,j =1,2,5$. Here, the Poisson bracket $\{\cdot,\cdot\}$ is defined by Equation \eqref{Poisson-Bracket}. Furthermore, the Hamiltonians satisfy the stronger condition
        \begin{equation}
            \left\{H_k,H_j \right\} = \frac{\partial H_k}{\partial t_j} - \frac{\partial H_j}{\partial t_k} = 0.
        \end{equation}
    Explicitly, these functions are given by
\begin{align}
        H_1&= P_{U} Q_{W}+6 P_{W}^{2}-\frac{3}{8} Q_{U} Q_{W}^{2}+\frac{1}{2} P_{V}^{2}-\frac{1}{8} Q_{U}^{4}-\frac{3}{2} Q_{U} Q_{V}^{2}-t_1Q_U + 2t_2Q_V\nonumber\\
    &+ \frac{1}{8}t_5 (16 Q_{U} P_{W}-2 Q_{U}^{3}+4 Q_{V}^{2}-Q_{W}^{2}) - \frac{1}{2}t_5^2(4P_W-Q_U^2)+\frac{19}{27}t_5^3Q_U+\frac{41}{54}t_5^4-\frac{4}{3}t_5t_1\label{Hamiltonian-nu}\\
    H_2&=\frac{1}{2} P_{V} Q_{U} Q_{W}+\frac{1}{4} Q_{V} Q_{W}^{2}-2P_{U} P_{V}-6 P_{W} Q_{U} Q_{V}+Q_{V}^{3}+Q_{U}^{3} Q_{V} +2t_1Q_V\nonumber\\
    &+t_2(4P_W-Q_U^2) + \frac{1}{2}t_5(Q_{V} Q_{U}^{2}-P_{V} Q_{W}+4 Q_{V} P_{W})-2t_5t_2Q_U-\frac{65}{27}t_5^3Q_V-\frac{22}{9}t_5^2t_2\label{Hamiltonian-mu}\\
    H_5&=\frac{1}{2} Q_{W} P_{V} Q_{U} Q_{V}-\frac{3}{4} P_{U} Q_{W} Q_{U}^{2}-P_{U} P_{V} Q_{V}+P_{U} P_{W} Q_{W}+\frac{3}{8} Q_{V}^{4}-\frac{1}{128} Q_{W}^{4}+4 P_{W}^{3}\nonumber\\
    &-\frac{1}{16} Q_{U}^{6}-P_{W} P_{V}^{2}+P_{W} Q_{U}^{4}+P_{U}^{2} Q_{U}-\frac{9}{2} P_{W}^{2} Q_{U}^{2}-\frac{1}{8} Q_{U}^{3} Q_{V}^{2}+\frac{1}{8} Q_{U}^{2} P_{V}^{2}+\frac{3}{32} Q_{W}^{2} Q_{U}^{3}-\frac{1}{16} Q_{W}^{2} Q_{V}^{2}\nonumber\\
    &+t_1\left(2 Q_{U} P_{W}-\frac{1}{8} Q_{W}^{2}-\frac{1}{4} Q_{U}^{3}+\frac{1}{2} Q_{V}^{2}\right) + \frac{1}{2}t_2\left(Q_{V} Q_{U}^{2}-P_{V} Q_{W}+4 Q_{V} P_{W}\right)\nonumber\\
    &+t_5\bigg(\frac{3}{16} Q_{U}^{5}-2 P_{U}^{2}-\frac{1}{16} Q_{U}^{2} Q_{W}^{2}-\frac{1}{4} P_{W} Q_{W}^{2}+5 P_{W}^{2} Q_{U}-2 P_{W} Q_{U}^{3}-5 P_{W} Q_{V}^{2}+\frac{3}{4} Q_{U}^{2} Q_{V}^{2}\nonumber\\
    &-\frac{1}{4} P_{V}^{2} Q_{U}+\frac{1}{2} P_{V} Q_{V} Q_{W}+P_{U} Q_{U} Q_{W}\bigg) -t_2^2Q_U-t_1t_2(4P_W-Q_U^2)\nonumber\\
    &+ t_5^2\bigg(\frac{47}{12} Q_{U} Q_{V}^{2}-\frac{29}{18} P_{U} Q_{W}-\frac{3}{2} P_{W} Q_{U}^{2}+\frac{29}{48} Q_{U} Q_{W}^{2}+\frac{7}{18} Q_{U}^{4}-\frac{14}{9} P_{V}^{2}-\frac{20}{3} P_{W}^{2}\bigg)\nonumber\\
    &-\frac{1}{108}t_5^3(284 Q_{U} P_{W}-49 Q_{U}^{3}+152 Q_{V}^{2}-11 Q_{W}^{2})+\frac{19}{9}t_5^2t_1Q_U-\frac{65}{9}t_5^2t_2Q_V\nonumber\\
    &+\frac{1}{216}t_5^4(1304P_W-299Q_U^2)-\frac{2173}{972}t_5^5Q_U-\frac{2}{3}t_1^2-\frac{22}{9}t_5t_2^2+\frac{82}{27}t_5^3t_1-\frac{556}{243}t_5^6\label{Hamiltonian-eta}
\end{align}
    Conversely, if one starts with the functions $H_{5}$, $H_{2}$, $H_{1}$, Hamilton's equations \eqref{Hamilton-eq}
    for these functions are equivalent to the string equation \eqref{string-equation}, \eqref{U_mu}--\eqref{V_eta}.
\end{theorem}
The Darboux coordinates given above seemingly arise from thin air. We give an algorithm for how to construct a suitable set of coordinates for similar
systems in Appendix \ref{AppendixB}. It is tempting to think that this algorithm always produces a set of Darboux coordinates. However, we were unable to
prove this, and so have delegated the discussion of this algorithm to an appendix, in the hope that it clarifies the origin of these coordinates.

The above theorem is new in that there are (to the knowledge of the author) no known examples of Darboux coordinates for `twisted' rank $3$ isomonodromic systems,
although similar work in the case of rank $2$ connections has been performed by J. Dou\c{c}out and G. Rembado \cite{DR}. Although we have by no means solved the problem of finding such coordinates for general isomonodromic deformations, we hope that the above result can provide some insight in constructing such coordinates for more general systems, a program currently being pursued by O. Marchal, N. Orantin, and M. Alameddine, \cite{Marchal1,Marchal2,Marchal3}, among others.

Our next result demonstrates that the string equation arises as the isomonodromic deformations of a linear differential equation with rational coefficients.
The statement of this result relies on some terminology from integrable systems; we refer the reader to Appendix \ref{AppendixA} for the details.
\begin{prop} \label{PropA}
    The string equation \eqref{string-equation}, and its compatibility with the (reduced) KP-flows $Q^{2/3}_+,Q^{5/3}_+$, 
    is equivalent to the isomonodromic deformations with respect to $t_{5},t_{2},x$ of the following linear differential equation for a function $\Psi = \Psi(\lambda;t_{5},t_{2},x)$:
        \begin{equation}\label{L-problem}
            \frac{\partial \Psi}{\partial \lambda} = L(\lambda;t_{5},t_{2},x)\Psi,
        \end{equation}
    where
        \begin{align} \label{L-connection}
		&L(\lambda;t_{5},t_{2},x) = 
			\begin{psmallmatrix}
				0 & 0 & 1\\
				0 & 0 & 0\\
				0 & 0 & 0
			\end{psmallmatrix}\lambda^2 + 
			\begin{psmallmatrix}
				0 & 2t_{5} + \frac{1}{4}Q_U & -Q_V\\
				1 & 0 & 2t_{5} + \frac{1}{4}Q_U\\
				0 & 1 & 0
			\end{psmallmatrix}\lambda \\
			&+\begin{psmallmatrix}
				\frac{1}{8}Q_U^2 - P_W + \frac{1}{2}P_V-\frac{1}{4}t_{5} Q_U - \frac{1}{6}t_{5}^2 & L_{12} & L_{13}\\
                \frac{1}{2}Q_V - \frac{1}{4}Q_W & 2P_W - \frac{1}{4}Q_U^2 + \frac{1}{2}t_{5} Q_U +\frac{1}{3}t_{5}^2 & L_{23}\\
				t_{5}-\frac{1}{2}Q_U & \frac{1}{2}Q_V + \frac{1}{4}Q_W &  \frac{1}{8}Q_U^2 - P_W - \frac{1}{2}P_V-\frac{1}{4}t_{5} Q_U - \frac{1}{6}t_{5}^2
			\end{psmallmatrix} ,\nonumber
	\end{align}
where
    \begin{align*}
        L_{12} &:= \frac{5}{16}Q_U Q_W - P_U +\frac{1}{4}t_{5} Q_W -\frac{3}{8}Q_UQ_V-\frac{1}{2}t_{5} Q_V + t_{2},\\
        L_{13} &:= \frac{1}{16}Q_W^2 + \frac{7}{32}Q_U^3+\frac{3}{4}Q_V^2-\frac{3}{2}P_WQ_U+\frac{5}{16}t_{5} Q_U^2-2t_{5} P_W+\frac{1}{4}t_{5}^2Q_U + x+ \frac{8}{27}t_{5}^3,\\
        L_{23} &:= -\frac{5}{16}Q_U Q_W + P_U -\frac{1}{4}t_{5} Q_W -\frac{3}{8}Q_UQ_V-\frac{1}{2}t_{5} Q_V + t_{2}.
    \end{align*}
\end{prop}

We then set about defining a $\tau$-function for this system; as it will turn out, the differential equation \eqref{L-problem} shares the same problem as the 
equivalent problem for the linear system associated to Painlev\'{e} I (PI): either the leading coefficient of the pole of $L$ is not diagonalizable, or (as we shall 
see) a transformed version of it does have diagonalizable leading coefficient at infinity, but carries a resonant Fuchsian singularity at the origin. Thus, the standard definition of the $\tau$-differential as given in \cite{JMU1} does not apply. If we try to ignore the contribution from the resonant singularity (as is done for PI, cf. \cite{JMU2,LR,ILP}), it turns out the $\tau$-differential is not closed. Thus, we must provide an alternate definition of the $\tau$-differential; this is established in Section \ref{tau-extension}.
Although most of this work is dedicated to the study of the string equation \eqref{string-equation}, we were able to derive a fairly
general formula for the $\tau$-differential of a linear differential equation with polynomial coefficients whose leading term is not diagonalizable. The motivation for the class of equations we study arises from the so-called $(p,q)$ string equations (see \cite{GGPZ} for an overview). An alternative formula was derived by Bertola and Mo \cite{BM} in terms of spectral invariants; the formula we present is in terms of a residue in the local gauge, and thus may merit interest, as it gives an alternative way to compute the $\tau$-differential which is amenable to Deift-Zhou analysis. 
We thus present our result as a theorem:
\begin{theorem}\label{tau-theorem.}
    Fix $q\geq 2$, and consider the differential equation
        \begin{equation}\label{psi-ODE-thm}
            \frac{\partial \Psi}{\partial \lambda} = A(\lambda;{\bf t})\Psi,
        \end{equation}
    where $A(\lambda;{\bf t})$ is a $q\times q$ matrix, polynomial in $\lambda$, whose leading term is the nondiagonalizable matrix
       \begin{equation}
        A(\lambda;{\bf t}) = 
        \Lambda^r \lambda^k + \cdots,
    \end{equation}
for some $0 < r < q$, $k\geq 0$, where $\Lambda = \Lambda(\lambda)$ is
    \begin{equation*}
        \Lambda(\lambda) := 
            \begin{pmatrix}
                0 & 0 & \cdots & 0 & 0 & \lambda\\
                1 & 0 & \cdots & 0 & 0 & 0\\
                0 & 1 & \cdots & 0 & 0 & 0\\
                \vdots & \vdots & \ddots & \vdots & \vdots & \vdots\\
                0 & 0 & \cdots & 1 & 0 & 0\\
                0 & 0 & \cdots & 0 & 1 & 0
            \end{pmatrix}.
    \end{equation*}
Equation \eqref{psi-ODE-thm} admits a formal solution in a neighborhood of $\lambda = \infty$ of the form
    \begin{equation}
        \Psi(\lambda;{\bf t}) = \underbrace{g(\lambda)\left[\mathbb{I} + \frac{\Psi_1({\bf t})}{\lambda^{1/q}} + \OO(\lambda^{-2/q})\right]}_{\mathfrak{G}(\lambda;{\bf t})}e^{\Theta(\lambda;{\bf t})},
    \end{equation}
    where $\Theta(\lambda;{\bf t})$ is a diagonal matrix, polynomial in $\lambda^{1/q}$, and $g(\lambda) = \lambda^{\Delta_q}\mathcal{U}_q$, for some constant, diagonal, traceless matrix $\Delta_q$, and constant matrix $\mathcal{U}_q$ (see Equations \eqref{Delta-def}, \eqref{U-def} for the exact definitions).
    Let $t_{\ell}$ be the collection of isomonodromic times. If we define
    \begin{equation}
         \hat{\boldsymbol{\omega}}_{JMU} := \sum_{\ell}\left(\left\langle A(\lambda;{\bf t}) \frac{d \mathfrak{G}}{dt_{\ell}} \mathfrak{G}^{-1}\right\rangle - \left\langle\frac{\Delta_q }{\lambda} \frac{d\mathfrak{G}}{dt_{\ell}} \mathfrak{G}^{-1}\right\rangle \right)dt_{\ell},
    \end{equation}
    then we have that
    \begin{equation}
        {\bf d }\, \hat{\boldsymbol{\omega}}_{JMU} = 0.
    \end{equation}
    Moreover, this formula is gauge-invariant, in the following sense: If we replace $\mathfrak{G}(\lambda;{\bf t})$ by $\mathfrak{h}({\bf t})\mathfrak{G}(\lambda;{\bf t})$ in above expression, where $\mathfrak{h}({\bf t})$ is an upper-triangular matrix with $1$'s on the diagonal (and depending smoothly on ${\bf t}$), then the resulting differential $\hat{\boldsymbol{\omega}}_{JMU}$ does not change.
\end{theorem}
We can then formally define a $\tau$-function by $\hat{\boldsymbol{\omega}}_{JMU} = {\bf d }\,\log \tau({\bf t})$.
Using this definition in the case of \eqref{L-connection}, we obtain the following proposition:
\begin{prop} \label{Okamoto-JMU-equivalence}
    The (modified) JMU isomonodromic tau function for the isomonodromic system defined by \eqref{L-connection} is given by
        \begin{equation} \label{JMU-tau}
            {\bf d }\log \tau(t_{5},t_{2},t_{1}) = \frac{1}{2} \left(H_{5} dt_{5} +  H_{2} dt_{2} + H_{1} dt_{1} \right),
        \end{equation}
    where $H_{5}, H_{2}$, and $H_{1}$ are the Hamiltonians of Theorem \ref{Hamiltonian-Prop}.
\end{prop}
We shall see that this definition coincides (up to an overall multiplicative factor) with the $\tau$-function as defined by Okamoto \cite{Okamoto2,Okamoto},
justifying our modification of the isomonodromic $\tau$-function.

A ``dressed'' version of the $\tau$-function as defined here will be what appears as the critical partition function for the quartic 2-matrix model; this will be
the main result of the forthcoming work \cite{DHL3}.
This work is the analogy of the analyses of Painlev\'{e} I \cite{Okamoto2,FIKN}, which were subsequently used for the analysis of 
the critical points of the quartic and cubic $1$-matrix models \cite{DK0,BleherDeano}. 

There has been much interest in recent years concerning the dependence of the isomonodromic $\tau$-function on the monodromy data (equivalently, on any set of initial conditions for the isomonodromy equations) \cite{Bertola1,ILP,LR,IP}, in particular due to its applications in determining the
constant factors in the asymptotics of $\tau$-functions. Building on earlier works, in \cite{IP} the authors greatly simplify the procedure for calculating these constant factors for the $6$ Painlev\'{e} equations. They proposed two conjectures to this end, which we give the full statement of in Section \ref{tau-extension}.
In our situation, these conjectures are equivalent to the following proposition, which we prove in Section \ref{tau-extension}:
    \begin{prop}\label{tau-extended-theorem}
    The extended $\tau$-differential $\boldsymbol{\omega}_{0}$ for the system defined by \eqref{L-connection} is given by
        \begin{equation}
            \boldsymbol{\omega}_{0} = \frac{1}{2}\boldsymbol{\omega}_{cla} + dG,
        \end{equation}
    where $\boldsymbol{\omega}_{cla}$ is
        \begin{equation}
            \boldsymbol{\omega}_{cla} = \sum_{a \in \{U,V,W\}} P_{a} dQ_{a} - \sum_{k\in\{1,2,5\}} H_k dt_k,
        \end{equation}
    and $G$ is the polynomial
        \begin{equation}
            G = \frac{1}{7}\left[3t_1H_1 + \frac{5}{2}t_2 H_2 + t_5H_5 - P_UQ_U - \frac{3}{2}P_VQ_V - \frac{3}{2}P_WQ_W\right].
        \end{equation}
\end{prop}
The full definition of $\boldsymbol{\omega}_{0}$ is given in Section \ref{tau-extension}. This result is in agreement with the conjectures of \cite{IP}, extending these conjectures to higher rank systems. This sheds some light on the hamiltonian structure of such equations.

Finally, we would like to comment that we have accompanying Maple worksheets that we can provide upon request to supplement some of the proofs.

\subsection{Notations.}
Throughout this work, we will frequently make use of several notations without comment. We list some of these notations here, for the convenience of the reader.
    \begin{itemize}
        \item $\omega = e^{\frac{2\pi i}{3}}$ denotes the principal third root of unity,
        \item $E_{ij}$ will denote the $3\times 3$ matrix with a $1$ in the $ij^{th}$ position, and zeros elsewhere,
        \item If $A$ is a square matrix, the notation $\lambda^A$ is defined to mean $\lambda^A := \exp( A\log \lambda)$, where ``$\exp$'' here is the usual matrix exponential.
        \item Throughout, we make the identification of coordinates $t_1 \equiv x$.
    \end{itemize}
    
\subsection{Acknowledgements.}
This research was partially supported by the European Research Council (ERC), Grant Agreement No. 101002013. The author would also like to thank Marco Bertola, Maurice Duits and Seung-Yeop Lee for valuable discussions during the preparation of this manuscript. We would also like to thank an anonymous referee for pointing out a number of relevant references, and providing helpful improvements to the present work.

\section{Hamiltonian Structure of the $(3,4)$ String Equation.}
Here, we develop a Hamiltonian formulation of the string equation \eqref{string-equation}, \eqref{U_mu}--\eqref{V_eta}. We develop this formalism 
before moving to the isomonodromy setting, as the notations of this section will serve as convenient coordinates for parameterizing the solution to the isomonodromy problem. We will first define a set of Darboux coordinates, and show that the corresponding Hamilton equations are equivalent to the string equation. We further show that one can define a $\tau$-function in the sense of Okamoto \cite{Okamoto2, Okamoto} via this construction. 

In this section, we shall revert to the notation
    \begin{equation}
        x =t_1,
    \end{equation}
as it will be more convenient here when indexing sums.

\subsection{Hamiltonian Structure and Okamoto $\tau$-function.}
Here, we prove Theorem \eqref{Hamiltonian-Prop}, and define the Okamoto $\tau$-function. We assume that the function $U,V$ defining the Darboux coordinates
\ref{Darboux-Q}, \ref{Darboux-P} satisfy \ref{string-equation}, as well as the `compatibility conditions' coming from the reduced KP flows 
    \begin{align}
            \frac{\partial U}{\partial t_{2}} &= -2V',\\
            \frac{\partial V}{\partial t_{2}} &= \frac{1}{6}U''' - UU',\\
            \frac{\partial U}{\partial t_{5}} &=\frac{\partial}{\partial x}\left[-\frac{1}{6}UU'' + \frac{1}{8}(U')^2 + \frac{1}{4}U^3 - 
            \frac{1}{2}V^2 - \frac{5}{9}t_{5}\left(3U^2-U''\right) + \frac{4}{3}x\right],\\
            \frac{\partial V}{\partial t_{5}} &= \frac{\partial}{\partial x}\left[ \frac{1}{12}U''V - \frac{1}{4}U'V' + \frac{5}{16}U^2V - \left(\frac{5}{3}t_{5} + \frac{1}{4}U\right)^2V - t_{2} U\right].
        \end{align}
The precise origin/meaning of these flows is explained in Appendix \ref{AppendixA}.
    \begin{proof}
        Let us prove that the equations
            \begin{equation*}
                \frac{\partial Q_{a}}{\partial t_{1} } = \frac{ \partial H_{1} }{\partial P_{a}}, \qquad\qquad 
                \frac{\partial P_{a}}{\partial t_{1} } = -\frac{ \partial H_{1} }{\partial Q_{a}},
            \end{equation*}
        $a \in \{U,V,W\}$, can be integrated to a function $H_{1}$. By direct calculation, 
            \begin{equation*}
                \frac{\partial Q_U}{\partial t_{1}} = U' = Q_W.
            \end{equation*}
        On the other hand, Hamilton's equations tell us that
            \begin{equation*}
                Q_W = \frac{\partial Q_U}{\partial t_{1}} = \frac{\partial H_{1}}{\partial P_U}.
            \end{equation*}
        Integrating, we find that
            \begin{equation*}
                H_{1} = P_U Q_W + f(Q_U,Q_V,Q_W,P_V,P_W;t_{5},t_{2},t_{1}).
            \end{equation*}
        (note that $f$ is independent of the variable $P_U$). Next, we have that
            \begin{align*}
                \frac{\partial P_U}{\partial t_{1}} &= \frac{1}{4}\left( 3(U')^2 + 3UU'' - \frac{1}{3}U'''' - \frac{7}{3}t_{5} U''\right)\\
                    &= \frac{3}{8} (U')^2 + \frac{1}{2}U^3 + \frac{3}{2}V^2 - \frac{5}{4}t_{5} U^2-\frac{1}{6}t_{5} U'' + t_{1}\\
                    &= \frac{3}{8}Q_W^2 - 2t_{5} P_{W} + \frac{1}{2}Q_{U}^3+\frac{3}{4}t_{5} Q_{U}^2 - t_{5}^2Q_{U} + \frac{3}{2} Q_{V}^2 + t_{1} - \frac{19}{27}t_{5}^3,
            \end{align*}
        where we have used the string equation \eqref{string-equation} to rewrite $\frac{\partial P_U}{\partial t_{1}}$ in terms of
        the Hamiltonian variables $\{Q_{a},P_{a}\}$, and $t_{5},t_{2}$, $t_{1}$. Hamilton's equations tell us that
            \begin{equation*}
                \frac{3}{8}Q_W^2 - 2t_{5} P_{W} + \frac{1}{2}Q_{U}^3+\frac{3}{4}t_{5} Q_{U}^2 - t_{5}^2Q_{U} + \frac{3}{2} Q_{V}^2 + t_{1} - \frac{19}{27}t_{5}^3
                = \frac{\partial P_U}{\partial t_{1}}
                = -\frac{\partial H_{1}}{\partial Q_U} = -\frac{\partial f}{\partial Q_U}.
            \end{equation*}
        The left hand side of the above is independent of $P_U$, and so both sides can be integrated to obtain that
            \begin{equation*}
                f =-\frac{3}{8} Q_{U} Q_{W}^{2}+2 t_{5} Q_{U} P_{W}-\frac{1}{8} Q_{U}^{4}-\frac{1}{4} t_{5}  Q_{U}^{3}+\frac{1}{2} t_{5}^{2} Q_{U}^{2}-\frac{3}{2} Q_{U} Q_{V}^{2}-x Q_{U} +\frac{19}{27} t_{5}^{3} Q_{U} + \tilde{f}(Q_{V},Q_{W},P_{V},P_{W}).
            \end{equation*}
        Our expression for $H_{1}$ now reads
            \begin{equation*}
                H_{1} = P_U Q_W -\frac{3}{8} Q_{U} Q_{W}^{2}+2 t_{5}  Q_{U} P_{W}-\frac{1}{8} Q_{U}^{4}-\frac{1}{4} t_{5}  Q_{U}^{3}+\frac{1}{2} t_{5}^{2} Q_{U}^{2}-\frac{3}{2} Q_{U} Q_{V}^{2} - t_{1} Q_{U}  +\frac{19}{27} t_{5}^{3} Q_{U} + \tilde{f},
            \end{equation*}
        i.e. we have completely determined the dependence of $H_{1}$ on $Q_U, P_U$. Continuing in this fashion, one is able to determine the function $H_{1}$ up to an explicit function of the variables $t_{5},t_{2}, t_{1}$; similar calculations for the Hamilton equations in the variables $t_{2}$, $t_{5}$
        result in functions $H_{2}$, $H_{5}$, also defined up to the addition of an explicit function of the variables $t_{5},t_{2}, t_{1}$. Denote these 
        functions, which we will call ``integration constants'', by $c_k(t_{5},t_{2},t_{1})$, $k=1,2,5$. Calculating the Poisson brackets of the Hamiltonians in pairs, we obtain the equations
            \begin{align*}
                \left\{ H_{5}, H_{2}\right\} + \frac{\partial H_{5}}{\partial t_{2}} - \frac{\partial H_{2}}{\partial t_{5}} &= \frac{\partial c_5}{\partial t_{2}} - \frac{\partial c_2}{\partial t_{5}},\\
                \left\{ H_{5}, H_{1}\right\} + \frac{\partial H_{5}}{\partial t_{1}} - \frac{\partial H_{1}}{\partial t_{5}} &= \frac{\partial c_5}{\partial t_{1}} - \frac{\partial c_1}{\partial t_{5}} + \frac{4}{3}t_{1} - \frac{82}{27}t_{5}^3,\\
                \left\{ H_{2}, H_{1}\right\} + \frac{\partial H_{2}}{\partial t_{1}} - \frac{\partial H_{1}}{\partial t_{2}} &= \frac{\partial c_2}{\partial t_{1}} - \frac{\partial c_1}{\partial t_{2}}.
            \end{align*}
        So, for example, we can take 
            \begin{align*}
                c_5 =-\frac{2}{3} t_{1}^{2}+\frac{82}{27} t_{5}^{3} t_{1} -\frac{556}{243} t_{5}^{6}-\frac{22}{9} t_{5}  t_{2}^{2}, \qquad\qquad c_2 = -\frac{22}{9} t_{5}^{2} t_{2}, \qquad\qquad c_1 = -\frac{4}{3} t_{5}  t_{1} +\frac{41}{54} t_{5}^{4}.
            \end{align*}
        From this calculation one can see that the condition $\left\{ H_k, H_j\right\} = \frac{\partial H_k}{\partial t_j} - \frac{\partial H_j}{\partial t_k} = 0$ holds, for $k,j =1,2,5$.

        Conversely, suppose we start with the functions $H_{5}, H_{2},$ and $H_{1}$. We check only that the first Hamiltonian flow is equivalent to
        the string equation \eqref{string-equation}; the remaining equations can be obtained in an identical manner. Given $H_{1}$, Hamilton's equations in the 
        variable $t_{1}$ read
            \begin{align*}
                Q_U' &= \frac{\partial H_{1} }{\partial P_U} = Q_W, \qquad\qquad\qquad Q_V' = \frac{\partial H_{1} }{\partial P_V} = P_V,
                \qquad\qquad\qquad Q_W' = \frac{\partial H_{1} }{\partial P_W} = 12 P_W + 2t_{5} Q_U - 2t_{5}^2,\\
                P_U' &= -\frac{\partial H_{1} }{\partial Q_U} =\frac{1}{2}Q_U^3 + \frac{3}{2}Q_V^2 +\frac{3}{8}Q_W^2 +\frac{3}{4}t_{5} Q_U^2 -2t_{5} P_W-
                t_{5}^2Q_U -\frac{19}{27}t_{5}^3 + t_{1},\\
                P_V' &= -\frac{\partial H_{1} }{\partial Q_V} = 3Q_UQ_V - t_{5} Q_V - 2t_{2},
                \qquad\qquad P_W' = -\frac{\partial H_{1} }{\partial Q_W} = 12P_W - 2t_{5} Q_U-2t_{5}^2.
            \end{align*}
        If we define $U := Q_U+\frac{4}{3}t_{5}$, $V := Q_V$, then the first three equations tell us that $Q_W = U'$, $P_V = V'$, and $P_W = \frac{1}{12}U'' -\frac{1}{6}t_{5} U + \frac{7}{18}t_{5}^2$. 
        Making these substitutions into the equation $P_V' = -\frac{\partial H_{1} }{\partial Q_V}$, we obtain
            \begin{equation*}
                V'' = 3UV - 5t_{5} V - 2t_{2},
            \end{equation*}
        which is the second part of the string equation. Differentiating the equation $P_W' = -\frac{\partial H_{1} }{\partial Q_W}$ once more with respect to
        $t_{1}$, and inserting the expression for $P_U'$, we obtain the first part of the string equation.
        \end{proof}
        \begin{remark} (\textit{Homogeneous changes of coordinate.})
            Although the explicit equations for the Hamiltonians are rather unwieldy, the Hamiltonians themselves enjoy some nice properties, 
            as we shall see in the subsequent remarks. The first observation one can make is that $H_1,H_2$, and $H_5$ are weighted homogeneous polynomials, in the following sense.
            \begin{prop}
                Fix $\kappa \in \CC\setminus\{0\}$. Under the change of variables
                \begin{equation}
                    (Q_U,Q_V,Q_W,P_U,P_V,P_W,t_1,t_2,t_5)\mapsto (\kappa^2 Q_U,\kappa^3 Q_V,\kappa^3 Q_W,\kappa^5 P_U,\kappa^4 P_V,\kappa^4 P_W,\kappa^6 t_1,\kappa^5 t_2,\kappa^2 t_5),
                \end{equation}
            the Hamiltonians $H_1,H_2,H_5$ transform as
                \begin{equation}\label{Hamiltonian-Homogeneity}
                    (H_1,H_2,H_5)\mapsto (\kappa^8 H_1,\kappa^9 H_2,\kappa^{12}H_5).
                \end{equation}
            \end{prop}
                    
            One should note that the calculation of the integration constants in the above does not determine them uniquely; we have made a choice which is consistent with the formulae we shall meet later, and the requirement that the Hamiltonians are weighted homogeneous polynomials.
        \end{remark}
        \begin{remark}(\textit{$t_2\to 0$ limit.})
            There is a well-defined Hamiltonian system which emerges in the $t_2 \to 0$ limit, obtained by simultaneously sending $(Q_V,P_V,t_2)$ to zero. The result is a 
            $2+2$-dimensional completely integrable non-autonomous Hamiltonian system, in the variables $(Q_U,Q_W,P_U,P_W,t_1,t_5)$. The corresponding Hamiltonians are obtained
            by directly setting $Q_V=P_V=t_2 = 0$ in Formulas \eqref{Hamiltonian-nu}--\eqref{Hamiltonian-eta}. This Hamiltonian system corresponds to the $\ZZ_2$-symmetric
            reduction that we shall study in Section \ref{Z2-Symmetry-Breaking}.
        \end{remark}

        \begin{remark} (\textit{Okamoto $\tau$-function and a stronger integrability condition.})
        The above proposition also provides us with another useful object: since 
            \begin{equation}\label{mixed-der}
                \frac{\partial H_j}{\partial t_k} - \frac{\partial H_k}{\partial t_j} = 0,\qquad\qquad k,j =1,2,5,
            \end{equation}
        we have the following corollary:
            \begin{cor}
                Consider the differential
                    \begin{equation}\label{Okamoto-tau-function}
                        \boldsymbol{\omega}_{Okamoto} := H_{5}dt_{5} + H_{2}dt_{2} + H_{1}dt_{1}.
                    \end{equation}
                and let ${\bf d}$ denote the exterior differential in the variables $t_{5},t_{2},t_{1}$. Then, $\boldsymbol{\omega}_{Okamoto}$ is closed:
                    \begin{equation}
                        {\bf d }\,\boldsymbol{\omega}_{Okamoto} = 0.
                    \end{equation}
            \end{cor}
        Thus, we can locally integrate this differential up to a function on the parameter space, $\tau = e^{\int \boldsymbol{\omega}_{Okamoto}}$.
        This observation for similar Painlev\'{e} systems was made by Okamoto in \cite{Okamoto2,Okamoto}, and was used as the \textit{definition} of 
        the $\tau$
        -differential for such equations. This definition of the $\tau$-function is perhaps less familiar to the readership than the usual isomonodromic $\tau$-function defined by Jimbo, Miwa, and Ueno (JMU) \cite{JMU1}. As we shall see in Section \ref{tau-function-section}, the Okamoto definition coincides (up to an overall multiplicative constant) with the JMU isomonodromic $\tau$-function. 
        \begin{remark}
            In Appendix \ref{AppendixA} (see the discussion just above formula \eqref{P-operator}), we set a time labeled $t_4$ to $0$ in the definition of the string equation, as we claimed it could be eliminated
            by a translation $V\to V+ ct_4$. If we had left this time in throughout the above calculations, we would find that the string equation is also
            hamiltonian in the variable $t_4$, with Hamiltonian $H_4$ satisfying
                \begin{equation}
                    H_4 + \frac{5}{3}t_5 H_2 + \frac{4}{3}t_4 H_1 = 0.
                \end{equation}
            In other words, the flow along the direction $t_4$ is not independent of the other flows, although the equations arising from it are still
            integrable.
        \end{remark}        
        As a final remark, one can calculate that the coordinates we use here actually satisfy a stronger condition than \eqref{mixed-der} still: we have that
            \begin{equation}
               \frac{\partial}{\partial t_k}\bigg|_{P,Q = const.}H_j = \frac{\partial H_j}{\partial t_k},
            \end{equation}
        for any $k,j =1,2,5$.
        \end{remark}

\section{The Isomonodromy Approach.}
In this section, we study the isomonodromy approach to the $(3,4)$ string equation. A Riemann-Hilbert formulation of this equation is given, and the various
symmetries of the Riemann-Hilbert problem are studied.

We now want to study the monodromy preserving deformations of the equation
    \begin{equation} \label{Equation-1}
        \frac{\partial \Psi}{\partial \lambda} = L(\lambda;t_{5},t_{2},x)\Psi,
    \end{equation}
where $L(\lambda;t_{5},t_{2},x)$ is given by the expression
\begin{align} 
		&L(\lambda;t_{5},t_{2},x) = 
			\begin{psmallmatrix}
				0 & 0 & 1\\
				0 & 0 & 0\\
				0 & 0 & 0
			\end{psmallmatrix}\lambda^2 + 
			\begin{psmallmatrix}
				0 & 2t_{5} + \frac{1}{4}Q_U & -Q_V\\
				1 & 0 & 2t_{5} + \frac{1}{4}Q_U\\
				0 & 1 & 0
			\end{psmallmatrix}\lambda \\
			&+\begin{psmallmatrix}
				\frac{1}{8}Q_U^2 - P_W + \frac{1}{2}P_V-\frac{1}{4}t_{5} Q_U - \frac{1}{6}t_{5}^2 & L_{12} & L_{13}\\
                \frac{1}{2}Q_V - \frac{1}{4}Q_W & 2P_W - \frac{1}{4}Q_U^2 + \frac{1}{2}t_{5} Q_U +\frac{1}{3}t_{5}^2 & L_{23}\\
				t_{5}-\frac{1}{2}Q_U & \frac{1}{2}Q_V + \frac{1}{4}Q_W &  \frac{1}{8}Q_U^2 - P_W - \frac{1}{2}P_V-\frac{1}{4}t_{5} Q_U - \frac{1}{6}t_{5}^2
			\end{psmallmatrix} ,\nonumber
	\end{align}
and
    \begin{align*}
        L_{12} &:= \frac{5}{16}Q_U Q_W - P_U +\frac{1}{4}t_{5} Q_W -\frac{3}{8}Q_UQ_V-\frac{1}{2}t_{5} Q_V + t_{2},\\
        L_{13} &:= \frac{1}{16}Q_W^2 + \frac{7}{32}Q_U^3+\frac{3}{4}Q_V^2-\frac{3}{2}P_WQ_U+\frac{5}{16}t_{5} Q_U^2-2t_{5} P_W+\frac{1}{4}t_{5}^2Q_U + x+ \frac{8}{27}t_{5}^3,\\
        L_{23} &:= -\frac{5}{16}Q_U Q_W + P_U -\frac{1}{4}t_{5} Q_W -\frac{3}{8}Q_UQ_V-\frac{1}{2}t_{5} Q_V + t_{2}.
    \end{align*}
(Note that this expression for $L$ coincides with the definition of $\mathcal{P}$ in Appendix \ref{AppendixA}, with the definitions \eqref{Darboux-Q}, \eqref{Darboux-P} taken into account). 
\begin{remark} \textit{A formula for the spectral curve.}
    Similarly to the works \cite{BHH,BM}, one can recover the Hamiltonians of the previous section from the spectral curve. The spectral curve corresponding to $L(\lambda;t_5,t_2,x)$ admits an explicit representation in terms of these Hamiltonians:
    \begin{equation}
        0=\det(w \mathbb{I} - L(\lambda;t_5,t_2,x)) = w^3 - \left[5t_5\lambda^2+2t_2\lambda +\frac{1}{2}H_1 + \frac{5}{3}t_5x\right]w - \ell_0(\lambda;t_5,t_2,x),
    \end{equation}
where $\ell_0$ is the degree $4$ polynomial
    \begin{equation}
        \ell_0(\lambda;t_5,t_2,x) = \lambda^4 + \left(\frac{125}{27}t_5^3+x\right)\lambda^2 + \left(\frac{1}{2}H_2+\frac{50}{9}t_5^2t_2\right)\lambda
        + \frac{1}{2}H_5 + \frac{25}{18}t_5^2H_1 + \frac{20}{9}t_5t_2^2+\frac{1}{3}x^2.
    \end{equation}
\end{remark}

However, as an ODE, Equation \eqref{Equation-1} is ``defective''; the leading coefficient at the only singularity of $L$ ($\lambda = \infty$) is not diagonalizable. Thus, the usual technology used for linear differential equations with rational coefficients \cite{Wasow,JMU1,FIKN} does not directly apply. This 
situation is reminiscent of the situation for the so called ``Fuchs-Garnier'' Lax pair for Painlev\'{e} I (see C2 of \cite{JMU2}). The resolution in the case of Painlev\'{e} I, discovered in \cite{JMU2}, is to make an appropriate gauge transformation which (after a change of variables $\lambda = \xi^2$) diagonalizes the leading term at infinity, at the price of introducing a resonant Fuchsian singularity at the origin (see C5 of \cite{JMU2}). 

The first goal of this section is to try and find an analogous transformation for the equation \eqref{Equation-1}.  We have the following Proposition:
\begin{prop}
    Define the matrix 
    \begin{equation}\label{gauge-matrix}
        g(\lambda) = 
        \frac{i}{\sqrt{3}}\underbrace{\begin{pmatrix}
            \lambda^{1/3} & 0 & 0\\
            0 & 1 & 0\\
            0 & 0 & \lambda^{-1/3}
        \end{pmatrix}}_{\lambda^{\Delta/3}}
        \underbrace{\begin{pmatrix}
            1 & \omega & \omega^2\\
            1 & 1 & 1\\
            1 & \omega^2 & \omega
        \end{pmatrix}}_{-i\sqrt{3}\mathcal{U}},
    \end{equation}
    and set $\Psi := g \Phi$ (note that $\det g(\lambda) = 1$, $\Delta = \text{diag }(1, 0, -1)$, and that $\mathcal{U}^{\dagger}\mathcal{U} = \mathcal{U}\mathcal{U}^{\dagger} = \mathbb{I}$). Then, if $\Psi$ satisfies the ODE \eqref{Equation-1}, after the change of variables $\lambda = \xi^3$, the function $\Phi := \Phi(\xi;t_{5},t_{2},x)$ satisfies the ODE
        \begin{equation} \label{Phi-ODE}
            \frac{\partial \Phi}{\partial \xi} = \mathcal{L}(\xi;t_{5},t_{2},x)\Phi,
        \end{equation}
    where 
        \begin{equation} \label{L-zeta}
            \mathcal{L}(\xi) = 3
            \begin{pmatrix}
                1 & 0 & 0\\
                0 & \omega & 0\\
                0 & 0 & \omega^2
            \end{pmatrix} \xi^6 + \sum_{k=0}^4 \mathcal{L}_k \xi^k + 
            \frac{ \frac{i}{\sqrt{3}}\begin{psmallmatrix}
                0 & -1 & 1\\
                1 & 0 & -1\\
                -1 & 1 & 0
            \end{psmallmatrix}}{\xi}.
        \end{equation}
\end{prop}
\begin{proof}
    The proof is a direct calculation. Explicitly, one has that $\mathcal{L}(\xi) = 3\xi^2 \tilde{L}(\xi^3)$, 
    where 
        \begin{equation*}
            \tilde{L}(\lambda) = \left[g^{-1}L(\lambda)g - g^{-1}\frac{dg}{d\lambda}\right].
        \end{equation*}
\end{proof}
\begin{remark}\label{gauge-remark}
    Although the proof of the above proposition is straightforward, some remarks are in order.
        \begin{enumerate}
            \item Note that the matrix $\mathcal{U}$ conjugates $\mathcal{L}$ from the outside; if we had instead simply defined the gauge transformation simply by
            $g(\lambda) := \lambda^{\Delta/3}$, and subsequently made the change of variables $\lambda=\xi^3$, the effect we set out for (making the leading coefficient at infinity diagonalizable) would still be achieved. In other words, after an appropriate change of variables, $\lambda^{\Delta/3}$ makes the leading coefficient of $\mathcal{L}$ at infinity \textit{diagonalizable}, and $\mathcal{U}$ makes this coefficient \textit{diagonal}.
            \item The choice of such $g$ is not unique; one may also additionally multiply $g(\lambda)$ on the right by upper triangular matrix $\mathfrak{h}({\bf t})$ with $1$'s on the diagonal, and obtain an operator $\mathcal{L}$ with the same form. This gauge freedom will become important later.
            \item Note the appearance of a resonant singularity at the origin. The residue at $0$ of $\mathcal{L}$ has eigenvalues $\pm 1$, $0$. Thus, there is no 
            monodromy around this singularity: the solution will have a first order pole at zero. The form of the solution near $\xi = 0$ is
                \begin{equation} \label{phi-origin-solution}
                    \Phi(\xi) = \left[\mathcal{U}^{-1}+\OO(\xi)\right]\xi^{-\Delta}=\frac{\mathcal{U}^{-1} E_{11}}{\xi} + \OO(1),\qquad \xi\to 0,
                \end{equation}
        \end{enumerate}
    \end{remark}

We now return to the analysis of the equation \eqref{Phi-ODE}. Since $\mathcal{L}$ has diagonal leading coefficient at infinity, a standard theorem in the
theory of linear ODEs with rational coefficients states that
    \begin{prop}
        The ODE \eqref{Phi-ODE} admits the formal series solution at $\xi = \infty$
            \begin{equation} \label{Phi-expansion}
                \Phi(\xi) = \left[\mathbb{I} + \frac{\Phi_1}{\xi} + \frac{\Phi_2}{\xi^2} + \OO(\xi^{-3})\right] e^{\Theta(\xi;t_{5},t_{2},x)},
            \end{equation}
        where $\Theta(\xi;t_{5},t_{2},x) = \text{diag }(\vartheta_1(\xi;t_{5},t_{2},x),\vartheta_2(\xi;t_{5},t_{2},x),\vartheta_3(\xi;t_{5},t_{2},x))$, and
            \begin{equation} \label{xi-exponents}
                \vartheta_j(\xi;t_{5},t_{2},x) = \frac{3}{7}\omega^{j-1}\xi^7 + \omega^{1-j}t_{5}\xi^5 + \omega^{1-j}t_{2}\xi^2 + \omega^{j-1}x\xi,
            \end{equation}
        $j=1,2,3$.
    \end{prop}
\begin{proof}
    We refer to \cite{Wasow,FIKN,JMU1} for the details. The exact form the exponential part of the asymptotics can be inferred by considering the 
    eigenvalues of the matrix $\mathcal{L}(\xi;t_{5},t_{2},x)$; indeed, one can readily check that the expressions $\vartheta_j(\xi;t_{5},t_{2},x)$ are the
    principal part of the eigenvalues of $\mathcal{L}(\xi;t_{5},t_{2},x)$ at $\xi = \infty$.
\end{proof}
\begin{remark}
    The previous proposition implies in turn that (by ``undoing'' the gauge transformation) that $\Psi$ admits the formal expansion 
    \begin{equation}
        \Psi(\lambda) = g(\lambda)\left[\mathbb{I} + \frac{\Phi_1}{\lambda^{1/3}} + \frac{\Phi_2}{\lambda^{2/3}} + \OO(\lambda^{-3})\right] e^{\Theta(\lambda^{1/3};t_{5},t_{2},x)}.
    \end{equation}
Note that the coefficients in the subexponential part of the expansion agree with the corresponding coefficients of $\Phi$.
These asymptotics are precisely what appear in the local parametrices of the critical quartic $2$-matrix model. Thus, we can use results about $\Phi$ to
construct our model Riemann-Hilbert problem.
\end{remark}

The explicit form of the coefficient matrices $\mathcal{L}_k$ is not so important at this stage; the only immediately relevant information is the form of the 
formal asymptotic expansion for $\Phi(\xi)$, as in Equation \eqref{Phi-expansion}.

\begin{remark}
    For completeness, we record the form of the ``regularized'' spectral curve here:
        \begin{equation}
            0=\det\left[w\mathbb{I} - \mathcal{L}(\xi;t_5,t_2,x)\right] = w^3 - \left[45t_5\xi^{10} + 18t_2\xi^7 + \left(\frac{9}{2}H_1+15t_5t_1\right)\xi^4 - 3\frac{\partial H_1}{\partial P_V}\xi + \frac{1}{\xi^2}\right]w - \tilde{\ell}_0,
        \end{equation}
    where $\tilde{\ell}_0 = \tilde{\ell}_0(\lambda;t_5,t_2,t_1)$ is
        \begin{align}
            \tilde{\ell}_0(\lambda;t_5,t_2,t_1) &= 27\xi^{18} + (125t_5^3+27x)\xi^{12} + \left(\frac{27}{2}H_2+150t_5^2t_2\right)\xi^9
            + \bigg(\frac{27}{2}H_5 + \frac{75}{2}t_5^2H_1 - 9\frac{\partial H_2}{\partial P_V} \nonumber\\
            &+ 60t_5t_2^2 + 9x^2\bigg)\xi^6 - \frac{\partial}{\partial P_V}\left(9H_5 + 25t_5^2H_1\right)\xi^3 - 6P_W + \frac{3}{4}Q_U^2-\frac{3}{2}t_5Q_{U}-t_5^2.
        \end{align}
\end{remark}

Before proceeding to the construction of an appropriate Riemann-Hilbert problem for $\Phi$, we first study some of the symmetries of the equation.

\subsection{Symmetry of $\Phi(\xi)$.} \label{Phi-symmetry-section}
\begin{prop}\label{N-sym}
    Define the matrix
    \begin{equation}
        \mathcal{S} := 
            \begin{psmallmatrix}
                0 & 1 & 0\\
                0 & 0 & 1\\
                1 & 0 & 0
            \end{psmallmatrix}.
    \end{equation}
Then, $\mathcal{L}(\xi;t_{5},t_{2},x)$ satisfies the symmetry condition
    \begin{equation}
        \mathcal{L}(\xi;t_{5},t_{2},x) = \omega \mathcal{S}^T \mathcal{L}(\omega\xi;t_{5},t_{2},x)\mathcal{S}.
    \end{equation}
\end{prop}
\begin{proof}
    The proof of this fact follows almost immediately from the fact that $g(\omega\xi(\lambda)) = g(\xi(\lambda)) \mathcal{S}^T$. By the definition of $\mathcal{L}(\xi)$, we have that:
        \begin{align*}
            \omega \mathcal{L}(\omega\xi) &= 3\xi^2 \left[g^{-1}(\omega\xi)L(\xi^3)g(\omega\xi) - g^{-1}(\omega\xi) \frac{dg}{d\lambda}(\omega\xi)\right]\\
                &= 3\xi^2 \mathcal{S}\left[g^{-1}(\xi)L(\xi^3)g(\xi) - g^{-1}(\xi) \frac{dg}{d\lambda}(\xi)\right]\mathcal{S}^T\\
                &= \mathcal{S}\mathcal{L}(\xi;t_{5},t_{2},x) \mathcal{S}^{T};
        \end{align*}
    multiplication on the left by $\mathcal{S}^T$ and the right by $\mathcal{S}$ yields the result.
\end{proof}
As an immediate corollary,
\begin{cor}
    The formal expansion $\Phi(\xi;t_{5},t_{2},x)$ satisfies the symmetry condition
    \begin{equation}
        \Phi(\xi;t_{5},t_{2},x) = \mathcal{S}^T\Phi(\omega\xi;t_{5},t_{2},x)\mathcal{S}.
    \end{equation}
    Furthermore, the coefficients $\Phi_k$ of any asymptotic solution satisfy
    \begin{equation} \label{coeff-symmetry}
        \Phi_k = \omega^{-k} \mathcal{S}^T \Phi_k \mathcal{S}.
    \end{equation}
\end{cor}

\subsection{Proof of Proposition \eqref{PropA}.}
In this subsection, we prove a version of Proposition \eqref{PropA}. What is contained here is the direct analog of Proposition 5.6, 5.7 and Theorem 5.3 of \cite{FIKN} for the Painlev\'{e} I system. Before formulating the Proposition, we state a technical lemma:
    \begin{lemma}\label{xi-comparison-lemma}
        Consider the functions $\vartheta_j(\xi) = \vartheta_j(\xi;t_{5},t_{2},x)$ defined by Equation \eqref{xi-exponents}, and fix $\epsilon>0$. For any $t_{5},t_{2},x$ in some fixed compact set $K \subset \CC^3$, there exists a constant $M = M_K$ such that, for all $|\xi| > M_K$, and for any $\ell \in \ZZ$,
            \begin{align}
                \text{Re } \vartheta_1(\xi) &< \text{Re } \vartheta_2(\xi), \qquad \qquad \frac{\pi}{21}(6\ell + 2)+\epsilon < \arg\xi < \frac{\pi}{21}(6\ell + 5)-\epsilon,\\
                \text{Re } \vartheta_2(\xi) &< \text{Re } \vartheta_3(\xi), \qquad \qquad \frac{\pi}{21}(6\ell)+\epsilon < \arg\xi < \frac{\pi}{21}(6\ell + 3)-\epsilon,\\
                \text{Re } \vartheta_1(\xi) &< \text{Re } \vartheta_3(\xi), \qquad \qquad \frac{\pi}{21}(6\ell+1)+\epsilon < \arg \xi < \frac{\pi}{21}(6\ell+4)-\epsilon.
            \end{align}
    \end{lemma}
\begin{proof}
   The lemma follows from straightforward calculation; one has that
        \begin{equation*}
            \frac{1}{|\xi|^7}\text{Re } \vartheta_j(\xi) = \frac{3}{7}\cos\left[7\arg\xi + \frac{2\pi}{3}(j-1)\right] \left(1 + \OO(|\xi|^{-2})\right),
        \end{equation*}
    and so it is clear that for $|\xi|$ taken to be sufficiently large, $\text{Re } \vartheta_j(\xi)$ is dominated by the first term. Comparison of the values of $\cos(7\theta + \frac{2\pi}{3}(j-1))$ functions for different values of the argument $\theta$ yields the result.
\end{proof}
We can now formulate and prove the following Proposition, which is a more precise statement of \eqref{PropA}.

\begin{prop}
    Let $U(t_{5},t_{2},t_1), V(t_{5},t_{2},t_1)$ solve the string equation \eqref{string-equation}, \eqref{U_mu}--\eqref{V_eta}, and assume $(t_5,t_2,t_1)$ is not a singular point of $U,V$. Let
    $\Phi^{(k)}(\xi;t_{5},t_{2},t_1)$ be solutions to the linear ODE \eqref{Phi-ODE}, which are uniquely determined by the condition that
        \begin{equation}\label{sectorial-phi}
            \Phi^{(k)}(\xi;t_{5},t_{2},x) = \left[\mathbb{I} + \OO(\xi^{-1}) \right] e^{\Theta(\xi;t_{5},t_{2},x)}, \qquad\qquad \xi \to \infty, \qquad \xi\in \Omega_k,
        \end{equation}
where the open sectors $\Omega_k$ are defined to be
        \begin{equation}\label{sector-def}
            \Omega_k := \left\{\xi\in \CC : \frac{\pi}{21}(k-2) < \arg \xi < \frac{\pi}{21}(k+1)\right\},\qquad\qquad k = 1,...,42.
        \end{equation}
The functions $\Phi^{(k)}$ are related by
        \begin{equation*}
            \Phi^{(k+1)}(\xi;t_{5},t_{2},x) = \Phi^{(k)}(\xi;t_{5},t_{2},x)S_k, \qquad k = 1,...,41, \qquad\qquad \Phi^{(1)}(\xi;t_{5},t_{2},x) = \Phi^{(42)}(e^{2\pi i}\xi;t_{5},t_{2},x) S_{42}
        \end{equation*}
where the matrices $S_k$ have the form
    \begin{equation} \label{Stokes-structure}
        S_{k} =  \mathbb{I} +
            s_k\begin{cases}
                E_{32}, & k \equiv 0 \mod 6,\\
                E_{31}, & k \equiv 1 \mod 6,\\
                E_{21}, & k \equiv 2 \mod 6,\\
                E_{23}, & k \equiv 3 \mod 6,\\
                E_{13}. & k \equiv 4 \mod 6,\\
                E_{12}, & k \equiv 5 \mod 6.
            \end{cases}
    \end{equation}
    Furthermore, the $S_k$ satisfy the identities
        \begin{equation} \label{Stokes-equations}
            S_{k+14} = \mathcal{S}^T S_k \mathcal{S},\qquad\qquad S_1\cdots S_{14} = \mathcal{S}^T.
        \end{equation}
    In particular it follows from the above that, $s_{k+14} = s_k$, and that generically there are only $6$ independent Stokes parameters. Furthermore, denote $\Phi^{(0)}(\xi;t_{5},t_{2},x)$ 
    to be the solution of \eqref{Phi-ODE} near $\xi = 0$, normalized as follows:
        \begin{equation}\label{connection-A}
            \Phi^{(0)}(\xi;t_{5},t_{2},x) = \left[\mathcal{U}^{-1} + \OO(\xi)\right]\xi^{-\Delta}.
        \end{equation}
    The functions $\Phi^{(1)}(\xi;t_{5},t_{2},x)$ and $\Phi^{(0)}(\xi;t_{5},t_{2},x)$ are related by the unimodular constant matrix $\mathcal{C}$:
        \begin{equation}\label{connection-B}
            \Phi^{(1)}(\xi;t_{5},t_{2},x) = \Phi^{(0)}(\xi;t_{5},t_{2},x)\mathcal{C}, \qquad\qquad \det \mathcal{C} = 1.
        \end{equation}
    The equations \eqref{connection-A},\eqref{connection-B} imply that $\mathcal{C}$ has three free parameters. Thus, the string equation \eqref{string-equation}, \eqref{U_mu}--\eqref{V_eta}
    are associated with $6 + 3 = 9$ constant monodromy data.
\end{prop}
\begin{proof}
    Standard ODE theory \cite{Wasow,FIKN,JMU1} establishes that the functions $\Phi^{(k)}(\xi;t_{5},t_{2},x)$ are indeed uniquely specified by the asymptotic condition 
    \eqref{sectorial-phi}. The structure of the Stokes matrices $S_k$ can be inferred as follows. Note that $\Phi^{(k)}$, $\Phi^{(k+1)}$ are both defined on the 
    sector
        \begin{equation*}
            \delta \Omega_k := \Omega_k \cap \Omega_{k+1} = \left\{ \xi\in\CC : \frac{\pi}{21}(k-2) < \arg \xi < \frac{\pi}{21}(k+1)\right\}. 
        \end{equation*}
    Since both $\Phi^{(k)}$, $\Phi^{(k+1)}$ satisfy equation \eqref{Phi-ODE}, their ratio $\left(\Phi^{(k)}\right)^{-1}\Phi^{(k+1)} =: S_k$ is a constant matrix. We have that
        \begin{equation*}
            S_k = \lim_{\substack{\xi\to \infty\\ \xi\in \delta \Omega_k}} \left(\Phi^{(k)}\right)^{-1}\Phi^{(k+1)} = \lim_{\substack{\xi\to \infty\\ \xi\in \delta \Omega_k}} e^{-\Theta(\xi)}\left[\mathbb{I} + \OO(\xi^{-1})\right] e^{\Theta(\xi)}.
        \end{equation*}
    Equivalently, component-wise,
        \begin{equation*}
            (S_k)_{ij} = \lim_{\substack{\xi\to \infty\\ \xi\in \delta \Omega_k}} e^{\vartheta_j(\xi) - \vartheta_i(\xi)}\left[\delta_{ij} + \OO(\xi^{-1})\right],
        \end{equation*}
    where the functions $\vartheta_j(\xi) = \vartheta_j(\xi;t_{5},t_{2},x)$ are defined by Equation \eqref{xi-exponents}. The above formula implies immediately that
    the diagonal components of $S_k$ are all identically $1$. The all but one of the remaining entries can be determined by taking the above limit in 
    various parts of the sector. We furnish the proof here for the case $k = 6\ell  + 1$; the structure of the Stokes matrices in the other cases may be obtained in an identical manner.

    If $k = 6\ell+1$, consider first the sector $\left\{\frac{\pi}{21}(6\ell-1) < \arg\xi < \frac{\pi}{21}(6\ell)\right\} \subset \delta \Omega_k$. Using Lemma \ref{xi-comparison-lemma}, we see that there is a definite ordering of $\text{Re } \vartheta_j(\xi)$ in this sector for
    $|\xi|$ sufficiently large, given by
            $\text{Re } \vartheta_3(\xi) <  \text{Re } \vartheta_1(\xi) <  \text{Re } \vartheta_2(\xi)$.
    This implies that $\left(S_{k}\right)_{21} = \left(S_{k}\right)_{23} = \left(S_{k}\right)_{13} = 0$. On the other hand, in the
    sector $\left\{\frac{\pi}{21}(6\ell) < \arg\xi < \frac{\pi}{21}(6\ell+1)\right\} \subset \delta \Omega_k$, for $|\xi|$ sufficiently large 
    the ordering
            $\text{Re } \vartheta_3(\xi) <  \text{Re } \vartheta_2(\xi) <  \text{Re } \vartheta_1(\xi)$
    holds, and so we find that $\left(S_k\right)_{12} = 0$ in addition. Finally, in the sector $\left\{\frac{\pi}{21}(6\ell+1) < \arg\xi < \frac{\pi}{21}(6\ell+2)\right\} \subset \delta \Omega_k$, for $|\xi|$ sufficiently large the ordering
            $\text{Re } \vartheta_2(\xi) <  \text{Re } \vartheta_3(\xi) <  \text{Re } \vartheta_1(\xi)$
    holds, and so we see that $(S_k)_{32} = 0$. The only entry which cannot be determined by the above line of argumentation is $(S_k)_{31}$; thus, the 
    two solutions are related by
        \begin{equation*}
            \Phi^{(6\ell+2)} = \Phi^{(6\ell+1)} \left[\mathbb{I} + s_{6j+1}E_{31}\right].
        \end{equation*}
    Now, Proposition \ref{N-sym} implies that, if $\Phi(\xi;t_{5},t_{2},x)$ is a solution to the linearization equations, 
    then so is $\mathcal{S}^{T}\Phi(\omega \xi;t_{5},t_{2},x) \mathcal{S}$. Since $\xi \in \Omega_k$ implies that $\omega\xi \in \Omega_{k+14}$,
    we obtain the relations
        \begin{equation*}
            \Phi^{(k)}(\xi;t_{5},t_{2},x) = \mathcal{S}^T \Phi^{(k+14)}(\omega\xi;t_{5},t_{2},x)\mathcal{S}\qquad
            \Longleftrightarrow \qquad \mathcal{S}\Phi^{(k)}(\omega^2\xi;t_{5},t_{2},x)\mathcal{S}^T = \Phi^{(k+14)}(\xi;t_{5},t_{2},x).
        \end{equation*}
    This implies the relation $\mathcal{S}^T S_k \mathcal{S} = S_{k+14}$; one can further check that this is consistent with the formula \eqref{Stokes-structure} for the Stokes matrices, i.e. the only relation that this implies is the following one among the parameters: $s_{k + 14} = s_{k}$.
    Furthermore, we have that
        \begin{equation*}
            \mathcal{S} \Phi^{(1)}(\xi;t_{5},t_{2},x) \mathcal{S}^T = \Phi^{(15)}(\omega\xi;t_{5},t_{2},x) = \Phi^{(1)}(\omega\xi;t_{5},t_{2},x)S_1\cdots S_{14},
        \end{equation*}
    which implies the following identity for the solution $\Psi^{(0)}(\xi;t_{5},t_{2},x)$ in a neighborhood of $\xi = 0$:
        \begin{equation*}
            \mathcal{S} \Phi^{(0)}(\xi;t_{5},t_{2},x) \mathcal{C} \mathcal{S}^T = \Phi^{(0)}(\omega\xi;t_{5},t_{2},x) \mathcal{C} S_1\cdots S_{14}
        \end{equation*}
    Using Equation \eqref{phi-origin-solution}, we further see that $\Phi^{(0)}(\omega\xi;t_{5},t_{2},x) = \mathcal{S} \Phi^{(0)}(\xi;t_{5},t_{2},x)$, and so
        \begin{equation*}
            \mathcal{S} \Phi^{(0)}(\xi;t_{5},t_{2},x) \mathcal{C} \mathcal{S}^T = \Phi^{(0)}(\omega\xi;t_{5},t_{2},x) \mathcal{C} S_1\cdots S_{14} = \mathcal{S}\Phi^{(0)}(\xi;t_{5},t_{2},x) \mathcal{C} S_1\cdots S_{14}.
        \end{equation*}
    Since $\mathcal{S} \Phi^{(0)}(\xi;t_{5},t_{2},x) \mathcal{C}$ is invertible, we obtain the identity $\mathcal{S}^T = S_1\cdots S_{14}$.
    Equations \eqref{Stokes-equations} imply that there are only $6$ independent Stokes parameters.

    Now, let us show that the matrix $\mathcal{C}$ depends only on $3$ independent parameters. Suppose $\Phi^{(0)}$, $\tilde{\Phi}^{(0)}$ are two different solutions in a neighborhood of $\xi = 0$
    which connect to $\Phi^{(1)}$ through the matrices $\mathcal{C}$, $\tilde{\mathcal{C}}$. In other words,
        \begin{equation*}
            \Phi^{(1)} = \Phi^{(0)}\mathcal{C} = \tilde{\Phi}^{(0)}\tilde{\mathcal{C}}.
        \end{equation*}
    Now, the functions $\Phi^{(0)}(\xi)\xi^{\Delta}$, $\tilde{\Phi}^{(0)}(\xi)\xi^{\Delta}$ are holomorphic and invertible in a neighborhood of zero, and so it follows that the matrix
        \begin{equation*}
            \mathcal{J}(\xi) := \xi^{-\Delta}\mathcal{C}\tilde{\mathcal{C}}^{-1}\xi^{\Delta}
        \end{equation*}
    must be holomorphic and invertible as well. This places constraints on the matrix $K := \mathcal{C}\tilde{\mathcal{C}}^{-1}$; we find that
        \begin{equation*}
            K = 
            \begin{psmallmatrix}
                k_{11} & 0 & 0 \\
                k_{21} & k_{22} & 0 \\
                k_{31} & k_{32} & k_{33}
            \end{psmallmatrix},
        \end{equation*}
    where the diagonal elements are subject to the constraint $k_{11}k_{22}k_{33} = 1$. Furthermore, one can utilize the gauge freedom (cf. Remark \ref{gauge-remark}, point 2.) to further eliminate two
    of the parameters below the diagonal. This leaves three free parameters.
\end{proof}

\begin{remark}
    One can show that the generic solution to the constraint equations \eqref{Stokes-equations} is given by
    \begin{align}
        s_{7} &= \frac{s_{1} s_{3} s_{6}-s_{2} s_{6}+s_{1}+1}{W_1},\qquad s_{8} = \frac{-s_{1} s_{3} W_{2}+s_{2} W_{2}+W_{1}-s_{2} s_{4}-s_{3}}{W_{1}W_{2}}, \nonumber\\
        s_{9} &= \frac{W_{1} - s_{2}s_{4} - s_{3} }{W_{2}}, \qquad s_{10} = - W_{1}, \qquad s_{11} = -W_{2},\qquad
        s_{12} = \frac{-W_{2}-s_{3}s_{5} + s_{4}}{W_{1}},\\s_{13} &= \frac{-s_{4}s_{6}W_{1} - s_{5}W_{1} + W_{2} + s_{3}s_{5}-s_{4}}{W_{1}W_{2}} \qquad s_{14} = \frac{-s_{1} s_{4} s_{6}-s_{1} s_{5}-s_{6}+1}{W_2},\nonumber
    \end{align}
with $s_{1},...,s_{6}$ free parameters, if 
    \begin{equation}
        W_1 := s_{1} s_{3} s_{5}-s_{1} s_{4}-s_{2} s_{5}-1 \neq 0 \qquad \text{ and }\qquad W_2 := s_{2} s_{4} s_{6}+s_{2} s_{5}+s_{3} s_{6}+1 \neq 0.
    \end{equation}
There are many subcases if $W_1$ or $W_2$ vanish; we shall save the study of these for later.
We now see that there are generically $6$ free Stokes parameters, which is consistent with the fact that the string equation \eqref{string-equation} is of order $4 + 2 = 6$. 

On the other hand, the space of solutions to the Stokes equation contains $7$ special planes (the corresponding Stokes parameters satisfy $s_k = -s_{k+8}$, $k=-7,...,-1$, and we parameterize in terms of the remaining $7$ Stokes parameters $(s_1,....,s_7)$:
        \begin{align*}
            \Pi_0 &:= \{(x+1,-1,0,0,1,x,y)\},\\
            \Pi_1 &:= \{(0,-1,x,y,1-x,-1,0)\},\\
            \Pi_2 &:= \{(x,y-1,1,0,0,-1,y)\},\\
            \Pi_3 &:= \{(0,0,1,x,y,-x-1,1)\},\\
            \Pi_4 &:= \{(x,y,1-x,-1,0,0,1)\},\\
            \Pi_5 &:= \{(1,0,0,-1,1-y,x,y)\},\\
            \Pi_6 &:= \{(1,-1-y,x,y,1,0,0)\}.
        \end{align*}
        These planes intersect pairwise: $\Pi_k \cap \Pi_{k\pm 1} \neq \emptyset$ for $k\in \ZZ_7$, and otherwise are completely disjoint. As these planes are embedded in $\CC^7$, their intersection is a point in each case. The solution corresponding to one of these special intersection points appears in 
        the multicritical quartic $2$-matrix model \cite{DHL1,DHL3,Nathan2}.
\end{remark}

We now state the ``converse'' to the above: we formulate a Riemann-Hilbert problem associated to the string equation.

\begin{prop} \label{RHP-Phi-prop}
    Let $\{S_k\}_{k=1}^{42}$ be the constant $3\times 3$ matrices defined by \eqref{Stokes-structure}, satisfying relations \eqref{Stokes-equations}. Furthermore,
    let $\mathcal{C}$ be a constant $3\times 3$ matrix satisfying
        \begin{equation*}
            \det \mathcal{C} = 1, \qquad\qquad K\mathcal{C} \sim \mathcal{C},
        \end{equation*}
    where $K$ is any unimodular lower-triangular matrix as prescribed by the previous proposition, with $\sim$ denoting similarity equivalence. For 
    $\xi,t_{5},t_{2},x \in \CC$, define the $3\times 3$ sectionally analytic function $X(\xi;t_{5},t_{2},x)$ as follows:
        \begin{equation}
            X(\xi;t_{5},t_{2},x) :=
                \begin{cases}
                    X^{(0)}(\xi;t_{5},t_{2},x), & |\lambda| < 1,\\
                    X^{(k)}(\xi;t_{5},t_{2},x), & \lambda \in \Omega_k \cap \{ |\lambda| > 1\},\qquad k = 1,...,42,
                \end{cases}
        \end{equation}
    where the sectors $\Omega_k$ are defined as in \eqref{sector-def}. Finally, let $X(\xi;t_{5},t_{2},x)$ solve the following Riemann-Hilbert problem:
        \begin{align}
            X^{(k+1)}(\xi;t_{5},t_{2},x) &= X^{(k)}(\xi;t_{5},t_{2},x)e^{\Theta(\xi;t_{5},t_{2},x)}S_k e^{-\Theta(\xi;t_{5},t_{2},x)}, \qquad k = 1,...,42, \qquad X_{43} = X_1. \nonumber\\
            X^{(1)}(\xi;t_{5},t_{2},x) &= X_0(\xi;t_{5},t_{2},x)e^{\Theta(\xi;t_{5},t_{2},x)}\xi^{-\Delta} \mathcal{C} e^{-\Theta(\xi;t_{5},t_{2},x)},\\
            X^{(1)}(\xi;t_{5},t_{2},x) &= \mathbb{I} + \OO(\xi^{-1}), \qquad \xi \to\infty, \nonumber
        \end{align}
    where $\Theta(\xi;t_{5},t_{2},x)$, $\Delta$ are as previously defined. Then, the above Riemann-Hilbert problem defined a unique matrix $X(\lambda;t_{5},t_{2},x)$ which is meromorphic in $t_{5},t_{2},x$. Furthermore, if we denote
        \begin{equation}
            X^{(1)}(\xi;t_{5},t_{2},x) = \mathbb{I} + \frac{X^{(1)}_1(t_{5},t_{2},x)}{\xi} + \frac{X^{(1)}_2(t_{5},t_{2},x)}{\xi^2} + \OO(\xi^{-3}),
        \end{equation}
    then
        \begin{align*}
            U(t_{5},t_{2},x) &:= 2\frac{d}{dx}\left[X^{(1)}_1(t_{5},t_{2},x)\right]_{11},\\
            V(t_{5},t_{2},x) &:= -2\frac{d}{dx}\left[X^{(1)}_2(t_{5},t_{2},x) - \frac{1}{2}X^{(1)}_1(t_{5},t_{2},x)^2\right]_{11} = -\frac{d}{dt_2}\left[X^{(1)}_1(t_{5},t_{2},x)\right]_{11},
        \end{align*}
    Then $U,V$ are meromorphic in $t_{5},t_{2},x$, and satisfy the string equation \eqref{string-equation}, \eqref{U_mu}--\eqref{V_eta}.
\end{prop}
    \begin{proof}
        Uniqueness of the solution to this problem follows from the usual Liouville argument.
        Observe that if we set
        \begin{align*}
            \Phi^{(k)} &= X^{(k)}e^{\Theta}, \qquad k = 1,...,42, \qquad\qquad\Phi^{(0)} = X^{(0)}e^{\Theta}\xi^{-\Delta},
        \end{align*}
        Then the functions $\Phi^{(k)}$ satisfy the relations \eqref{Stokes-structure}, \eqref{connection-B}. By construction, we can find a system of contours
        $\Gamma$ such that the jump matrix of the above Riemann-Hilbert problem is smooth, and decays exponentially for $\xi \to \infty$ (for example,
        one may take the unit circle unioned with the rays $\{\arg\xi = \frac{\pi}{21}k\}_{k=1}^{42} \cap \{|\xi| > 1\}$). By standard Riemann-Hilbert arguments
        (cf. \cite{FIKN}), we obtain that the solution to this RHP exists, and depends meromorphically on its parameters $t_{5},t_{2},x$. Our next task is to extract
        the string equation from the isomonodromy/zero-curvature conditions. Write the asymptotic expansion for $\Phi(\xi;t_{5},t_{2},x)$ as
            \begin{equation*}
                \Phi(\xi;t_{5},t_{2},x) = \left(\mathbb{I} + \sum_{k=1}^{\infty} \frac{\Phi_k(t_{5},t_{2},x)}{\xi^k}\right)e^{\Theta(\xi;t_{5},t_{2},x)}.
            \end{equation*}
        
 In general, we have the following procedure for determining the entries of the matrices $\Phi_k$:
    \begin{enumerate}
        \item First, observe that we only have to determine the first row of $\Phi_k$:
            \begin{equation*}
                [\Phi_k]_{1,\cdot} := [a_k(t_{5},t_{2},x),b_k(t_{5},t_{2},x),c_k(t_{5},t_{2},x)].
            \end{equation*}
        The rest of the entries are determined by the symmetry constraint \eqref{coeff-symmetry}.
        \item Using the formal expansion of $\Phi$, form the series
            \begin{equation*}
                \frac{d \Phi}{d\xi} \Phi^{-1} = \mathcal{L}(\xi;t_{5},t_{2},x) + \sum_{k=2}^{\infty} \frac{R_k(t_{5},t_{2},x)}{\xi^k};
            \end{equation*}
        where we use our previous expression for $\mathcal{L}$ \eqref{L-zeta} to parameterize the entries of the above.
        The condition that $\Phi(\xi)$ satisfies the differential equation $\frac{d \Phi}{d\xi} = \mathcal{L} \Phi$ determines the
        coefficients $\{b_k ,c_k\}_{k=1}^7$ as differential polynomials in the variables $U(t_{5},t_{2},x)$, $V(t_{5},t_{2},x)$, and (for the $k^{th}$ function, $k\geq 3$) the functions $\{a_j\}_{j=1}^{k-2}$; it also imposes the constraint 
            \begin{equation*}
                R_k(t_{5},t_{2},x) \equiv 0, \qquad\qquad k=2,3, ...
            \end{equation*}
        \item The condition that the coefficients $R_k$ vanish identically allows us to solve for the rest of the variables. More precisely, for
        $k=2,3,...$, we have that
            \begin{enumerate}
                \item $[R_k]_{11}$ can be solved for $a_{k-1}$ as a differential polynomial in $U(t_{5},t_{2},x)$, $V(t_{5},t_{2},x)$,
                \item $[R_k]_{12}$ can be solved for $b_{k+6}$ as a differential polynomial in $U(t_{5},t_{2},x)$, $V(t_{5},t_{2},x)$, and the functions $\{a_j\}_{j=1}^{k-2}$,
                \item $[R_k]_{13}$ can be solved for $c_{k+6}$ as a differential polynomial in $U(t_{5},t_{2},x)$, $V(t_{5},t_{2},x)$, and the functions $\{a_j\}_{j=1}^{k-2}$.
            \end{enumerate}
        the symmetry constraint \eqref{coeff-symmetry} implies that solving the above three equations makes $R_k \equiv 0$.
    \end{enumerate}
In particular, we obtain that
    \begin{equation}\label{first-coeff}
        \Phi_1 =
        \begin{psmallmatrix}
            -\frac{1}{2}H_{1} & 0 & 0\\
            0 & -\frac{\omega^2}{2} H_{1} & 0\\
            0 & 0 & -\frac{\omega}{2} H_{1}
        \end{psmallmatrix}
    \end{equation}
where $\frac{d}{dx} H_{1} = -U$ is the Hamiltonian for the $x$-variable. Also,
    \begin{equation}\label{second-coeff}
        \Phi_2 =
        \begin{psmallmatrix}
            \frac{1}{8}(H_1)^2-\frac{1}{4}H_2 & -\frac{i\omega^2\sqrt{3}}{12} U & \frac{i\omega\sqrt{3}}{12} U \\
            \frac{i\omega^2\sqrt{3}}{12} U& \omega \left(\frac{1}{8}(H_1)^2-\frac{1}{4}H_2\right) & -\frac{i\sqrt{3}}{12} U\\
            -\frac{i\omega\sqrt{3}}{12}  U &  \frac{i\sqrt{3}}{12} U& \omega^2 \left(\frac{1}{8}(H_1)^2-\frac{1}{4}H_2\right)
        \end{psmallmatrix},
    \end{equation}
where $\frac{d}{dx} H_{2} = 2V$ is the Hamiltonian for the $t_{2}$-variable.

Direct calculation then shows that:
    \begin{enumerate}
        \item The zero-curvature equation between the $\xi,x$ variables is equivalent to the string equation \eqref{string-equation};
        \item The zero-curvature equation between the $\xi,t_{2}$ variables is equivalent to the equations \eqref{U_mu}, \eqref{V_mu}, modulo the string equation\footnote{By ``modulo the string equation'' we mean that we must use the string equation to replace higher order derivatives of the functions $U,V$ in the variable $x$.},
        \item The zero-curvature equation between the $\xi,t_{5}$ variables is equivalent to the equations \eqref{U_eta}, \eqref{V_eta}, modulo the string equation,
        \item Modulo the string equation \eqref{string-equation}, \eqref{U_mu}--\eqref{V_eta}, the other zero-curvature equations between $(t_{5},t_{2}),(t_{2},x)$, and $(t_{5},x)$, vanish identically. In other words, these equations result in no new differential conditions on the functions $U,V$.
    \end{enumerate}
    \end{proof}
\begin{remark}\label{a-calculations}
    We list the 1-1 entries of the first few matrices $\Phi_k(t_{5},t_{2},x)$ here, for the convenience of the reader.
        \begin{align*}
            [\Phi_1]_{11} &= -\frac{1}{2}H_{1}, \\ 
            [\Phi_2]_{11} &= \frac{1}{8}(H_{1})^2-\frac{1}{4}H_{2},\\
            [\Phi_3]_{11} &= -\frac{1}{48}(H_{1})^3 + \frac{1}{2} H_{1}H_{2},\\
            [\Phi_4]_{11} &= \frac{1}{384}(H_{1})^4 + \frac{1}{32}(H_{1})^2H_{2} + \frac{1}{32}(H_{2})^2 -\frac{5}{24}t_{5} H_{2} + \frac{1}{12}V' - \frac{1}{96} U^2 + \frac{1}{6}t_{2} x ,\\
            [\Phi_5]_{11} &= -\frac{1}{38400}(H_{1})^5 + \frac{1}{192}(H_{1})^3 H_{2} -\frac{1}{64} H_{1}(H_{2})^2 
        - \frac{1}{24} \left(V' - \frac{1}{8} U^2 \right)H_{1}\\
        &-\frac{5}{48}t_{5} H_{1}H_{2} - \frac{1}{12}t_{2} x H_{1} +\frac{1}{90}U'''- \frac{1}{16} U U' -\frac{1}{24} UV - \frac{1}{10}H_{5},\\
            [\Phi_6]_{11} &= \frac{1}{46080}H_1^6-\frac{1}{1536}H_1^4H_2 + \frac{1}{256}H_1^2H_2^2 - \frac{1}{384}H_2^3 + \frac{1}{20}H_5H_1 +\frac{5}{192}t_5H_2H_1^2\\
            &-\frac{1}{768}\left(U^2-8V'-16t_2x\right)H_1^2 -\frac{5}{96}t_5H_2^2 + \frac{1}{180}\left(\frac{15}{4}UV + \frac{45}{8}UU'-U'''\right)H_1\\
            &+\frac{1}{384}\left(U^2-8V'-16t_2x\right)H_2 + \frac{5}{192}U^3 - \frac{1}{72}U''U + \frac{1}{12}V^2 + \frac{1}{144}(U')^2,\\
        \end{align*}
\end{remark}
As a final result of this subsection, we state without proof the equivalent Riemann-Hilbert formulation in the $\lambda$-plane.

\begin{figure}
    \begin{center}
    \begin{overpic}[scale=.5]{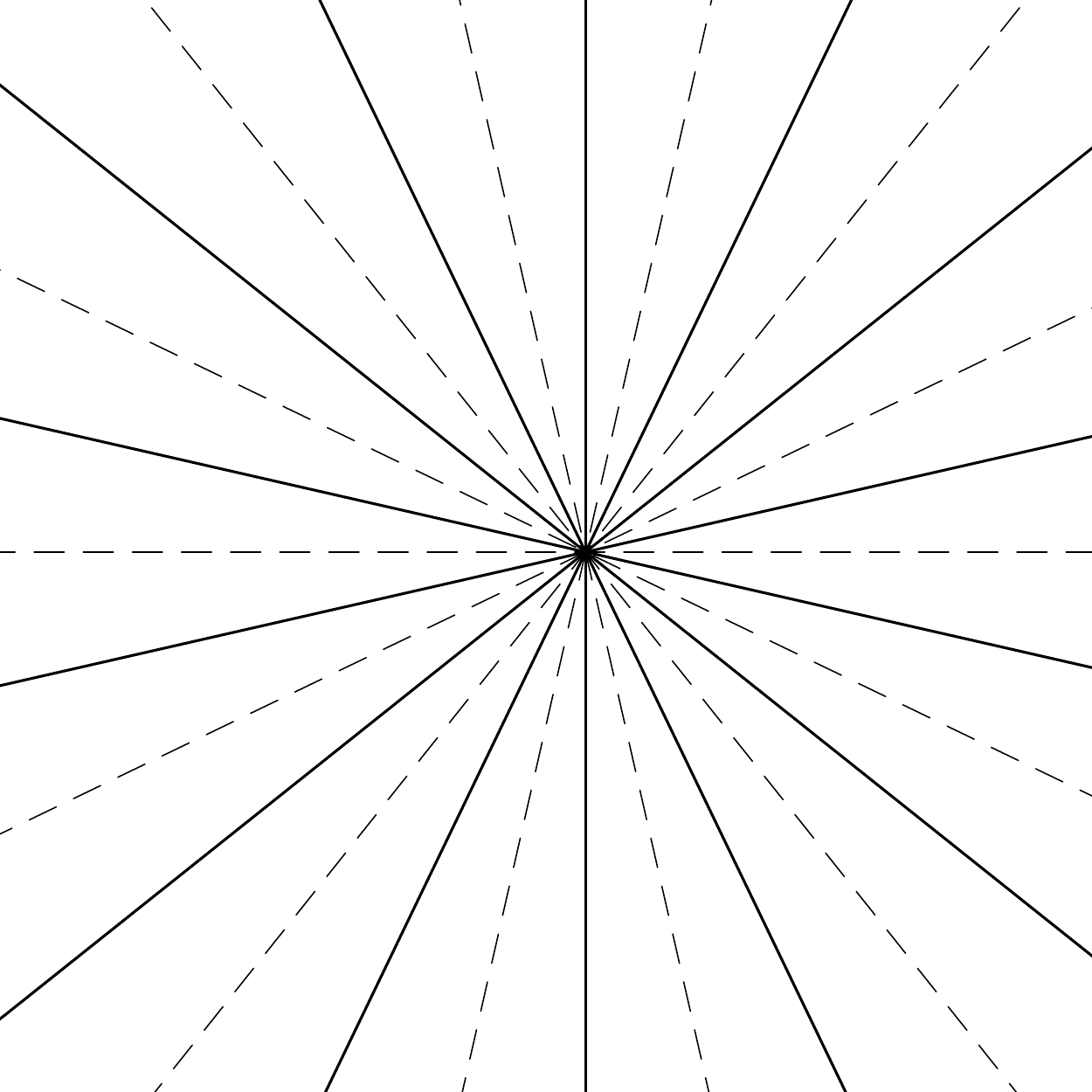}
                \put (100,50) {$\text{Re } \lambda$}
          	\put (103,60) {$\gamma_1$}
      		  \put (103,37) {$\gamma_{-1}$}
     		  \put (98,82) {$\gamma_2$}
      		  \put (98,15) {$\gamma_{-2}$}
                \put (78,95) {$\gamma_3$}
      		  \put (78,3) {$\gamma_{-3}$}
                \put (55,99) {$\gamma_4$}
      		  \put (55,1) {$\gamma_{-4}$}
                \put (34,98) {$\gamma_5$}
      		  \put (34,1) {$\gamma_{-5}$}
                \put (5,92) {$\gamma_6$}
      		  \put (5,6) {$\gamma_{-6}$}
                \put (-5,61) {$\gamma_7$}
      		  \put (-5,36) {$\gamma_{-7}$}
                \put (-4,48) {$\rho$}
                \put (80,57){\small \rotatebox{15}{\textcolor{BrickRed}{$\mathbb{I} + s_1E_{21}$}} } 
                \put (80,44){\small \rotatebox{-11}{\textcolor{BrickRed}{$\mathbb{I} + s_{-1}E_{31}$}} } 
                \put (76,70){\small \rotatebox{42}{\textcolor{BrickRed}{$\mathbb{I} + s_2E_{23}$}}  } 
                \put (79,30){\small \rotatebox{-36}{\textcolor{BrickRed}{$\mathbb{I} + s_{-2}E_{32}$}}  } 
                \put (66,83){\small \rotatebox{65}{\textcolor{BrickRed}{$\mathbb{I} + s_{3}E_{13}$}}  } 
                \put (69,19){\small \rotatebox{-61}{\textcolor{BrickRed}{$\mathbb{I} + s_{-3}E_{12}$}}  } 
                \put (50,85){\small \rotatebox{90}{\textcolor{BrickRed}{$\mathbb{I} + s_{4}E_{12}$}}  } 
                \put (54,20){\small \rotatebox{-90}{\textcolor{BrickRed}{$\mathbb{I} + s_{-4}E_{13}$}}  } 
                \put (35,90){\small \rotatebox{-63}{\textcolor{BrickRed}{$\mathbb{I} + s_{5}E_{32}$}}  } 
                \put (30,8){\small \rotatebox{65}{\textcolor{BrickRed}{$\mathbb{I} + s_{-5}E_{23}$}}  } 
                \put (18,80){\small \rotatebox{-38}{\textcolor{BrickRed}{$\mathbb{I} + s_{6}E_{31}$}}  } 
                \put (12,19){\small \rotatebox{42}{\textcolor{BrickRed}{$\mathbb{I} + s_{-6}E_{21}$}}  } 
                \put (6,61){\small \rotatebox{-12}{\textcolor{BrickRed}{$\mathbb{I} + s_{7}E_{21}$}}  } 
                \put (6,40){\small \rotatebox{12}{\textcolor{BrickRed}{$\mathbb{I} + s_{-7}E_{31}$}}  } 
                \put (6,50){\small \rotatebox{0}{\textcolor{BrickRed}{$\mathcal{S}$}}  } 
    		\end{overpic}
    \end{center}
    \caption{The Stokes lines $\gamma_j$, for the Riemann-Hilbert problem for $\Psi(\lambda)$. Each of the Stokes sectors is bisected by an anti-Stokes line, depicted by a dashed line. All contours are oriented {\bf \textit{outwards}} from the origin. Note that we have also labelled the anti-Stokes line $(-\infty,0]$ by $\rho$. The Stokes matrix $S_k$ is the matrix associated to the parameter $s_k$; these parameters are not all independent, and must satisfy the equation
    $S_{-7}\cdots S_{-1}S_{1}\cdots S_{7} = \mathcal{S}^T$.
    The equivalent diagram for the $\Phi(\xi)$-Riemann-Hilbert problem in the $\xi$ plane has $42$ rays. }
    \label{fig:StokesRays}
\end{figure}
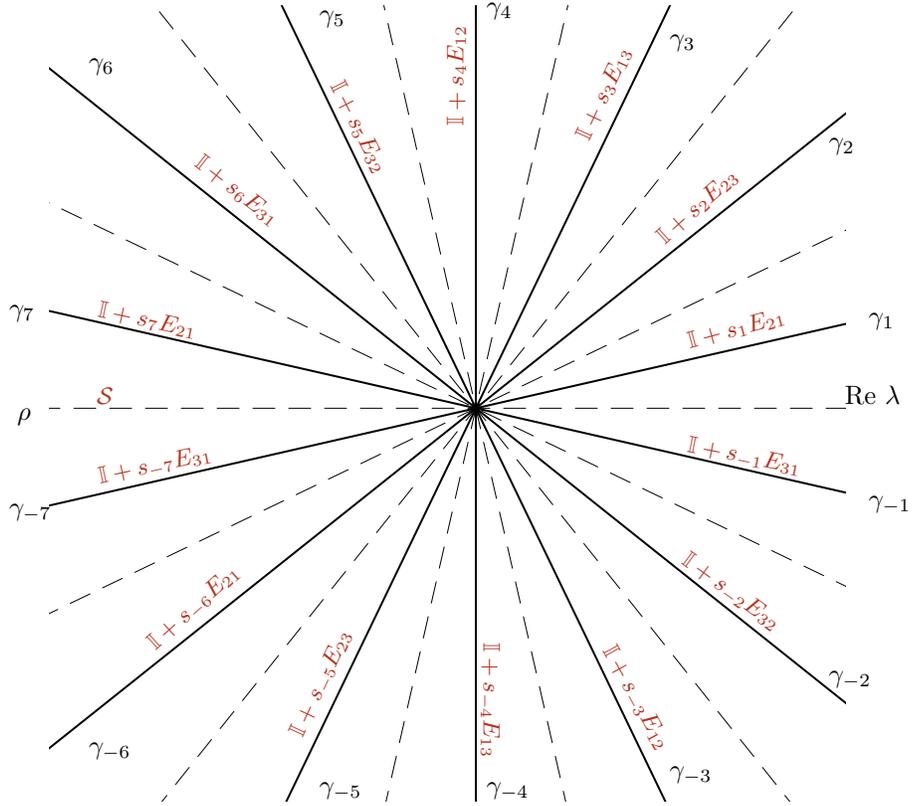

\begin{prop}
    Define Stokes rays $\{\gamma_k\}$, $k=\pm 1,...,\pm 7$, as shown in Figure \ref{fig:StokesRays}. Explicitly, these rays are defined as
    \begin{align*}
        \gamma_{\pm k} := \left\{\lambda \big| \arg \lambda = \pm \frac{\pi}{14} \pm \frac{\pi}{7}(k-1) \right\}, \qquad k = 1,...,7.
    \end{align*}
    Furthermore, set $\rho := (-\infty,0)$; orient all of these rays outwards from the origin. Let $\{S_k,S_{-k}\}_{k=1,...,7}$ be a collection of 
    constant matrices of the form given in \ref{fig:StokesRays}, subject to the constraint 
        \begin{equation}
            S_{-7}\cdots S_{-1} S_{1} \cdots S_{7} = \mathcal{S}^T.
        \end{equation}
    Consider the following Riemann-Hilbert problem for a $3\times 3$ sectionally analytic function $\Psi(\lambda;t_{5},t_{2},x)$:
        \begin{equation}
            \begin{cases}
                \Psi_{+}(\lambda;t_{5},t_{2},x) =  \Psi_{-}(\lambda;t_{5},t_{2},x) S_k, & \lambda \in \gamma_k,\qquad k = \pm 1, ...,\pm 7,\\
                \Psi_{+}(\lambda;t_{5},t_{2},x) =  \Psi_{-}(\lambda;t_{5},t_{2},x) \mathcal{S}, & \lambda \in \rho,\\
                \Psi(\lambda;t_{5},t_{2},x) = g(\lambda)\left[\mathbb{I} + \frac{\Phi_1}{\lambda^{1/3}} + \frac{\Phi_2}{\lambda^{2/3}} + \OO(\lambda^{-1})\right]e^{\Theta(\lambda^{1/3};t_{5},t_{2},x)}, & \lambda \to \infty,
            \end{cases}
        \end{equation}
    where $g(\lambda)$ is as defined in \eqref{gauge-matrix}, and $\Phi_1$, $\Phi_2$ are as defined in Equations \eqref{first-coeff}, \eqref{second-coeff}.
    Then, the solution to this Riemann-Hilbert problem is unique, provided the asymptotics (this includes $\Phi_1,\Phi_2$!) above are specified. Furthermore,
    The functions
        \begin{align*}
            U(t_{5},t_{2},x) &:= 2\frac{d}{dx}\left[\Phi_1(t_{5},t_{2},x)\right]_{11},\\
            V(t_{5},t_{2},x) &:= -2\frac{d}{dx}\left[\Phi_2(t_{5},t_{2},x) - \frac{1}{2}\Phi_1(t_{5},t_{2},x)^2\right]_{11} = -\frac{d}{dt_2}\left[\Phi_1(t_{5},t_{2},x)\right]_{11},
        \end{align*}
    satisfy the string equation.
\end{prop}
As a corollary of the results of this section, we obtain the result of  Proposition \ref{PropA}:
    \begin{cor}
        Proposition \ref{PropA} holds.
    \end{cor}
\begin{remark}
    It is important to note that the solution to the above Riemann-Hilbert problem is not unique \textit{unless the coefficients $\Phi_1$, $\Phi_2$ in the asymptotic expansion are specified.} This is again a phenomenon shared by the Painlev\'{e} I system \cite{Kapaev,Kapaev-Kitaev,FIKN}, and is ultimately 
    due to the gauge freedom arising from the resonant singularity at the origin. Since all choices of gauge lead to the same integrability condition, we are
    free to fix a gauge, and work in it. Fixing a gauge is equivalent to making a choice of the form of $\Phi_1$, $\Phi_2$, up to multiplication by a 
    unimodular lower triangular matrix. This gauge freedom was first pointed out in the work of Drinfeld and Sokolov \cite{DrinfeldSokolov}.
\end{remark}

\subsection{A non-local $\ZZ_2$ symmetry reduction.}\label{Z2-Symmetry-Breaking}
As it turns out, there is an additional (and generically, non-local) symmetry of the above Riemann-Hilbert problem, which can be identified with the $\ZZ_2$-parity of the magnetic field. This is summarized in the following Proposition:
    \begin{prop}
        Let $\Phi(\xi;t_{5},t_{2},x)$ be a solution to the Riemann-Hilbert problem defined by Proposition \ref{RHP-Phi-prop}. Then, 
            \begin{equation}
                \Phi(-\xi;t_{5},-t_{2},x)^{-T} = \Phi(\xi;t_{5},t_{2},x),
            \end{equation}
        provided the Stokes parameters satisfy
            \begin{equation}\label{Z2-Stokes-symmetry}
                s_k(t_{5},t_{2},x;U,U',U'',U''',V,V') = -s_{k+7}(t_{5},-t_{2},x;U,U',U'',U''',-V,-V'),\qquad k\in \ZZ_{14}.
            \end{equation}
    \end{prop}
    \begin{proof}
        Note that $-\Theta(-\xi;t_{5},-t_{2},x)^{T} = \Theta(\xi;t_{5},t_{2},x)$. This implies that the functions $\Phi(\xi;t_{5},t_{2},x)$ and
        $\Phi(-\xi;t_{5},-t_{2},x)^{-T}$ both have the same leading-order asymptotics at infinity. Thus, if these two functions have the same jumps,
        then their ratio is holomorphic, and equal to the identity at infinity; the usual Liouville argument then implies that 
        $\Phi(\xi;t_{5},t_{2},x) = \Phi(-\xi;t_{5},-t_{2},x)^{-T}$. Comparing the jumps of the these two functions, we see that the Stokes 
        matrices must satisfy
            \begin{equation*}
                S_k(t_{5},t_{2},x,U,U',U'',U''',V,V') = S_{k+21}(t_{5},-t_{2},x,\check{U},\check{U}',\check{U}'',\check{U}''',\check{V},\check{V}')^{-T} ,\qquad k\in \ZZ_{42},
            \end{equation*}
        where $\check{f}(t_{5},t_{2},x) = f(t_{5},-t_{2},x)$. Using the relations \eqref{Stokes-equations}, along with the formulae \eqref{Stokes-structure},
        this implies the following relation on the Stokes parameters:
            \begin{equation}\label{Z2-sym-2}
                s_k(t_{5},t_{2},x,U,U',U'',U''',V,V') = -s_{k+7}(t_{5},-t_{2},x,\check{U},\check{U}',\check{U}'',\check{U}''',\check{V},\check{V}')^T,\qquad k\in \ZZ_{14}.
            \end{equation}
        In fact, the above is equivalent to \eqref{Z2-Stokes-symmetry}. To see this, suppose that the relation \eqref{Z2-sym-2} holds, and expand the solutions
        $\Phi(\xi;t_{5},t_{2},x)$, $\Phi(-\xi;t_{5},-t_{2},x)^{-T}$ at infinity. One finds that
            \begin{align*}
                \Phi(\xi;t_{5},t_{2},x) &= \left[\mathbb{I} + \frac{\Phi_1(t_{5},t_{2},x)}{\xi} + \frac{\Phi_2(t_{5},t_{2},x)}{\xi^2} + \OO(\xi^{-3})\right]e^{\Theta(\xi;t_{5},t_{2},x)},\\
                \Phi(-\xi;t_{5},-t_{2},x)^{-T} &= \left[\mathbb{I} + \frac{\Phi_1^{T}(t_{5},-t_{2},x)}{\xi} + \frac{\Phi_1^{2T}(t_{5},-t_{2},x) - \Phi_2^T(t_{5},-t_{2},x)}{\xi^2} + \OO(\xi^{-3})\right]e^{\Theta(\xi;t_{5},t_{2},x)}
            \end{align*}
        Equating the coefficients\footnote{One should calculate the coefficients up to $\Phi_3$; from here, all of the relations stated can be inferred. This inference is direct for all of the relations except the one for $H_{5}$; this relation can be inferred from the rest of the relations, and the formula \eqref{Hamiltonian-eta} for $H_{5}$.}, one finds that
            \begin{equation*}
                H_{1}(t_{5},-t_{2},x) = H_{1}(t_{5},t_{2},x),\qquad H_{2}(t_{5},-t_{2},x) = -H_{2}(t_{5},t_{2},x), \qquad H_{5}(t_{5},-t_{2},x) = H_{5}(t_{5},t_{2},x),
            \end{equation*}
            \begin{equation*}
                U(t_{5},-t_{2},x) = U(t_{5},t_{2},x),\qquad\qquad V(t_{5},-t_{2},x) = -V(t_{5},t_{2},x).
            \end{equation*}
        In other words, $U,H_{1},H_{5}$ are even functions of $t_{2}$, and $H_{2}, V$ are odd functions of $t_{2}$. This justifies the equivalence of
        \eqref{Z2-sym-2} and \eqref{Z2-Stokes-symmetry}.
    \end{proof}
As a consequence of the above Proposition, we obtain a number of important corollaries:
\begin{cor}
    If the Stokes parameters satisfy Relation \eqref{Z2-Stokes-symmetry}, then functions $U,V$, and the Hamiltonians $H_{1}, H_{2}, H_{5}$ 
    satisfy the following relations:
        \begin{equation}\label{Z2-Hamiltonians}
            H_{1}(t_{5},-t_{2},x) = H_{1}(t_{5},t_{2},x),\quad H_{2}(t_{5},-t_{2},x) = -H_{2}(t_{5},t_{2},x), \quad H_{5}(t_{5},-t_{2},x) = H_{5}(t_{5},t_{2},x),
        \end{equation}
        \begin{equation}\label{Z2-UV}
            U(t_{5},-t_{2},x) = U(t_{5},t_{2},x),\qquad\qquad V(t_{5},-t_{2},x) = -V(t_{5},t_{2},x).
        \end{equation}
    Furthermore, the Okamoto $\tau$-function, defined by $d\log\tau_{Okamoto} = H_{5}dt_{5} + H_{2}dt_{2} + H_{1}dx$, satisfies
    \sloppy $\tau_{Okamoto}(t_{5},-t_{2},x) = \tau_{Okamoto}(t_{5},t_{2},x)$.
\end{cor}
\begin{cor}
    If the Stokes parameters satisfy Relation \eqref{Z2-Stokes-symmetry}, and $t_{2} = 0$, then $V \equiv 0$, $H_{2} \equiv 0$, and the string equation
    \eqref{string-equation} reduces to
        \begin{equation}
            \frac{1}{12} U^{(4)} -\frac{3}{4}U''U -\frac{3}{8}(U')^2 + \frac{1}{2}U^3 - \frac{5}{12}t_{5}\left(3U^2 - U''\right) + x.
        \end{equation}
    The only other nonzero part of the string equation is then
        \begin{equation}
            \frac{\partial U}{\partial t_{5}} =\frac{\partial}{\partial x}\left[-\frac{1}{6}UU'' + \frac{1}{8}(U')^2 + \frac{1}{4}U^3 
            - \frac{5}{9}t_{5}\left(3U^2-U''\right) + \frac{4}{3}x\right].
        \end{equation}
    Furthermore, the generic dimension of the Stokes manifold is reduced from $6$ to $4$.
\end{cor}
\begin{proof}
    On the hyperplane $t_{2} = 0$, the nonlocal equations \eqref{Z2-Stokes-symmetry}, \eqref{Z2-Hamiltonians}, and \eqref{Z2-UV} become local. In particular,
    we see that
        \begin{equation*}
            V(t_{5},0,x) = V(t_{5},-0,x) = -V(t_{5},0,x),
        \end{equation*}
    i.e. $V(t_{5},0,x) \equiv 0$. Similarly, we obtain that $H_{2}(t_{5},0,x) \equiv 0$. Finally, when $t_{2} = 0$, the $\ZZ_{14}$-periodicity of the 
    Stokes parameters further reduces to a $\ZZ_{7} \times \ZZ_{2}$-periodicity:
        \begin{equation*}
            s_{k+7} = -s_k,
        \end{equation*}
    and consequently the relation $S_1\cdots S_{14} = \mathcal{S}^T$ implies that the Stokes manifold is generically of dimension $4$.
\end{proof}
\begin{remark}
    The generic solution to the constraint equations \eqref{Stokes-equations} on the $t_{2} = 0$ hyperplane (which further implies $s_{k+7} = -s_k$) 
    is given by
        \begin{align}
            s_{5} &= \frac{s_1s_4 + s_3 + 1}{s_1s_3 - s_2},\qquad s_6 = -\frac{s_1(s_2s_4+s_3)- s_4(s_1s_3 - s_2) -s_2s_3 }{(s_1s_3 - s_2)(s_2s_4 + s_3)},
            \qquad s_7 = \frac{s_1s_4 - s_2 + 1}{s_2s_4 + s_3}.
        \end{align}
    where $s_1,s_2,s_3,s_4$ are free parameters, and provided $s_1s_3 - s_2\neq 0$, $s_1s_3 - s_2 \neq 0$. This reduction of the dimension of the Stokes 
    manifold is consistent with the reduction of the order of the string equation from $6$ to $4$.
\end{remark}
This symmetry gives an interpretation to our previous statement that $V$ is responsible for the non-perturbative $\ZZ_2$ symmetry-breaking of the model: this function is only non-zero when $t_{2}$, the parameter that is identified with a shift in the magnetic field, is nonzero.
We point out that this symmetry also passes through to the $\lambda$-gauge, without issue.

\begin{remark}\textit{$\ZZ_{14}$ symmetry.}
    Finally, we remark that the above symmetry is a special case of a more general $\ZZ_{14}$ symmetry, which is the analog of 
    the $\ZZ_5$ symmetry possessed by Painlev\'{e} I \cite{Kapaev0}; it is in fact just a realization of the subgroup $\ZZ_2$ of $\ZZ_{14}$.
    Let $\beta := e^{\frac{2\pi i}{42}} = e^{\frac{\pi i}{21}}$ denote the principal $42^{nd}$ root of unity, and note that
        \begin{equation}
            \mathcal{S}^T\Theta(\beta^{-1}\xi;\beta^{12}t_{5},\beta^{9}t_{2},\beta^{-6}x)\mathcal{S} = -\Theta(\xi;t_{5},t_{2},x).
        \end{equation}
    An identical line of argumentation to the preceding section shows that $\Phi(\xi;t_{5},t_{2},x)$ satisfies the symmetry condition
        \begin{equation}
            \Phi(\xi;t_{5},t_{2},x) = \mathcal{S}^T \Phi(\beta^{-1}\xi;\beta^{12}t_{5},\beta^{9}t_{2},\beta^{-6}x)\mathcal{S},
        \end{equation}
    Provided that the Stokes parameters satisfy
        \begin{equation} \label{Z-42}
            s_k(t_{5},t_{2},x;U,U',U'',U''',V,V') = -s_{k+1}(\beta^{12}t_{5},\beta^{9}t_{2},\beta^{-6}x;\check{U},\check{U}',\check{U}'',\check{U}''',\check{V},\check{V}'),
        \end{equation}
    where here $\check{f}(t_{5},t_{2},x) = f(\beta^{12}t_{5},\beta^{9}t_{2},\beta^{-6}x)$.
    We also have the following relations:
        \begin{align*}
            \beta^2 U(\beta^{12}t_{5},\beta^{9}t_{2},\beta^{-6}x) &= U(t_{5},t_{2},x),\\
            \beta^3 V(\beta^{12}t_{5},\beta^{9}t_{2},\beta^{-6}x) &= V(t_{5},t_{2},x),
        \end{align*}
    and finally that the Okamoto $\tau$-function satisfies
        \begin{equation}
            \tau_{Okamoto}(\beta^{12}t_{5},\beta^{9}t_{2},\beta^{-6}x) = \tau_{Okamoto}(t_{5},t_{2},x).
        \end{equation}
    Based on the appearance of a $42^{nd}$ root of unity, one might be tempted to think that this system possesses a full $\ZZ_{42} = \ZZ_{7}\oplus \ZZ_3 \oplus \ZZ_2$ symmetry group.
    In fact, as we shall now show, the subgroup $\ZZ_3$ appears in a trivial manner, and thus does not play a role.
    Let us denote by $\chi$ the operation of acting by this symmetry on $\Phi$, i.e. the map 
        \begin{equation}
            \chi \left[\Phi(\xi;t_{5},t_{2},x)\right] := \mathcal{S}^T \Phi(\beta^{-1}\xi;\beta^{12}t_{5},\beta^{9}t_{2},\beta^{-6}x)\mathcal{S}.
        \end{equation}
    Clearly, $\chi^{42} = 1$, the identity map on $\Phi$. Note that $\chi^6$ is the generator of the subgroup $\ZZ_7$, 
    $\chi^{14}$ is the generator of the subgroup $\ZZ_3$, $\chi^{21}$ is the generator of the subgroup $\ZZ_2$. Let us first see that $\chi^{14} = 1$,
    the identity map. If we apply $\chi^{14}$, we obtain that
        \begin{equation*}
            \chi^{14}[\Phi(\xi;t_{5},t_{2},x)] = \mathcal{S}^T\Phi(\omega\xi;t_{5},t_{2},x)\mathcal{S} = \Phi(\xi;t_{5},t_{2},x),
        \end{equation*}
    as we already observed in Subsection \ref{Phi-symmetry-section}. So, the subgroup $\ZZ_3$ does not participate, and there is generically a $\ZZ_{14}$
    symmetry acting on the solutions.

    Note further that, if we apply the generator of the subgroup $\ZZ_2$ to $\Phi$, we obtain that
        \begin{equation}
            \chi^{21}\left[\Phi(\xi;t_{5},t_{2},x)\right] =  \Phi(-\xi;t_{5},-t_{2},x)^{-T},
        \end{equation}
    which is precisely the $\ZZ_2$ symmetry described in this subsection. The $\ZZ_7$ symmetry is nontrivial, i.e. the generator of this subgroup $\chi^{6}$
    acts nontrivially on $\Phi$. However, there is no clear simplification or physical interpretation of this symmetry, as was the case for the $\ZZ_2$ subgroup.
\end{remark}

\section{The Isomonodromic and Extended $\tau$-Function.}\label{tau-function-section}
Associated to \textit{almost} any linear differential equation with rational coefficients is an object called the isomonodromic $\tau$-function. The $\tau$-function
has many important properties. For example, its zeros determine where the inverse monodromy problem for the associated linear equation are not solvable \cite{Malgrange,Palmer}. Furthermore, the $\tau$-function itself is often the object that appears in many physical applications; this is also the
case for the multi-critical quartic $2$-matrix model. However, the word ``almost'' is the antagonist in this story. As we have seen, in many cases of interest,
the leading coefficient of the singularity of the connection matrix is not diagonalizable, or, if we make a change of gauge, a resonant Fuchsian singularity 
manifests at the origin. In either case, the theory introduced in \cite{JMU1} is not applicable. This motivates us to give a modified definition of the $\tau$-differential. Most of this section can be read completely independently of the rest of this work.

In this first part of this section, we will work in slightly more generality, in order to show that our definition is indeed a sensible extension of the $\tau$-function, as defined by Jimbo, Miwa, and Ueno. Our definition is meant to address the case of the general $(p,q)$ string equations, which all share the feature that $1.$ the leading term of the \textit{polynomial} connection matrix $A(\lambda)$ is not diagonalizable, and $2.$ in a suitably regularized gauge, the connection matrix develops a resonant Fuchsian singularity at the origin.

We divide this section into the following parts: in Subsection \ref{assumptions-subsection}, we lay out a set of assumptions for a model problem with a single non-diagonalizable singularity at infinity (or, equivalently, a resonant Fuchsian singularity at $0$), for which we will define a suitable $\tau$-differential. In Subsection \ref{tau-modification}, we will see the shortcoming of the original JMU definition, and
show that the modified definition of the $\tau$-function (up to an irrelevant constant factor) indeed makes sense, and coincides with the Okamoto $\tau$-function \eqref{Okamoto-tau-function} in the settings of the rest of this work. In Section \ref{tau-extension}, we extend the $\tau$-function
to the initial data of the associated Hamiltonian system (cf. Proposition \ref{Hamiltonian-Prop}), and verify Conjectures $1.$ and $2.$ of \cite{IP} for the system at hand.

We adopt the following set of notations. First, let $q\geq 2$. If $X:\CC\to M_q(\CC)$ is a matrix-valued function which admits a Laurent expansion at $\lambda = \infty$, we define
    \begin{equation}
        \langle X(\lambda)\rangle := \Res_{\lambda= \infty}\tr X(\lambda)d\lambda.
    \end{equation}
We list some of the key properties of $\langle \cdot \rangle$ below, which the reader may readily check:
    \begin{enumerate}
        \item \textit{(Cyclicity)} $\langle X(\lambda) Y(\lambda)\rangle = \langle  Y(\lambda) X(\lambda)\rangle$,
        \item \textit{(Integration by Parts)}$\langle \frac{\partial}{\partial \lambda} X(\lambda) \rangle = 0$, and, consequentially, 
        $\langle  X'(\lambda) Y(\lambda) \rangle = -\langle  X(\lambda) Y'(\lambda) \rangle$.
        \item If $X = X(\lambda; {\bf t})$ depends on additional parameters ${\bf t}$, and ${\bf d}$ denotes the exterior differential in these parameters, then
        ${\bf d} \langle X(\lambda; {\bf t})\rangle  = \langle {\bf d} X(\lambda; {\bf t})\rangle$.
        \item If $A$ is a constant (in $\lambda$) matrix, then $\langle A\lambda^{-k-1}\rangle = (\tr A) \delta_{k,0}$. In particular, this implies that if 
        $X(\lambda)$ is a polynomial, then $\langle X(\lambda) \rangle = 0$. 
        \item \textit{(Ad-invariance)} If $A,B,C$ are matrix-valued functions, then $\langle A [B,C]\rangle = \langle [A,B] C\rangle$.
        \item \textit{Changes of variable}. If we make a change of variables of the form $\lambda(\zeta) = \zeta^q$, and retain the notation 
        $\langle X(\zeta)\rangle = \Res_{\zeta= \infty}\tr X(\zeta)d\zeta$, then $\langle Y(\lambda)\rangle = q\langle Y(\lambda(\zeta))\rangle$.
    \end{enumerate}
We now list a set of assumptions on which the remainder of our calculation will be based. The motivation for this assumptions comes from what one should expect out
of the $(p,q)$ string equation in general. When restricted to the case $q=3,p=4$, these assumptions coincide with what we have derived in the preceding sections of the present work. 

\subsection{Main assumptions for the model problem.}\label{assumptions-subsection}
Given $q\geq 2$, we fix an integer $p$ coprime to $q$. Setting $\omega_q:=e^{\frac{2\pi i}{q}}$, we then define $q\times q$ matrices $\Delta_q, \mathcal{U}_q$ as follows:
    \begin{align}
        \left(\Delta_q\right)_{ij} &:= \frac{1}{2}\left(1-\frac{2j-1}{q}\right) \delta_{ij}, \label{Delta-def}\\
        \left(\mathcal{U}_q\right)_{ij} &:= 
                \begin{cases}
                    \omega_q^{\frac{1}{2}(q-2i+1)(j-1)}, & q \textit{ odd},\\
                    \omega_q^{\frac{1}{2}(q-2i)(j-1) + 1}, & q \textit{ even},
                \end{cases} \label{U-def}
    \end{align}
where $i,j = 1,...,q$.
Also define the \textit{shift matrix}
    \begin{equation}
        \mathcal{S}_q :=
        \begin{psmallmatrix}
            0 & 1 & 0 & \hdots & 0 & 0\\
            0 & 0 & 1 & \hdots & 0 & 0\\
            \vdots & \vdots &\vdots & \ddots & \vdots & \vdots \\
            0 & 0 & 0 & \hdots & 0 & 1\\
            1 & 0 & 0 & \hdots & 0 & 0\\
        \end{psmallmatrix} \omega_q^{\frac{1}{2}r_q},
    \end{equation}
where $r_q = 0$ if $q$ is odd, and $r_q = 1$ if $q$ is even. (Note that $\mathcal{S}_q$ is unitary: $\mathcal{S}_q^{\dagger} = \mathcal{S}_q^{-1}$).
Finally, we define functions $\vartheta^{(q,p)}_j(\lambda;{\bf t})$ as
    \begin{equation}
        \vartheta^{(q,p)}_j(\lambda;{\bf t}) := \frac{q}{p+q}\omega_q^{(j-1)p}\lambda^{\frac{p+q}{q}} + \sum_{\substack{\ell=1\\ \ell \mod q \neq 0}}^{p+q-1} t_{\ell}\omega_q^{(j-1)\ell} \lambda^{\ell/q},
    \end{equation}
for $j=1,...,q$.
Subsequently, we define the matrix-valued functions
    \begin{equation}
        g_q(\lambda):= \lambda^{\Delta_q} \mathcal{U}_q, \qquad\qquad \Theta(\lambda;{\bf t}) := \text{diag }(\vartheta^{(q,p)}_1(\lambda;{\bf t}),\cdots, \vartheta^{(q,p)}_{q}(\lambda;{\bf t})).
    \end{equation}
If we denote $\Theta_a := \frac{\partial \Theta}{\partial t_a}$, we can see that the conditions
\begin{equation}\label{Theta-integrable}
        \frac{\partial \Theta_a}{\partial t_b} - \frac{\partial \Theta_b}{\partial t_a} = [\Theta_a,\Theta_b] = 0;
    \end{equation}
hold trivially. By construction, these matrices have jumps on the negative real axis (with orientation taken outwards), given by:
    \begin{lemma} For $\lambda < 0$,
        \begin{equation}
            g_{q,+}(\lambda) = g_{q,-}(\lambda) \mathcal{S}_q, \qquad\qquad \Theta_+(\lambda;{\bf t}) = \mathcal{S}_q^{\dagger}\Theta_-(\lambda;{\bf t})\mathcal{S}_q
        \end{equation}
    \end{lemma}
Consider the following lemma:
\begin{lemma}
    Let $g_q(\lambda)$ be an $\text{SL}_q(\CC)$-valued function on $\CC\setminus (-\infty,0]$ such that $g_{q,+}(\lambda) = g_{q,-}(\lambda) \mathcal{S}_q$, where $\mathcal{S}_{q}\in \text{SL}_d(\CC)$.
    Consider the series
        \begin{equation}
            R(\lambda) := \mathbb{I} + \sum_{m=1}^{\infty}\frac{\Psi_m}{\lambda^{m/q}}.
        \end{equation}
    Then, the function
        \begin{equation}
            \hat{R}(\lambda) := g_q(\lambda) R(\lambda) g_q^{-1}(\lambda)
        \end{equation}
    is holomorphic in a neighborhood of infinity if and only if the coefficients $\Psi_m$ satisfy the symmetry relation
        \begin{equation}
            \Psi_m = \omega_q^{-m} \mathcal{S}_q^{-1}\Psi_m \mathcal{S}_q.
        \end{equation} 
\end{lemma}
    \begin{proof}
        Set $\Psi_0:=\mathbb{I}$. For each $r=0,...,q-1$, set
            \begin{equation*}
                R_r(\lambda) := \sum_{k=0}^{\infty} \Psi_{kq+r} \lambda^{-k}.
            \end{equation*}
        Then, $R(\lambda)$ can be rewritten as
            \begin{equation*}
                R(\lambda) = \sum_{r=0}^{q-1} R_r(\lambda)\lambda^{r/q}.
            \end{equation*}
        The functions $R_r(\lambda)$ are analytic at infinity, and satisfy the relation $R_r(\lambda) = \omega_q^r \mathcal{S}_q^{-1}R_r(\lambda) \mathcal{S}_q$, since 
        $\omega_q^{kq+r} = \omega_q^r$. Since $\lambda^{r/q}_+ = \lambda^{r/q}_- \omega_q^r$, it follows that
            \begin{equation*}
                \left[R_r(\lambda)\lambda^{r/q}\right]_+ = \mathcal{S}_q^{-1}\left[R_r(\lambda)\lambda^{r/q}\right]_-\mathcal{S}_q,
            \end{equation*}
        for each $r = 0,...,q-1$, and thus that 
        \begin{equation*}
            R_+(\lambda) = \mathcal{S}_q^{-1}R_-(\lambda)\mathcal{S}_q.
        \end{equation*}
        Therefore, 
        \begin{equation*}
            \hat{R}_+(\lambda) = g_{+}(\lambda)R_+(\lambda)g_{+}^{-1}(\lambda) = g_{-}(\lambda)\mathcal{S}_q\mathcal{S}_q^{-1}R_-(\lambda)\mathcal{S}_q\mathcal{S}_q^{-1}g_{-}^{-1}(\lambda) = g_{-}(\lambda)R_-(\lambda)g_{-}^{-1}(\lambda),
        \end{equation*}
        and thus $\hat{R}(\lambda)$ has no jumps near $\lambda = \infty$. Thus, $\hat{R}(\lambda)$ extends to a holomorphic function in a neighborhood of infinity.
        Reading the above proof from bottom to top yields the other direction of the lemma.
    \end{proof}
With this lemma in mind, we are motivated to define the asymptotic series
    \begin{equation}\label{Psi-formal-expansion}
        \Psi(\lambda;{\bf t}) := g_q(\lambda)\left[\mathbb{I} + \sum_{k=1}^{\infty}\frac{\Psi_k({\bf t})}{\lambda^{k/q}}\right]e^{\Theta(\lambda;{\bf t})},
    \end{equation}
where the coefficients $\Psi_k({\bf t})$ satisfy the symmetry constraint $\omega_q^{-k}\mathcal{S}_q^{\dagger}\Psi_k({\bf t})\mathcal{S}_q = \Psi_k({\bf t})$. This formal
series therefore satisfies the jump condition
    \begin{equation*}
        \Psi_{+}(\lambda;{\bf t}) = \Psi_-(\lambda;{\bf t})\mathcal{S}_q,\qquad\qquad \lambda<0,
    \end{equation*}
as a consequence of the above lemma. Similarly, if we define the function $\mathfrak{G}(\lambda;{\bf t}) := \Psi(\lambda;{\bf t}) e^{-\Theta(\lambda;{\bf t})}$, we can 
see that $\mathfrak{G}_+(\lambda;{\bf t}) = \mathfrak{G}_-(\lambda;{\bf t})\mathcal{S}_q$, for $\lambda<0$. We assert that $\Psi(\lambda;{\bf t})$ is a (formal) solution to the 
following collection of differential equations
    \begin{equation}\label{Psi-formal-equations}
        \frac{\partial \Psi}{\partial \lambda} = A(\lambda;{\bf t}) \Psi(\lambda;{\bf t}),\qquad\qquad \frac{\partial \Psi}{\partial t_{\ell}} = B_{\ell}(\lambda;{\bf t}) \Psi(\lambda;{\bf t}), \quad \ell = 1,...,p+q-1,\quad \ell \mod q \not \equiv 0.
    \end{equation}
Here, all matrices $A(\lambda;{\bf t})$, $B_{\ell}(\lambda;{\bf t})$ are assumed to be \textit{polynomials} in $\lambda$. Furthermore, by formal differentiation of the series \eqref{Psi-formal-expansion}, one can deduce that the leading coefficient of $A(\lambda;{\bf t})$ is (for $p=kq+r$)
    \begin{equation}
        A(\lambda;{\bf t}) = 
        \Lambda^r \lambda^k + \cdots,
    \end{equation}
where $\Lambda = \Lambda(\lambda)$ is the matrix
    \begin{equation}
        \Lambda(\lambda) := 
            \begin{pmatrix}
                0 & 0 & \cdots & 0 & 0 & \lambda\\
                1 & 0 & \cdots & 0 & 0 & 0\\
                0 & 1 & \cdots & 0 & 0 & 0\\
                \vdots & \vdots & \ddots & \vdots & \vdots & \vdots\\
                0 & 0 & \cdots & 1 & 0 & 0\\
                0 & 0 & \cdots & 0 & 1 & 0
            \end{pmatrix}.
    \end{equation}
In particular, it is apparent that the leading coefficient of $A(\lambda;{\bf t})$ is \textit{not} diagonalizable. If we perform a gauge transformation $\lambda = \zeta^{q}$, $\Psi = g_q\Phi$, then the transformed connection matrices have the following properties:
\begin{prop}
    Under the change of gauge $\lambda = \zeta^{q}$, $\Psi = g_q\Phi$, the matrix $\Phi$ satisfies the differential equations
        \begin{equation}\label{general-Y-system}
            \frac{\partial \Phi}{\partial \zeta} = \hat{A}(\zeta;{\bf t}) \Phi(\zeta;{\bf t}),\qquad\qquad \frac{\partial \Phi}{\partial t_{\ell}} = \hat{B}_{\ell}(\zeta;{\bf t}) \Phi(\zeta;{\bf t}),
        \end{equation}
    for $\ell = 1,..,p+q-1, \ell \mod q \not \equiv 0$, where the matrices $\hat{A}(\zeta;{\bf t})$, $\hat{B}_{\ell}(\zeta;{\bf t})$ are given by
        \begin{align}
            \hat{A}(\zeta;{\bf t}) = q\zeta^{q-1}\tilde{A}(\zeta^q;{\bf t}),\\
            \hat{B}_{\ell}(\zeta;{\bf t}) = g_q^{-1}B_{\ell}(\zeta^q;{\bf t})g_q,
        \end{align}
    where $\tilde{A}(\lambda;{\bf t}) = g_q^{-1} A(\lambda;{\bf t})g_q - g_q^{-1}\frac{d g_q}{d\lambda}.$
\end{prop}
This proposition is a direct analog of the calculations in \S3, and so we omit the proof. We want to emphasize the following facts about these matrices:
    \begin{enumerate}
        \item For $\zeta\to \infty$, $A(\zeta;{\bf t})$ has asymptotics
            \begin{equation}
                A(\zeta;{\bf t}) = q
                \begin{psmallmatrix}
                    1 & 0 & \hdots\\
                    0 & \omega_q^{p} & \hdots\\
                    \vdots & \vdots & \ddots \\
                \end{psmallmatrix}\zeta^{p+q-1} + \OO(\zeta^{p+q-2}),
            \end{equation}
        so we have indeed `regularized' the singular point at infinity: the leading term is \textit{diagonal}.
            
        \item The matrix $A(\zeta;{\bf t})$ has a first order pole at $\zeta = 0$, which arises from the term $g_q^{-1}\frac{d g_q}{d\lambda}$. Note that this term does not depend on the coefficients of the matrix $A(\lambda;{\bf t})$, and thus can be computed explicitly:
            \begin{equation}
                \lim_{\zeta \to 0} \zeta \hat{A}(\zeta;{\bf t}) = -q \mathcal{U}_q^{-1} \Delta_q \mathcal{U}_q,
            \end{equation}
        where $\Delta_q$ was the diagonal matrix from before (cf. Equation \eqref{Delta-def}). We have that $q(\Delta_q)_{jj} - q(\Delta_q)_{ii} = j-i$, and
        \begin{equation*}
            \max_{i,j} q|(\Delta_q)_{jj} - (\Delta_q)_{ii}| = q-1.
        \end{equation*}
        Crucially, we observe that this singularity is resonant.
        \item Due to the form of $g_q(\lambda(\zeta))$, the matrices $\hat{B}_{\ell}(\zeta;{\bf t})$ develop poles of order at most $q-1$ at $\zeta = 0$.
        \item The zero-curvature equations hold in the $\zeta$-gauge as well (this is just the trivial observation that the zero-curvature equations hold, independent of the choice of coordinate system).
    \end{enumerate}
Thus, the above isomonodromic system has analogous properties (and complications) as the Painlev\'{e} I system in \cite{JMU1} and the $(3,4)$ string equation discussed in this work. To conclude this section, let us summarize our main assumptions about the system we will be studying.
    \begin{itemize}
        \item \textbf{Assumption 1.} We are given a matrix-valued formal series $\Psi(\lambda;{\bf t})$ of the form \eqref{Psi-formal-expansion}.
        \item \textbf{Assumption 2.} $\Psi(\lambda;{\bf t})$ satisfies the differentials equations \eqref{Psi-formal-equations}, for
        \textit{polynomial} matrices $A(\lambda;{\bf t})$, $B(\lambda;{\bf t})$.
        \item \textbf{Assumption 3.} The zero-curvature equations $\frac{\partial A}{\partial t_{\ell}} - \frac{\partial B_{\ell}}{\partial \lambda} + [A,B_{\ell}] = 0$, $\frac{\partial B_{r}}{\partial t_{\ell}} - \frac{\partial B_{\ell}}{\partial t_r} + [B_r,B_{\ell}] = 0$, hold.
    \end{itemize}
One can see that the system we are studying in the present work emerges when we specialize to $q=3,p=4$; after an appropriate rescaling of variables, 
one also sees that the above system agrees with the $k^{th}$ member of the Painlev\'{e} I hierarchy upon specializing $q=2,p=2k+1$.

\subsection{Modification of $\omega_{JMU}$.}\label{tau-modification}
Before proceeding to discuss our modification of the $\tau$-differential, let us clarify why there is a need for such a modification.
First, if we start with a connection matrix $A(\lambda;{\bf t})$ whose leading term is not diagonalizable, then the standard definition
given by Jimbo, Miwa and Ueno fails to hold. One can then attempt to transform into a gauge which resolves this problem, as we have 
discussed. In this new gauge, the situation can be treated by \cite{JMU1} at $\zeta = \infty$. However, one finds that a new problem
arises at $\zeta = 0$: a resonant Fuchsian singularity emerges, which again brings us out of the context of the work of \cite{JMU1}. In
the literature for Painlev\'{e} I \cite{LR,IP}, this problem is typically surmounted by simply ignoring any contributions from the 
resonant singularity, and one is able to proceed without further complications. However, in the present situation (and also
the situation we outlined in the previous subsection), we are not afforded this luxury. Indeed, if we transform into the $\zeta$-gauge
and directly apply the definition of \cite{JMU1} to the connection \eqref{L-zeta}, simply ignoring the contribution from the
resonant singularity, we find that
    \begin{align}
        {\bf d }\,\boldsymbol{\omega}_{JMU} &= \left(U^3+3V^2+\frac{1}{4}(U')^2-\frac{1}{2}UU'' + \frac{5}{2}t_5\left(\frac{1}{3}U''-U^2\right) + 2x\right)dt_5dt_1 \nonumber\\
        &+ \left(U''V-3U^2V+2t_2U+\frac{25}{3}t_5^2V+\frac{10}{3}t_5t_2\right)dt_5dt_2 \neq 0.
    \end{align}
Of course, one can simply add to $\boldsymbol{\omega}_{JMU}$ the differential
    \begin{equation}
        \alpha = -\frac{1}{3} \left(\frac{1}{6}U''' - UU'\right) d t_{5} = -\frac{1}{3} \frac{\partial V}{\partial t_{2}}dt_{5},
    \end{equation}
whose differential is precisely $-{\bf d }\,\boldsymbol{\omega}_{JMU}$, so that this new, \textit{modified} differential is indeed closed. However, it is
not obvious where this term arises from, or how to treat closely related systems apart from ad-hoc analysis. 

The aim of this subsection is to provide a general definition of a modified $\tau$-differential $\hat{\boldsymbol{\omega}}_{JMU}$ for systems of the form discussed in the previous subsection, which has the following properties:
    \begin{enumerate}
        \item The modified differential is closed: ${\bf d }\,\hat{\boldsymbol{\omega}}_{JMU} = 0$,
        \item When there are no resonant Fuchsian singularities, the modified differential coincides with the definition given in \cite{JMU1}: $\hat{\boldsymbol{\omega}}_{JMU} = \boldsymbol{\omega}_{JMU}$.
        \item The modified differential is \textit{gauge-invariant}: if we replace the formal series at infinity $\Psi(\lambda;{\bf t})$ by 
        $\mathfrak{h}({\bf t})\Psi(\lambda;{\bf t})$, where $\mathfrak{h}({\bf t})$ is an upper-triangular matrix with $1$'s on the diagonal, the 
        modified differential does not change its value.
    \end{enumerate}
With this in mind, we define the \textit{modified $\tau$-differential} to be
\begin{defn}\label{modified-tau-differential-definition}
    \begin{equation}
        \hat{\boldsymbol{\omega}}_{JMU} = \sum_{\ell}\left(\left\langle A(\lambda;{\bf t}) \frac{\partial \mathfrak{G}}{\partial t_{\ell}} \mathfrak{G}^{-1}\right\rangle - \left\langle\frac{\Delta_q }{\lambda} \frac{\partial \mathfrak{G}}{\partial t_{\ell}} \mathfrak{G}^{-1}\right\rangle \right)dt_{\ell},
    \end{equation}
where $\Delta_q$ is as defined in Equation \eqref{Delta-def}.
Equivalently, expressed in terms of local quantities in the $\xi$-gauge,
    \begin{equation}
        \hat{\boldsymbol{\omega}}_{JMU} = q\sum_{\ell}\left\langle \hat{A}(\zeta;{\bf t}) \frac{\partial S}{\partial t_{\ell}} S^{-1}\right\rangle dt_{\ell},
    \end{equation}
where $S(\zeta;{\bf t}) = \mathbb{I} + \frac{\Psi_1({\bf t})}{\zeta} + \frac{\Psi_2({\bf t})}{\zeta^2} + \OO(\zeta^{-3})$, where
the residue is now taken at $\zeta = \infty$.
\end{defn}
Let us remark that both definitions indeed make sense, in that the terms inside the brackets $\langle\cdot\rangle$ are formal
Laurent series in $\lambda$ (respectively, $\zeta$). This requires no commentary in the latter case. In the former case, this is slightly 
more subtle: note that $\frac{\partial \mathfrak{G}}{\partial t_{\ell}}$ and $\mathfrak{G}$ both have jumps only on the left. Since the jump matrix for $\mathfrak{G}$ is constant, we see
that the ratio $\frac{\partial \mathfrak{G}}{\partial t_{\ell}} \mathfrak{G}^{-1}$ is single-valued near infinity, and thus admits a Laurent series expansion there. 

Our first important observation is that the above differential is indeed gauge-invariant: if one replaces $\mathfrak{G}(\lambda;{\bf t})$ in the right
hand side of \eqref{modified-tau-differential-definition} with $\mathfrak{h}({\bf t})\mathfrak{G}(\lambda;{\bf t})$, where $\mathfrak{h}({\bf t})$ is an
upper-triangular matrix with $1$'s on the diagonal, the left hand side remains 
unchanged.

Our second important observation is that this definition agrees with the definition given in \cite{JMU1} in the case when the resonant
Fuchsian singularity at the origin vanishes (equivalently, when the leading term of $A(\lambda;{\bf t})$ is diagonalizable). Recall that,
if $A(\lambda;{\bf t})$ is a polynomial in $\lambda$ with diagonalizable leading term, then we can write a formal series solution to
the differential equation $\frac{\partial \Psi}{\partial \lambda} = A(\lambda;{\bf t})\Psi$ as
    \begin{equation*}
        \Psi(\lambda;{\bf t}) = \underbrace{\left[\mathbb{I} + \frac{\Psi_1({\bf t})}{\lambda} + \frac{\Psi_2({\bf t})}{\lambda^2} + \OO(\lambda^{-3})\right]}_{\mathfrak{G}(\lambda;{\bf t})} e^{\Theta(\lambda;{\bf t})},
    \end{equation*}
where $\Theta$ is a diagonal matrix whose entries are \textit{polynomials} in $\lambda$. The definition of the Jimbo-Miwa-Ueno $\tau$-differential is then
    \begin{equation}
        \boldsymbol{\omega}_{JMU} := -\sum_{\ell} \left\langle \mathfrak{G}^{-1}\frac{d \mathfrak{G}}{d\lambda} \frac{\partial \Theta}{\partial t_{\ell}}\right\rangle dt_{\ell}.
    \end{equation}
An alternative, equivalent expression given later (cf. \cite{ILP,IP}) is
    \begin{equation*}
        \boldsymbol{\omega}_{JMU} = \sum_{\ell} \left\langle A(\lambda;{\bf t})\frac{\partial \mathfrak{G}}{\partial t_{\ell}}\mathfrak{G}^{-1}\right\rangle dt_{\ell}.
    \end{equation*}
Comparing this definition to our definition of $\hat{\boldsymbol{\omega}}_{JMU}$, we see that $\hat{\boldsymbol{\omega}}_{JMU}$ indeed reduces to $\boldsymbol{\omega}_{JMU}$
when we are in a standard case.

It remains to see that ${\bf d }\, \hat{\boldsymbol{\omega}}_{JMU} = 0$. In order to establish this, we work in the $\zeta$-gauge; transferring results back into the $\lambda$-gauge is a matter of using the change of variables formula for $\langle\cdot\rangle$. We establish this through a sequence of lemmas.

    \begin{lemma}\label{new-tau-lem-0}
        \begin{equation}
            \hat{\boldsymbol{\omega}}_{JMU} = \boldsymbol{\omega}_{JMU} + \sum_{\ell}\left\langle \hat{A}(\zeta;{\bf t}) \hat{B}_{\ell}(\zeta;{\bf t})\right\rangle dt_{\ell},
        \end{equation}
    where the residue here is taken at $\zeta = \infty$.
    \end{lemma}
\begin{proof}
    The proof mimics the calculation that $\hat{\boldsymbol{\omega}}_{JMU} = \boldsymbol{\omega}_{JMU}$ in the case where the matrices $\hat{A}(\zeta;{\bf t})$, $\hat{B}_{\ell}(\zeta;{\bf t})$ are all polynomials; we must take care to make sure that every step pushes through. Let us expand the expression
        \begin{equation*}
            \sum_{\ell} \langle \hat{A}(\zeta;{\bf t})\hat{B}_{\ell}(\zeta;{\bf t})\rangle dt_{\ell};
        \end{equation*}
    note that when the Lax matrices are polynomials, this expression vanishes identically.
    Now, since $\Phi(\zeta;{\bf t}) = S(\zeta;{\bf t})e^{\Theta(\zeta;{\bf t})}$, we can rearrange the identities
    $\frac{\partial \Phi}{\partial \zeta} = \hat{A}\Phi$, $\frac{\partial \Phi}{\partial t_{\ell}} = \hat{B}_{\ell}\Phi$ to read
    $\hat{A} = S \Theta_{\zeta} S^{-1} + \frac{\partial S}{\partial \zeta} S^{-1}$, $\hat{B}_{\ell} = S \Theta_{\ell} S^{-1} + \frac{\partial S}{\partial t_{\ell}} S^{-1}$. Inserting these expressions into our previous identity, we find that
        \begin{align*}
        \sum_{\ell} \langle \hat{A}(\zeta;{\bf t})\hat{B}_{\ell}(\zeta;{\bf t})\rangle dt_{\ell} &= \sum_a\langle \hat{A} \hat{B}_{\ell}\rangle dt_{\ell} = \sum_{\ell}\left\langle (S\Theta_{\zeta} S^{-1} + \frac{\partial S}{\partial \zeta} S^{-1}) (S\Theta_{\ell} S^{-1} + \frac{\partial S}{\partial t_{\ell}} S^{-1})\right\rangle dt_{\ell}\\
        &=\sum_{\ell} \left[ \underbrace{\left\langle\Theta_\zeta\Theta_{\ell}\right\rangle}_{\textit{polynom.}} + \left\langle \frac{\partial S}{\partial \zeta} S^{-1} \frac{\partial S}{\partial t_{\ell}} S^{-1}\right\rangle + \underbrace{\left\langle S^{-1}\frac{\partial S}{\partial \zeta} \Theta_{\ell}\right\rangle}_{-[\boldsymbol{\omega}_{JMU}]_{\ell}}
        + \left\langle \Theta_{\zeta} S^{-1} \frac{\partial S}{\partial t_{\ell}}\right\rangle\right] dt_{\ell}\\
        &= -\boldsymbol{\omega}_{JMU} + \sum_{\ell} \left[\left\langle \frac{\partial S}{\partial \zeta} S^{-1} \frac{\partial S}{\partial t_{\ell}} S^{-1}\right\rangle +\left\langle \Theta_{\zeta} S^{-1} \frac{\partial S}{\partial t_{\ell}} \right\rangle \right]dt_{\ell}\\
        &= -\boldsymbol{\omega}_{JMU} + \sum_{\ell} \left\langle  \left(S\Theta_{\zeta} S^{-1} + \frac{\partial S}{\partial \zeta} S^{-1}\right)\frac{\partial S}{\partial t_{\ell}} S^{-1}\right\rangle dt_{\ell}\\
        &=-\boldsymbol{\omega}_{JMU} + \sum_{\ell}\left\langle \hat{A}(\zeta;{\bf t}) \frac{\partial S}{\partial t_{\ell}} S^{-1}\right\rangle dt_{\ell}.
    \end{align*}
\end{proof}

\begin{lemma}\label{new-tau-lem-1}
    For the system on $\Phi(\zeta;{\bf t})$, 
        \begin{equation}
            {\bf d }\,\boldsymbol{\omega}_{JMU} = \sum \left\langle\frac{\partial \hat{B}_a}{\partial \zeta} \hat{B}_b\right\rangle dt_a\wedge dt_b.
        \end{equation}
    \textit{(Note that, since $\hat{B}_a$ are no longer polynomials in the $\zeta$-gauge, the expression on the right hand side does not necessarily vanish).}
\end{lemma}
\begin{proof}
    Let $\omega_a$ denote the coefficient of $dt_a$ in $\boldsymbol{\omega}_{JMU}$. We shall first calculate $\frac{\partial \omega_a}{\partial t_b}$. By direct calculation,
        \begin{align*}
           \frac{\partial \omega_a}{\partial t_b} &= \left\langle S^{-1} \frac{\partial S}{\partial t_b}S^{-1} \frac{\partial S}{\partial \zeta}\Theta_a \right\rangle 
                                                    - \left\langle S^{-1} \frac{\partial^2 S}{\partial t_b\partial \zeta}\Theta_a\right\rangle 
                                                    - \left\langle S^{-1} \frac{\partial S}{\partial \zeta}\frac{\partial \Theta_a}{\partial t_b}\right\rangle.
        \end{align*}
    Now, the equation $\frac{\partial \Phi}{\partial t_b} = \hat{B}_b(\zeta;{\bf t})\Phi$ implies that $\frac{\partial S}{\partial t_b} = \hat{B}_b S - S\Theta_b$. 
    So, we can rewrite the above as
        \begin{align*}
            \frac{\partial \omega_a}{\partial t_b} &= \left\langle S^{-1}\hat{B}_b \frac{\partial S}{\partial \zeta}\Theta_a \right\rangle
            -\left\langle \Theta_b S^{-1} \frac{\partial S}{\partial \zeta}\Theta_a \right\rangle -
            \left\langle S^{-1} \frac{\partial}{\partial \zeta} \left[\hat{B}_b S - S\Theta_b\right]\Theta_a\right\rangle - \left\langle S^{-1} \frac{\partial S}{\partial \zeta}\frac{\partial \Theta_a}{\partial t_b}\right\rangle\\
            &=-\left\langle \Theta_b S^{-1} \frac{\partial S}{\partial \zeta}\Theta_a \right\rangle - \left\langle S^{-1} \frac{\partial \hat{B}_b}{\partial \zeta} S\Theta_a \right\rangle + \left\langle S^{-1}\frac{\partial S}{\partial \zeta}\Theta_b\Theta_a \right\rangle + \left\langle \frac{\partial \Theta_b}{\partial \zeta}\Theta_a \right\rangle - \left\langle S^{-1} \frac{\partial S}{\partial \zeta}\frac{\partial \Theta_a}{\partial t_b}\right\rangle.
        \end{align*}
    Using the identity
        \begin{equation}\label{useful-identity}
            \left\langle \frac{\partial \Theta_b}{\partial \zeta}\Theta_a \right\rangle = \left\langle \frac{\partial }{\partial \zeta}\left(S\Theta_b S^{-1}\right)S\Theta_a S^{-1}\right\rangle + \left\langle S^{-1}\frac{\partial S}{\partial \zeta} [\Theta_a,\Theta_b]\right\rangle,
        \end{equation}
    cyclicity, and representing $S\Theta_a S^{-1} = \hat{B}_a - \frac{\partial S}{\partial t_a} S^{-1}$, the above can be arranged to read
        \begin{align*}
            \frac{\partial \omega_a}{\partial t_b} &= -\left\langle S^{-1}\frac{\partial S}{\partial \zeta}\frac{\partial\Theta_a}{\partial t_b}\right\rangle - \left\langle \frac{\partial \hat{B}_b}{\partial \zeta}S\Theta_a S^{-1} \right\rangle +
            \left\langle \frac{\partial}{\partial \zeta}\left(S\Theta_b S^{-1}\right)S\Theta_a S^{-1}\right\rangle\\
            &= -\left\langle S^{-1}\frac{\partial S}{\partial \zeta}\frac{\partial\Theta_a}{\partial t_b}\right\rangle
            - \left\langle \frac{\partial \hat{B}_b}{\partial \zeta}\hat{B}_a \right\rangle + \left\langle \frac{\partial \hat{B}_b}{\partial \zeta}\frac{\partial S}{\partial t_a} S^{-1} \right\rangle + \left\langle \frac{\partial \hat{B}_b}{\partial \zeta} B_a\right\rangle - \left\langle \frac{\partial \hat{B}_b}{\partial \zeta} \frac{\partial S}{\partial t_a}S^{-1}\right\rangle \\
            &- \left\langle \frac{\partial}{\partial \zeta}\left( \frac{\partial S}{\partial t_b}S^{-1} \right)B_a \right\rangle 
            + \left\langle \frac{\partial}{\partial \zeta}\left( \frac{\partial S}{\partial t_b}S^{-1} \right)\frac{\partial S}{\partial t_a}S^{-1}  \right\rangle 
        \end{align*}
    Using integration by parts on the second to last term, we obtain that
        \begin{equation*}
            \frac{\partial \omega_a}{\partial t_b} = -\left\langle S^{-1}\frac{\partial S}{\partial \zeta}\frac{\partial\Theta_a}{\partial t_b}\right\rangle + 
            \left\langle \frac{\partial \hat{B}_a}{\partial \zeta}\frac{\partial S}{\partial t_b}S^{-1} \right\rangle +\left\langle \frac{\partial}{\partial \zeta}\left( \frac{\partial S}{\partial t_b}S^{-1} \right)\frac{\partial S}{\partial t_a}S^{-1}  \right\rangle .
        \end{equation*}
    Now, the argument of the last term is of order $\OO(\zeta^{-2})$, and thus has no residue. So, our final expression for $\frac{\partial \omega_a}{\partial t_b}$ is
        \begin{equation*}
            \frac{\partial \omega_a}{\partial t_b} = \left\langle  \frac{\partial \hat{B}_a}{\partial \zeta}\frac{\partial S}{\partial t_b}S^{-1}\right\rangle- \left\langle S^{-1}\frac{\partial S}{\partial \zeta}\frac{\partial\Theta_a}{\partial t_b}\right\rangle.
        \end{equation*}
    Interchanging the roles of $a$ and $b$ allows one to compute $\frac{\partial \omega_b}{\partial t_a}$; the difference of these two quantities is
        \begin{equation*}
            \frac{\partial \omega_a}{\partial t_b} - \frac{\partial \omega_b}{\partial t_a} =\left\langle  \frac{\partial \hat{B}_a}{\partial \zeta}\frac{\partial S}{\partial t_b}S^{-1}\right\rangle- \left\langle  \frac{\partial \hat{B}_b}{\partial \zeta}\frac{\partial S}{\partial t_a}S^{-1}\right\rangle -\left\langle S^{-1}\frac{\partial S}{\partial \zeta}\left(\frac{\partial\Theta_a}{\partial t_b} - \frac{\partial\Theta_b}{\partial t_a}\right)\right\rangle.
        \end{equation*}
    Closedness of $\boldsymbol{\omega}_{JMU}$ is equivalent to the vanishing of the above expression, for all $a,b$. Indeed, it is easy to see that the last term vanishes, 
    by the integrability condition \eqref{Theta-integrable}; it remains to see that the expression
        \begin{equation*}
            \frac{\partial \omega_a}{\partial t_b} - \frac{\partial \omega_b}{\partial t_a} =  
              \left\langle  \frac{\partial \hat{B}_a}{\partial \zeta}\frac{\partial S}{\partial t_b}S^{-1}\right\rangle
              -\left\langle  \frac{\partial \hat{B}_b}{\partial \zeta}\frac{\partial S}{\partial t_a}S^{-1}\right\rangle
        \end{equation*}
    vanishes. Here is the first place where the fact that the matrices $\hat{B}_a$ are \textit{not} polynomials in $\zeta$ comes into play. Again writing the identity 
    \eqref{useful-identity}, and representing $S\Theta_a S^{-1} = \hat{B}_a - \frac{\partial S}{\partial t_a} S^{-1}$, we see that
        \begin{align*}
            \left\langle \frac{\partial \Theta_b}{\partial \zeta}\Theta_a \right\rangle &= \left\langle \frac{\partial \hat{B}_b}{\partial \zeta} \hat{B}_a\right\rangle 
            - \left\langle \frac{\partial \hat{B}_b}{\partial\zeta} \frac{\partial S}{\partial t_a} S^{-1}\right\rangle - \left\langle \frac{\partial}{\partial \zeta}\left( \frac{\partial S}{\partial t_b} S^{-1}\right)\hat{B}_a\right\rangle \\
            &- \left\langle \frac{\partial}{\partial \zeta}\left( \frac{\partial S}{\partial t_b}S^{-1} \right)\frac{\partial S}{\partial t_a}S^{-1}  \right\rangle + \left\langle S^{-1}\frac{\partial S}{\partial \zeta} [\Theta_a,\Theta_b]\right\rangle.
        \end{align*}
    Now, the integrability condition \eqref{Theta-integrable} implies that the last term vanishes; we also have already seen that the second to last term vanishes, as 
    its argument is residueless. Similarly, since the matrices $\Theta_a$ are polynomial in $\zeta$, the argument left hand side is residueless, and thus vanishes. Integrating the third term by parts, we obtain the identity
        \begin{equation*}
            \left\langle \frac{\partial \hat{B}_b}{\partial\zeta} \frac{\partial S}{\partial t_a} S^{-1}\right\rangle - \left\langle \frac{\partial \hat{B}_a}{\partial \zeta} \frac{\partial S}{\partial t_b} S^{-1}\right\rangle = \left\langle \frac{\partial \hat{B}_b}{\partial \zeta} \hat{B}_a\right\rangle= -\left\langle \frac{\partial \hat{B}_a}{\partial \zeta} \hat{B}_b\right\rangle, 
        \end{equation*}
    and so we see that
        \begin{equation*}
            \frac{\partial \omega_a}{\partial t_b} - \frac{\partial \omega_b}{\partial t_a} =  \left\langle \frac{\partial \hat{B}_a}{\partial \zeta} \hat{B}_b\right\rangle .
        \end{equation*}
    If the matrices $\hat{B}_a$ are polynomials, then the right hand side vanishes identically; otherwise, we obtain the expression above.
\end{proof}

\begin{lemma}\label{new-tau-lem-2}
    \begin{equation}
        {\bf d}\,\sum_{\ell}\left\langle \hat{A}(\zeta;{\bf t}) \hat{B}_{\ell}(\zeta;{\bf t})\right\rangle dt_{\ell} = -\sum_{a<b} \left\langle\frac{\partial \hat{B}_a}{\partial \zeta} \hat{B}_b\right\rangle dt_a\wedge dt_b.
    \end{equation}
\end{lemma}
\begin{proof}
    Put
        \begin{equation*}
            \sigma :=\sum_{\ell}\left\langle \hat{A}(\zeta;{\bf t}) \hat{B}_{\ell}(\zeta;{\bf t})\right\rangle dt_{\ell},
        \end{equation*}
    and let $\sigma_a:=\langle\hat{A}\hat{B}_a\rangle$ denote the coefficient of $dt_a$ of $\sigma$. We then have that, using the integrability conditions for
    $\hat{B}_{\ell}$, $\hat{A}$,
        \begin{align*}
                \frac{\partial \sigma_a}{\partial t_b} &= \left\langle \frac{\partial \hat{A}}{\partial t_b} \hat{B}_a\right\rangle + \left\langle \hat{A} \frac{\partial \hat{B}_a}{\partial t_b}\right\rangle \\
                &=\left\langle \left(\frac{\partial \hat{B}_b}{\partial \zeta} + [\hat{B}_b,\hat{A}]\right) \hat{B}_a\right\rangle + \left\langle \hat{A} \frac{\partial \hat{B}_a}{\partial t_b}\right\rangle\\
                &=\left\langle \frac{\partial \hat{B}_b}{\partial \zeta}\hat{B}_a\right\rangle + \left\langle \hat{A}[\hat{B}_a,\hat{B}_b]\right\rangle + \left\langle \hat{A} \frac{\partial \hat{B}_a}{\partial t_b}\right\rangle
            \end{align*}
        where in the last equality we have used the Ad-invariance of the bracket. Then, we can compute $[{\bf d }\sigma]_{ba}$ to be
            \begin{align*}
                [{\bf d }\sigma]_{ba} = \frac{\partial \sigma_a}{\partial t_b} - \frac{\partial \sigma_b}{\partial t_a} &= 2 \left\langle \frac{\partial \hat{B}_b}{\partial \zeta}\hat{B}_a\right\rangle + 2\left\langle \hat{A}[\hat{B}_a,\hat{B}_b]\right\rangle + \left\langle \hat{A} \left(\frac{\partial \hat{B}_a}{\partial t_b} - \frac{\partial \hat{B}_b}{\partial t_a}\right)\right\rangle\\
                &=2 \left\langle \frac{\partial \hat{B}_b}{\partial \zeta}\hat{B}_a\right\rangle + \left\langle \hat{A}[\hat{B}_a,\hat{B}_b]\right\rangle.
            \end{align*}
        In the last line, we have again used the integrability conditions for $\hat{B}_{\ell}$, $\hat{A}$. We claim that 
        $\left\langle \hat{A}[\hat{B}_a,\hat{B}_b]\right\rangle = -\left\langle \frac{\partial \hat{B}_b}{\partial \zeta}\hat{B}_a\right\rangle$. 
        On one hand, expanding $\left\langle \hat{A}[\hat{B}_a,\hat{B}_b]\right\rangle$,
            \begin{align*}
                \left\langle \hat{A}[\hat{B}_a,\hat{B}_b]\right\rangle &= \left\langle q\zeta^{q-1}\left(g_q^{-1}A(\zeta^q)g_q - g_q^{-1}\frac{dg_q}{d\lambda}\right)[g_q^{-1}B_a(\zeta^q)g_q,g_q^{-1}B_b(\zeta^q)g_q]\right\rangle\\
                &=\left\langle q\zeta^{q-1}\left(g_q^{-1}A(\zeta^q)g_q - g_q^{-1}\frac{dg_q}{d\lambda}\right)g_q^{-1}[B_a(\zeta^q),B_b(\zeta^q)]g_q\right\rangle\\
                &=\langle q\zeta^{q-1}A(\zeta^q)[B_a(\zeta^q),B_b(\zeta^q)]\rangle -\left\langle \frac{q\mathcal{U}_q\Delta_q \mathcal{U}_q^{-1}}{\zeta}[B_a(\zeta^q),B_b(\zeta^q)]\right\rangle\\
                &=-\left\langle \frac{q\mathcal{U}_q\Delta_q \mathcal{U}_q^{-1}}{\zeta}[B_a(\zeta^q),B_b(\zeta^q)]\right\rangle,
            \end{align*}
        where the last equality follows from the fact that all of the expressions inside the first bracket are polynomials. On the other hand,
        we calculate that
            \begin{equation*}
                \frac{\partial \hat{B}_b}{\partial \zeta} = -g_q^{-1}\frac{dg_q}{d\zeta}g_q^{-1}B_b(\zeta^q)g_q+g_q^{-1}\frac{\partial }{\partial\zeta}B_b(\zeta^q)g_q + g_q^{-1}B_b(\zeta^q) g_q,
            \end{equation*}
        and so 
            \begin{align*}
                \left\langle \frac{\partial \hat{B}_b}{\partial \zeta}\hat{B}_a\right\rangle &= \left\langle \frac{dg_q}{d\zeta}g_q^{-1}[B_a(\zeta^q),B_b(\zeta^q)]\right\rangle + \left\langle \frac{\partial }{\partial \zeta}(B_b(\zeta^q))B_a(\zeta^q)\right\rangle\\
                &= \left\langle \frac{q\mathcal{U}_q\Delta_q \mathcal{U}_q^{-1}}{\zeta}[B_a(\zeta^q),B_b(\zeta^q)]\right\rangle.
            \end{align*}
        It follows that
            \begin{equation}
                [{\bf d }\sigma]_{ba} = 2 \left\langle \frac{\partial \hat{B}_b}{\partial \zeta}\hat{B}_a\right\rangle + \left\langle \hat{A}[\hat{B}_a,\hat{B}_b]\right\rangle = \left\langle \frac{\partial \hat{B}_b}{\partial \zeta}\hat{B}_a\right\rangle;
            \end{equation}
        an integration by parts yields that this is equal to $-\left\langle \frac{\partial \hat{B}_a}{\partial \zeta}\hat{B}_b\right\rangle$. This completes the proof.
\end{proof}

As a result of these lemmas, we obtain as a corollary the result of Theorem \ref{tau-theorem.}:
    \begin{cor} \textit{Proof of Theorem \ref{tau-theorem.}/closedness of the modified $\tau$-differential.} Under the assumptions of Subsection \ref{assumptions-subsection}, and given
    Definition \ref{modified-tau-differential-definition}, 
        \begin{equation}
            {\bf d }\,\hat{\boldsymbol{\omega}}_{JMU} = 0.
        \end{equation}
    \end{cor}
\begin{proof}
    One simply must add the results of Lemmas \ref{new-tau-lem-1} and \ref{new-tau-lem-2}, in accordance with the fact that
    \begin{equation*}
        \hat{\boldsymbol{\omega}}_{JMU} = \boldsymbol{\omega}_{JMU} + \sum_{\ell}\left\langle \hat{A}(\zeta;{\bf t}) \hat{B}_{\ell}(\zeta;{\bf t})\right\rangle dt_{\ell},
    \end{equation*}
    as per Lemma \ref{new-tau-lem-0}.
\end{proof}

\begin{remark}
    Note that the explicit form of the matrices $\Delta_q,\mathcal{U}_q$ was not so important in the proof of this proposition. The only
    details that mattered were the fact that $g_q(\lambda)$ was of the form $g_q(\lambda)=\lambda^{\Delta_q}\mathcal{U}_q$, and the fact
    that $\Psi(\lambda;{\bf t})$, $\mathfrak{G}(\lambda;{\bf t})$ had jumps only on the right.
\end{remark}

\begin{remark}
    In principle, the above theorem/definition of the $\tau$-differential should follow from the work of Bertola and Mo \cite{BM} on isomonodromic deformations of resonant rational connections. We nevertheless feel our theorem is worth writing down, for the following reasons:
        \begin{itemize}
            \item Although our theorem is less general, the corresponding expression for the modified $\tau$-differential is more manageable,
            \item The expression for the $\tau$-differential in \cite{BM} is in terms of spectral invariants, whereas our expression is in terms of formal residues in the local gauge. This is more in line with the original expression for the $\tau$-differential provided by Jimbo, Miwa, and Ueno \cite{JMU1}, and furthermore is amenable to Deift-Zhou asymptotic analysis.
        \end{itemize}
\end{remark}

\begin{remark}\textit{Irrelevance of modification in the case of Painlev\'{e} I.}
    This construction is unnecessary in the case of the usual Painlev\'{e} I Lax pair, and so the $\tau$-differential as defined by Jimbo, Miwa, and Ueno \cite{JMU2} 
    or Lisovyy and Roussillon \cite{LR} agrees with the one given here. Recall that this Lax pair in the $\zeta$-gauge is given by (cf. \cite{JMU2}, Formula C5, or \cite{LR}, Formulae 2.4a and 2.4b)
        \begin{align*}
            \hat{A}(\zeta;t) &= (4\zeta^4+2q^2+t)\sigma_3 - (2p\zeta+(2\zeta)^{-1})\sigma_1 - (4q\zeta^2+2q^2+t)i\sigma_2,\\
            \hat{B}(\zeta;t) &= (\zeta + q/\zeta)\sigma_3 -iq\zeta^{-1}\sigma_2,
        \end{align*}
    where $\sigma_k$ are the standard Pauli matrices. Here, $q$ solves Painlev\'{e} I, and $p=q'$ (although the calculations we perform now are independent of this fact). Using the fact that $\tr(\sigma_j\sigma_k) = 2\delta_{jk}$, we find that
        \begin{equation*}
            \tr \hat{A}(\zeta;t)\hat{B}(\zeta;t) = 4\zeta(2\zeta^4+2q\zeta^2-q^2+t/2).
        \end{equation*}
   Hence $\langle \hat{A}(\zeta;t)\hat{B}(\zeta;t)\rangle = 0$, and $\hat{\boldsymbol{\omega}}_{JMU} = \boldsymbol{\omega}_{JMU}$ by Lemma \ref{new-tau-lem-0}. However, as we have seen in the rest of the present work, this construction is nontrivial in general. Indeed, one can readily check that for the next member of the Painlev\'{e} I hierarchy (the $(2,5)$ string equation), this contribution is indeed nontrivial.
\end{remark}

By comparing the coefficients of the expression we obtained in the previous proposition, we can show that $\boldsymbol{\omega}_{Okamoto}$ is a constant multiple of the
modified differential $\hat{\boldsymbol{\omega}}_{JMU}$:
    \begin{cor} (Proposition \ref{Okamoto-JMU-equivalence})
        The differentials $\boldsymbol{\omega}_{Okamoto}$, $\hat{\boldsymbol{\omega}}_{JMU}$ are related by
            \begin{equation}
                \boldsymbol{\omega}_{Okamoto} = 2\hat{\boldsymbol{\omega}}_{JMU}.
            \end{equation}
    \end{cor}
\begin{proof}
    The proof is straightforward, and follows from definitions. Note that $\mathcal{L}$ has terms of degree $7$, and so we must compute 
    $S(\zeta) = \mathbb{I} + \frac{\Phi_1}{\zeta} + \cdots$ to terms of order $\zeta^{-8}$. This calculation can be performed by applying our previous calculations (cf. Remark \ref{a-calculations}); one finds that the coefficients $\hat{\boldsymbol{\omega}}_{t_{5}}$, $\hat{\boldsymbol{\omega}}_{t_{2}}$, and $\hat{\boldsymbol{\omega}}_{x}$
    are differential polynomials in the variables $\{[\Phi_{k}]_{11}\}_{k=1}^7$, and the functions $U,V$. One can match these coefficients
    explicitly to the Hamiltonians from before, up to a proportionality factor of $2$, and so we can identify $\hat{\boldsymbol{\omega}}_{JMU}$ with
    $\boldsymbol{\omega}_{Okamoto}$.
\end{proof}
Consequentially, if we define the corresponding $\tau$-functions by $\tau = e^{\int \omega}$, we see that the $\tau$-functions arising from these definitions are related by $\tau_{Okamoto} = \tau_{JMU}^2$.

\subsection{The $\tau$-function on the extended monodromy data.} \label{tau-extension}
The $\tau$-function also depends intrinsically on the extended monodromy data of the system; in our case, the Stokes parameters $s_1,...,s_6$. It is thus
natural to ask the question \textit{What is the dependence of the $\tau$-function on the extended monodromy data?} Such a question is by no means new, and
has been addressed in the literature before by various sources \cite{Malgrange,Palmer,Bertola1,LR,ILP}. This problem of determining the dependence of the
$\tau$-function on the extended monodromy data has many important applications, one of the main ones being the problem of determining constant factors for
the asymptotics of $\tau$-functions \cite{LR,ILP,IP}. This subsection will be organized as follows: we first introduce the definition of the extended JMU
differential, in the context of the previous section. We then overview some of the main points given in \cite{IP} about the role of the Hamiltonian structure
of Painlev\'{e} equations in the problem of computing constant factors. We also state their conjectures. We conclude by showing that their conjectures hold
in the case of the isomonodromic system associated to the string equation.

Let us first define the extended JMU $\tau$-differential. We work again with the system \eqref{general-Y-system}. Let $\mathcal{T}$ denote the space of isomonodromic deformation parameters of this system\footnote{Actually, one can extend all the definitions naturally to the universal 
covering $\tilde{\mathcal{T}}$ of $\mathcal{T}$, and this is really where we want to define ${\bf d}_{\mathcal{T}}$. However, we will not be going far enough 
into this subject for this distinction to make much of a difference.}, and denote ${\bf d}_{\mathcal{T}}$ the differential in these parameters (note that we 
had previously used the notation ${\bf d}$ for this object). Associated to the system \eqref{general-Y-system} are a number of parameters which we 
refer to collectively as \textit{monodromy data}. For the system \eqref{general-Y-system}, the monodromy data will consist of a number of Stokes 
parameters. We denote these parameters by $\{m_\ell\}$, denote the space of these parameters by $\mathcal{M}$, and denote the differential in these 
parameters as ${\bf d}_{\mathcal{M}}$.
A specification of a solution to the isomonodromy equations (the zero curvature conditions) depend intrinsically on the monodromy data $\{m_{\ell}\}$, and thus the
JMU $\tau$-function also depends on these parameters. One is then led to wonder \text{how} $\tau$ depends on these parameters. This question can essentially
be answered if one can extend the JMU differential from a closed differential on $\mathcal{T}$ to a closed differential on all of 
$\mathcal{T}\times \mathcal{M}$. This can be accomplished through the following steps \cite{ILP}:
    \begin{enumerate}
        \item Define the following $1$-form on $\mathcal{T}\times\mathcal{M}$:
            \begin{equation}\label{omega0-def}
                \boldsymbol{\omega}_{0} := \left\langle \hat{A}(\zeta) {\bf d }_{\mathcal{T}} S(\zeta) S^{-1}(\zeta)\right\rangle + \left\langle \hat{A}(\zeta) {\bf d }_{\mathcal{M}}  S(\zeta) S^{-1}(\zeta)\right\rangle.
            \end{equation}
        This $1$-form obviously has the property that its restriction to $\mathcal{T}$ coincides with the usual JMU $\tau$-differential. Furthermore, one
        can show that 
            \begin{equation}\label{extended-2-form}
                \boldsymbol{\omega}_{0} := \left({\bf d}_{\mathcal{T}} + {\bf d}_{\mathcal{M}}\right)\boldsymbol{\omega}_{0}
            \end{equation}
        is a $2$-form on $\mathcal{M}$ only. In other words, the restriction of $\boldsymbol{\omega}_{0}$ to $\mathcal{T}$ vanishes identically; this is equivalent to 
        the fact that (i.) ${\bf d}_{\mathcal{T}} \boldsymbol{\omega}_{JMU} = 0$, and (ii.) $\boldsymbol{\omega}_{0}$ contains no cross-terms of the form $dt_k\wedge dm_{\ell}$.
        \item By means of asymptotic analysis, one can calculate (at least in principle) $\boldsymbol{\omega}_{0}$ explicitly. Once this expression is obtained, construct
        a $1$-form $\boldsymbol{\omega}_{correction}$ on $\mathcal{M}$ such that $d\boldsymbol{\omega}_{correction} = \Omega$; then, put
            \begin{equation}
                \hat{\omega} := \boldsymbol{\omega}_{0} - \boldsymbol{\omega}_{correction}.
            \end{equation}
        The $1$-form $\hat{\omega}$ then by construction is closed on $\mathcal{T} \times \mathcal{M}$, and its restriction to $\mathcal{T}$ agrees with the
        JMU $\tau$-differential. We are thus justified in calling $\hat{\omega}$ the \textit{extended} $\tau$-differential. Such a differential is of course
        not unique, as our construction of $\boldsymbol{\omega}_{correction}$ is defined only up to the addition of an exact differential on $\mathcal{M}$. 
    \end{enumerate}
With the definition of the extended $\tau$-differential in place, we now proceed to discuss the Hamiltonian aspects of the problem.
In \cite{IP}, the central role of the Hamiltonian structure of Painlev\'{e} equations with regards to the problem of evaluation of constant factors was
demonstrated. Let us briefly overview some of their main philosophical arguments; we will essentially be summarizing Section 2 of \cite{IP}.

Consider a completely integrable Hamiltonian system with Darboux coordinates $\{P_a,Q_a\}$, with Hamiltonians $\{H_k\}$ with respect to the times 
$\{t_k\}$. Denote the parameter space of times by $\mathcal{T}$; suppose the Darboux coordinates depend additionally on a collection of \textit{monodromy parameters} $\{m_{\ell}\}$,
    \begin{equation}
        Q_a = Q_a(t_k,m_{\ell}), \qquad\qquad P_a = P_a(t_k,m_{\ell}), 
    \end{equation}
and denote the parameter space of the monodromy parameters by $\mathcal{M}$.  We define the \textit{classical action differential} on the total space $\mathcal{T} \times \mathcal{M}$:
    \begin{equation}\label{omega-classical}
        \boldsymbol{\omega}_{cla} := \sum_a P_a dQ_a - \sum_k H_k dt_k = \sum_k \left(P_a \frac{\partial Q_a}{\partial t_k} - H_k\right)dt_k +\sum_{\ell}\left(\sum_a P_a \frac{\partial Q_a}{\partial m_{\ell}}\right)dm_{\ell}.
    \end{equation}
The fact that the system is a completely integrable Hamiltonian system implies that the differential is closed in the time parameters. In other words, if we define ${\bf d}_{\mathcal{T}} := \sum_k dt_k\frac{\partial}{\partial t_k}$, then
    \begin{equation}
        {\bf d }_{\mathcal{T}} \left(\boldsymbol{\omega}_{cla}\big|_{\{m_{\ell} = const.\}}\right) = 0.
    \end{equation}
This is nothing but the classical statement that the symplectic form defined by \eqref{symplectic-form} vanishes along the trajectories of the Hamiltonian flows. Note that in many cases, including our own, there is already a connection between the classical action and the isomonodromic $\tau$-function: namely,
we have that $d\log\tau = \sum_k H_k dt_k$, and so the $\tau$-function appears as a ``truncation'' of the classical action integral. If we take the total 
differential (on the whole of $\mathcal{T}\times\mathcal{M}$) of formula \eqref{omega-classical}, we find that
    \begin{equation}
        \left({\bf d}_{\mathcal{T}} + {\bf d}_{\mathcal{M}}\right) \boldsymbol{\omega}_{cla} = \sum_a {\bf d}_{\mathcal{M}} P_a \wedge {\bf d}_{\mathcal{M}} Q_a =:\Omega,
    \end{equation}
which is reminiscent of formula \eqref{extended-2-form}. This observation led Its and Prokhorov to make the following conjectures:
\begin{conjecture} (\cite{IP}.) Suppose the parameter space $\mathcal{T}\times \mathcal{M}$ is equipped with a symplectic structure $\Omega$. Then,
there exists a constant $\gamma \in \CC$ such that
    \begin{equation}
        \boldsymbol{\omega}_{0} = \gamma \Omega,
    \end{equation}
where $\boldsymbol{\omega}_{0}$ is the $2$-form defined by \eqref{extended-2-form}.
\end{conjecture}
\begin{conjecture} (\cite{IP}.) 
    There exists a function $G(P_a,Q_a,t_k)$, rational in the variables $\{P_a\},\{Q_a\}, \{t_k\}$, such that
        \begin{equation}
            \boldsymbol{\omega}_{0} = \gamma \boldsymbol{\omega}_{cla} + d G.
        \end{equation}
\end{conjecture}
These conjectures allow one to write a formula for the variation of the $\tau$-function in terms of the monodromy parameters, which in practice is much 
more efficient in application to the evaluation of constant factors than many earlier procedures. If we define the $\tau$-function as
$\log \tau := \int_{C} \hat{\omega}$, where $C \subset \mathcal{T}$ is a `nice' curve in the deformation parameter space,
the formula is (see Remark 3 of \cite{IP})
    \begin{equation}
        \frac{\partial}{\partial m_{\ell}}\log\tau = \sum_a P_a\frac{\partial Q_a}{\partial m_{\ell}}\bigg|_{\partial C} + \frac{\partial G}{\partial m_{\ell}}\bigg|_{\partial C}.
    \end{equation}
In \cite{IP}, the authors were able to verify this conjecture for the classical Painlev\'{e} transcendents. In fact, these conjectures hold in our situation 
as well, as was stated in Proposition \ref{tau-extended-theorem}. We present the proof of this theorem here:
    \begin{proof} (of Proposition \eqref{tau-extended-theorem}.)
        The proof of this proposition is a straightforward, albeit tedious, calculation. since $\mathcal{L}(\zeta)$ is degree $7$ in $\zeta$, one must in principle compute terms
        up to order $\zeta^{-8}$ in the expansion of $S(\zeta)$; however, the symmetry of $\mathcal{L}$, $S$ under conjugation by the matrix $\mathcal{S}$ actually implies one
        must only compute up to terms of  order $\zeta^{-7}$. This calculation involves (cf. the proof of Proposition \ref{RHP-Phi-prop}) determining the off-diagonal terms of
        the matrices $\Phi_k$ up to order $13$. Once one has successfully calculated the coefficients $\Phi_1,...,\Phi_7$, one can use formula \eqref{omega0-def} with the Hamiltonian variables as coordinates on the monodromy manifold $\mathcal{M}$ to compute the coefficients of $\boldsymbol{\omega}_{0}$. Calculating the $dP_U\wedge dQ_U$-coefficient of $({\bf d}_{\mathcal{T}} + {\bf d}_{\mathcal{M}})\boldsymbol{\omega}_{0}$, 
            \begin{equation*}
                \frac{\partial (\boldsymbol{\omega}_{0})_{Q_U}}{\partial P_U} - \frac{\partial (\boldsymbol{\omega}_{0})_{P_U}}{\partial Q_U} = \frac{1}{2}.
            \end{equation*}
        Similarly, the coefficients of the $dP_V\wedge dQ_V$, $dP_W\wedge dQ_W$ terms in $({\bf d}_{\mathcal{T}} + {\bf d}_{\mathcal{M}})\boldsymbol{\omega}_{0}$ are constant, and equal to $\frac{3}{2}$. On the other hand, we have the equalities
            \begin{align*}
                -\frac{3}{2}\frac{\partial H_k}{\partial Q_a} &= \frac{\partial (\boldsymbol{\omega}_{0})_{t_k}}{\partial Q_a} - \frac{\partial (\boldsymbol{\omega}_{0})_{Q_a}}{\partial t_k},\\
                -\frac{3}{2}\frac{\partial H_k}{\partial P_a} &= \frac{\partial (\boldsymbol{\omega}_{0})_{t_k}}{\partial P_a} - \frac{\partial (\boldsymbol{\omega}_{0})_{P_a}}{\partial t_k},
            \end{align*}
        for every $k\in \{1,2,5\}$, $a\in \{U,V,W\}$; all other coefficients of $({\bf d}_{\mathcal{T}} + {\bf d}_{\mathcal{M}})\boldsymbol{\omega}_{0}$ vanish identically
        (Note that we could have also inferred this constant from the relation of $\hat{\boldsymbol{\omega}}_{JMU}$ and $\boldsymbol{\omega}_{Okamoto}$).  
        Subtracting $\frac{1}{2}\boldsymbol{\omega}_{cla}$ from $\boldsymbol{\omega}_{0}$, we obtain the differential
            \begin{equation*}
                dG := \boldsymbol{\omega}_{0} - \frac{1}{2}\boldsymbol{\omega}_{cla}.
            \end{equation*}
        By construction, this differential is closed. Consequentially, it can be integrated up to a function $G = G(Q_U,Q_V,Q_W,P_U,P_V,P_W;t_1,t_2,t_5)$.
        Direct calculation shows that this function is
            \begin{equation*}
                G = \frac{1}{7}\left[3t_1H_1 + \frac{5}{2}t_2 H_2 + t_5H_5 - P_UQ_U - \frac{3}{2}P_VQ_V - \frac{3}{2}P_WQ_W\right],
            \end{equation*}
        as claimed.
    \end{proof}

\section{Discussion and Outlook.}
In summary, we have constructed a Riemann-Hilbert formulation of the $(3,4)$ string equation, which will appear as the model Riemann-Hilbert problem in the local analysis of the multi-critical quartic $2$-matrix model \cite{DHL3}. The string equation is equivalent to a $3+3$-dimensional, completely integrable 
non-autonomous Hamiltonian system. Furthermore, we were able to calculate an appropriate $\tau$-function for this system. Upon extending this $\tau$-function
to the canonical coordinates, we were able to verify Conjectures 1 and 2 of \cite{IP}, lending them further validity.

Aside from the completion of the work \cite{DHL3}, we hope to further investigate the large-parameter asymptotics of the above Riemann-Hilbert problem. This will be accomplished in a forthcoming work \cite{Nathan2}. The work \cite{DHL3} shows that the partition function of the critical two matrix model can be written in terms of $U(t_{5},t_{2},x)$; the behavior of
this solution is described in \cite{Nathan2}.

There is also an additional physical motivation for the study of these asymptotics. As observed by Crnkovi\'{c}, Ginsparg, and Moore \cite{CGM}, there should exist a ``renormalization group flow'' between the multicritical points of the $2$-matrix model. Formally, this observation says that, given a solution \sloppy$U(t_{5},t_{2},x), V(t_{5},t_{2},x)$ of the string equation
\eqref{string-equation}, if we make the scaling
    \begin{equation}
       u(t_{5},t_{2},x) := t_{5}^{2/5} U(t_{5},t_{2},t_{5}^{1/5}x), \qquad\qquad v(t_{5},t_{2},x) := t_{5}^{3/5} V(t_{5},t_{2},t_{5}^{1/5}x),
    \end{equation}
and take a formal limit as $t_{5}\to \infty$, then $v\to 0$, and $u\to \hat{u}(x)$, where $\hat{u}$ solves the Painlev\'{e} I equation, after a rescaling of the variables.
This Riemann-Hilbert problem ``flows'' to a $3\times 3$ version of the Painlev\'{e} I Riemann-Hilbert problem. The associated Lax pair
has appeared in the literature before \cite{JKT}, and this $3\times 3$ problem also seems to appear in the local parametrices of the critical energy, critical 
temperature (but non-critical external field) quartic 2-matrix problem \cite{DHL1,DHL3}. The analysis of this problem and the large-parameter asymptotics of
the Riemann-Hilbert problem described in this paper is the subject of \cite{Nathan2}. We also remark that it would be interesting to see if this degeneration
can be identified using the Hamiltonian formalism, in a similar manner to the $t_2\to 0$ limit discussed in \S2.

The partition function of the $2$-matrix model is identified with the partition function of a particular theory of minimal matter coupled to topological gravity \cite{Kontsevich,Witten1,Witten2}, which counts a class of intersection numbers on the moduli space of Riemann surfaces. This implies that the Riemann-Hilbert problem discussed above could be of use in enumeration of these intersection numbers; we hope to investigate this in the future.

In this work, we essentially gave no analysis of the solutions to the string equation. There are several fundamental questions that should be addressed:
    \begin{itemize}
        \item \textit{Irreducibility of the string equation.} Due to the similar nature of the Riemann-Hilbert problems of the $(3,4)$ string equation and
        the Painlev\'{e} I Riemann-Hilbert problem, it is natural to conjecture that \textit{the string equation admits no solutions in terms of classical
        functions, in the sense of \cite{Okamoto})}. Indeed, there is a procedure (\cite{Umemura}, see also \cite{Umemura2}) by which one can infer the irreducibility of solutions of a given Hamiltonian system. This procedure applies in principle to the string equation; it would be interesting to see if this method can be applied practically.
        \item \textit{The space of initial conditions \& Stokes manifold.} Aside from determining its generic dimension, we provided essentially no 
        analysis of the Stokes manifold associated to the string equation. The Stokes manifolds of the classical Painlev\'{e} equations, 
        in particular PI and PII, have a rich mathematical structure, and carry their own Poisson tensor, as well as an association to certain 
        cluster algebras \cite{LR,BertolaTarricone}. A more complete analysis of this Stokes manifold, as well as an accompanying analysis of the space of 
        initial conditions (cf. \cite{Okamoto} for the equivalent analysis for PII) is certainly needed.
        \item \textit{Evaluation of constant factor in the $\tau$-function}. So far, we have only calculated the $\tau$-differential, and thus the free energy of the multi-critical matrix model up to a multiplicative constant. This problem was first noticed in \cite{Douglas1}, who believed the problem could be resolved by appealing to the general theory of $\tau$-functions. It would be interesting if one could apply the calculations in Section \ref{tau-extension} of this work to this end.
    \end{itemize}

\appendix

\section{Rational reductions of the KP hierarchy and the string equation.}\label{AppendixA}
In this appendix, we overview the derivation of the equation \eqref{string-equation} as the string equation of an appropriate rational reduction of the KP hierarchy; we then recast this equation in matrix form, which sets up the framework for us to later realize the equation as arising from isomonodromy deformation. We do not attempt to give a comprehensive introduction to reductions of the KP hierarchy; for a full introduction, one should consult \cite{Dickey}, for example. This Appendix is meant to be self-contained. 
We remark that what appears in this section can be treated as purely formal computation; we will use what we develop here 
as a objects to be compared to what appears in Section \S2.

\subsection{The basics of KP, rational reductions, and string equations.}
We begin with a list of definitions:
    \begin{itemize}
        \item A \textit{pseudodifferential operator} is an expression of the form $X = \sum_{k\in \ZZ} X_i \partial^i$,
        where the coefficients $X_i$ are functions of $t_1 := x$, and possibly a collection of other variables $\{t_k\}$, 
        $\partial := \frac{\partial}{\partial x}$, and the symbol $\partial^{-1}$ is defined using the generalized Leibniz 
        rule
            \begin{equation*}
                \partial^{-1} \circ f = \sum_{k=0}^{\infty} (-1)^k f^{(k)} \partial^{-k-1} = f\partial^{-1} - f' \partial^{-2} + f''\partial^{-3} + ...
            \end{equation*}
        Note the relation $\partial^{-1} \circ \partial = \partial \circ \partial^{-1} = 1$, the identity operator.
        Such operators are interpreted as acting on functions of $x$.
        \item The \textit{purely differential part} or \textit{principal part} of a pseudodifferential operator $X$ is written $X_+$, and is defined to be
            \begin{equation*}
                X_+ := \sum_{k\geq 0} X_k \partial^k.
            \end{equation*}
        \item The \textit{order} of a pseudodifferential operator $X$ is the largest $k$ such that $X_k \neq 0$; if no such $k$ exists,
        we say the operator is of infinite order. One can interpret the order of $X_+$, with the word \textit{order} standing for the
        usual definition of order of a differential operator.
    \end{itemize}
We now define the \textit{KP operator}
    \begin{equation}
        \mathfrak{L} := \partial + \alpha_1 \partial^{-1} + \alpha_2 \partial^{-2} +  \alpha_3 \partial^{-3} + \cdots,
    \end{equation}
where the $\alpha_k$ are assumed to be functions of $t_1 := x$, and a (possibly infinite) collection of other ``times'' $\{t_k\}$. 
We define operators $\mathcal{A}_k := \mathfrak{L}^k_+$; note that since $\mathfrak{L}$ is of order $1$, the operators $\mathcal{A}_k$ are of order $k$. The \textit{KP hierarchy} is defined by the set of equations
    \begin{equation}
        \left[ \mathfrak{L} - \lambda, \frac{\partial}{\partial t_k} - \mathcal{A}_k\right] = 0, \qquad\qquad k = 1,2, ...
    \end{equation}
with the assumption that the eigenvalue $\lambda$ is independent of the $t_k$'s. It then follows that the flows 
$\frac{\partial}{\partial t_k} - \mathcal{A}_k$ pairwise commute (cf. \cite{Dickey}):
    \begin{equation}
        \left[\frac{\partial}{\partial t_k} - \mathcal{A}_k,\frac{\partial}{\partial t_k} - \mathcal{A}_j\right] = 0,\qquad\qquad k,j = 1,2, ...
    \end{equation}
which implies the integrability of this collection of equations. A \textit{rational reduction} of the KP hierarchy is obtained by requiring that a given power of the KP operator $\mathfrak{L}$ is purely differential, i.e. that
    \begin{equation}
        \mathfrak{L}^q \equiv \mathfrak{L}^q_+ = \mathcal{A}_q,\qquad \textit{or, equivalently,}\qquad \mathfrak{L}^q_- \equiv 0.
    \end{equation}
The resulting hierarchy of equations retains the property of integrability. If we require that $\mathfrak{L}^q$ is purely differential, 
we call the hierarchy the $\text{KdV}_q$ hierarchy (sometimes, this hierarchy is also called the $q^{th}$ Gelfand-Dickey hierarchy). These hierarchies also carry a natural bihamiltonian structure \cite{Adler,DrinfeldSokolov,Dickey}. If $q = 2$, the resulting hierarchy agrees with the well-known KdV hierarchy. For the $\text{KdV}_q$ hierarchy, we define the differential operator $Q$ by
    \begin{equation}
        Q := \mathfrak{L}^q.
    \end{equation}
We sometimes express the original KP operator as $\mathfrak{L} = Q^{1/q}$, when there is no cause for ambiguity.
A \textit{string equation} of the $\text{KdV}_q$ hierarchy is obtained by the requirement that 
    \begin{equation}
        [Q,P] = 1,
    \end{equation}
where the operator $P$ is a polynomial in the operator $Q^{1/q}_+$, with of the form (cf. \cite{DFGZJ}):
    \begin{equation}
        P := \sum_{\substack{k \geq 1\\ k\mod q\not\equiv 0}} \left(1 + \frac{k}{q}\right)t_{k+q} Q^{k/q}_+ = \sum_{\substack{k \geq 1\\k\mod q\not\equiv 0}} \left(1 + \frac{k}{q}\right)t_{k+q} \mathcal{A}_k.
    \end{equation}
the \textit{order} of a string equation is the largest index $k$ such that $c_k := \left(1 + \frac{k}{q}\right)t_{k+q} \neq 0$. 
If the order of the string equation is $p$, we call the string equation the $(q,p)$ string equation. Such equations 
generate additional symmetries of the $\text{KdV}_q$ hierarchy. These symmetries do not commute amongst themselves, but rather satisfy the relations of the
so-called $W_q$-algebras \cite{Dickey}. In particular, the algebra $W_2$ is equivalent to the Virasoro algebra, and the algebra $W_3$ is Zamolodchikov's algebra
\cite{Zamolodchikov}.

\subsection{The $(3,4)$ string equation.}
We now specialize to the case $q=3,p=4$. Consider the operator
    \begin{equation}
        Q := \partial^3 - \frac{3}{2}U\partial - \frac{3}{4}U' + \frac{3}{2}V,
    \end{equation}
where $U,V$ are functions of the variables $t_{5},t_{2}$, and $x$.
We take this operator to be the generator of the $\text{KdV}_3$ hierarchy, and will be interested in the $(3,4)$ string equation.

Let us briefly explain the choice of parametrization of $Q$ (it essentially comes from \cite{Douglas1}, and more generally \cite{DFIZ}). We momentarily denote $Q := Q(x)$ to stress that the variable of differentiation is $x$. Under any 
diffeomorphism $x \to x(\sigma)$, the composition
    \begin{equation}
        \tilde{Q}(\sigma) := \left(\frac{dx}{d\sigma}\right)^2 \circ Q(x(\sigma))\circ \left(\frac{dx}{d\sigma}\right) = \partial_{\sigma}^3 - \frac{3}{2}\tilde{U}\partial_{\sigma} -\frac{3}{4}\tilde{U}_{\sigma} + \frac{3}{2}\tilde{V}
    \end{equation}
retains the form of $Q(x)$, while $U$, $V$ transform as
        \begin{align}
        \tilde{U}(\sigma) &= U(x(\sigma)) \left(\frac{dx}{d\sigma}\right)^2 + 2 \{x,\sigma\},\\
        \tilde{V}(\sigma) &= V(x(\sigma)) \left(\frac{dx}{d\sigma}\right)^3,
    \end{align}
i.e. as an projective connection\footnote{Here, $\{x,\lambda\}$ denotes the Schwarzian derivative of $x$ with respect to $\sigma$: 
$\{x,\sigma\} := \frac{\dddot{x}}{\dot{x}} - \frac{3}{2} \left(\frac{\ddot{x}}{\dot{x} }\right)^2$, where $\dot{x} = \frac{dx}{d\sigma}$. This is the only place we will see the appearance of the Schwarzian derivative in this work; we hope our notation will not cause later confusion when $\{\cdot,\cdot\}$ will represent the Poisson bracket.} 
and as a rank 3 tensor, respectively. The operator $Q$ is then seen to act covariantly from the space of rank $1$ tensors to rank $2$ tensors. At the physical level, this makes consistent our choice of parametrization of the operator $Q$: $U$ will 
act as the classical analog of the stress-energy tensor for the underlying conformal field theory, and $V$ will 
ultimately be responsible for the non-perturbative $\ZZ_2$-symmetry breaking of the model \cite{Douglas1}, i.e. the shift in the magnetic field away from zero (see Subsection \ref{Z2-Symmetry-Breaking} for an interpretation of this statement).

From here on, we will not make any changes of coordinate, and so $\partial = \frac{\partial}{\partial x}$. Now, expanding $Q^{1/3}$ in pseudodifferential operators, we find that
    \begin{align*}
        Q^{1/3} &= \partial - \frac{1}{2} U\partial^{-1} + \frac{1}{2}\left[V + \frac{1}{2}U'\right]\partial^{-2} 
        - \frac{1}{4}\left[\frac{1}{3}U'' + U^2 +2V'\right]\partial^{-3} + \OO(\partial^{-4}),\\
        Q^{2/3} &= \partial^2 - U + \OO(\partial^{-1}),\\
        Q^{4/3} &= \partial^4 -2 U \partial^2 + 2\left[V-U'\right]\partial + \left[V' + \frac{1}{2}U^2 -\frac{5}{6}U''\right] + \OO(\partial^{-1}),\\
        Q^{5/3} &= \partial^5 -\frac{5}{2}U \partial^3 +\frac{5}{2}\left[V-\frac{3}{2}U'\right]\partial^2 
        + \frac{5}{4}\left[U^2-\frac{7}{3}U'' + 2V'\right]\partial\\
        &+ \frac{5}{4}\left[\frac{4}{3}V'' + UU'- \frac{2}{3}U''' - 2UV \right] + \OO(\partial^{-1}).
    \end{align*}
The $(3,4)$ string equation is then given by $[Q,P] = 1$, where
    \begin{equation}\label{P-operator}
        P := Q^{4/3}_+ +\frac{5}{3}t_{5} Q^{2/3}_+ = \partial^4 -\left[2 U -\frac{5}{3}t_{5}\right]\partial^2 + 2\left[V-U'\right]\partial + \left[V' + \frac{1}{2}U^2 -\frac{5}{6}U''-\frac{5}{3}t_{5} U\right].
    \end{equation}
We set the flow $t_4$ along $Q^{1/3}_+$ to $0$, as it can be removed by an overall translation $V \to V + ct_4$; further, we have set the flow $t_7 := \frac{3}{7}$.
By direct calculation, one can verify that
    \begin{prop}
        The $(3,4)$ string equation is equivalent to the following system on $U$, $V$:
        \begin{equation}
            \begin{cases}
                0 = \frac{1}{2}V'' - \frac{3}{2}UV + \frac{5}{2}t_{5} V + t_{2},\\
                0 = \frac{1}{12} U^{(4)} -\frac{3}{4}U''U -\frac{3}{8}(U')^2+\frac{3}{2}V^2 + \frac{1}{2}U^3 - \frac{5}{12}t_{5}\left(3U^2 - U''\right) + x.
            \end{cases}
        \end{equation}
    \end{prop}
Since we will not consider any other string equations in what follows, we will refer to the $(3,4)$ string equation as simply the string equation.
We assert that this equation is linearized on the Baker-Akhiezer function $\psi = \psi(\lambda;t_{5},t_{2},x)$:
    \begin{equation}\label{STRING}
        \begin{cases}
            P\psi = \partial_{\lambda} \psi,\\
            Q\psi = \lambda \psi.
        \end{cases}
    \end{equation}
The compatibility of this linear system is equivalent to the string equation \eqref{string-equation}.
This linearization is useful to us, since we can now write the action of the operators $P$, $Q$ as a closed-form system of linear
differential equations on the functions $\psi$, $\psi'$, and $\psi''$. 

\begin{prop}
Define the column vector
$\Psi(\lambda;t_{5},t_{2},x) := \langle \psi'' - \frac{3}{4}U\psi, \psi',\psi\rangle^{T}$. 
Then, the pair of equations on $\psi$ written above are equivalent to the following vector equations:
    \begin{equation}\label{matrix-linear-system}
        \begin{cases}
            \partial_{x} \Psi = \mathcal{Q} \Psi,\\
            \partial_\lambda \Psi = \mathcal{P} \Psi,
        \end{cases}
    \end{equation}
where the matrices $\mathcal{Q}(\lambda;t_{5},t_{2},x)$, $\mathcal{P}(\lambda;t_{5},t_{2},x)$ are given by the expressions
    \begin{align}
        &\mathcal{Q}(\lambda;t_{5},t_{2},x) := 
            E_{13}\lambda + 
            \begin{psmallmatrix}
                0 & \frac{3}{4}U & -\frac{3}{2}V\\
                1 & 0 & \frac{3}{4}U\\
                0 & 1 & 0
            \end{psmallmatrix},\\
        &\mathcal{P}(\lambda;t_{5},t_{2},x)  := 
            E_{13}\lambda^2 + 
                \begin{psmallmatrix}
                    0 & \frac{5}{3}t_{5} + \frac{1}{4}U & -V\\
                    1 & 0 & \frac{5}{3}t_{5}+\frac{1}{4}U\\
                    0 & 1 & 0
                \end{psmallmatrix}\lambda  \label{spectral-lax-operator}\\
             &+  \begin{psmallmatrix}
                    \frac{1}{2}V' -\frac{1}{12}U'' + \frac{1}{8}U^2 -\frac{5}{12}t_{5} U & \frac{1}{12} U''' - \frac{7}{16}UU'-\frac{3}{8}UV +\frac{5}{12}t_{5} U' + t_{2} & \frac{1}{16}(U')^2 -\frac{1}{8}U U'' + \frac{7}{32}U^3 + \frac{3}{4}V^2-\frac{5}{12}t_{5} U^2 + x\\
                    \frac{1}{2}V-\frac{1}{4}U & \frac{1}{6}U'' - \frac{1}{4}U^2 + \frac{5}{6}t_{5} U & -\frac{1}{12} U''' + \frac{7}{16}U U'-\frac{3}{8}UV -\frac{5}{12}t_{5} U' + t_{2}\\
                    \frac{5}{3}t_{5}-\frac{1}{2}U &\frac{1}{2}V+\frac{1}{4}U & -\frac{1}{2}V' -\frac{1}{12}U'' + \frac{1}{8}U^2 -\frac{5}{12}t_{5} U
                 \end{psmallmatrix}. \nonumber
    \end{align}
    The compatibility condition $[Q,P] = 1$ is equivalent to the compatibility condition for the linear system \eqref{matrix-linear-system}:
        \begin{equation*}
            [Q,P] = 1 \Longleftrightarrow \frac{\partial \mathcal{P}}{\partial x } - \frac{\partial \mathcal{Q}}{\partial \lambda} + [\mathcal{P},\mathcal{Q}] = 0.
        \end{equation*}
\end{prop}
\begin{remark}
    \textit{The role of the gauge group.} What will be important to us in the main text is that the eigenvalues of the spectral matrix
    $\mathcal{P}$ have a particular form. This form is fixed if we replace our definition of $\vec{\psi}$ by
        \begin{equation*}
            \vec{\psi} \to \langle\psi'' + a\psi' + b\psi,\psi'+c\psi,\psi\rangle^T,
        \end{equation*}
    where $a,b,c$ are any sufficiently differentiable functions of $(t_5,t_2,x)$. In other words, we can multiply $\vec{psi}$ in the left by any upper
    triangular matrix with $1$'s on the diagonal, and obtain the same results.
    Thus, the factor of $-\frac{3}{4}U\psi$ we added to the first entry of $\vec{\psi}$ 
    is an aesthetic choice, and is not essential. It is only to make the resulting matrices look more symmetric. 
\end{remark}

One also requires that the string equation is compatible with the other flows of the hierarchy; in our situation, we require that 
the string equation is compatible with the $t_2$ and $t_5$ flows. The linearization of these flows on $\psi$ are given
by
    \begin{align}
        \frac{\partial}{\partial t_{2}} \psi &= Q^{2/3}_+\psi = \psi'' - U\psi, \label{mu-KP}\\
        \frac{\partial}{\partial t_{5}} \psi &= Q^{5/3}_+\psi = (\lambda + V)\psi'' + \frac{1}{12}\left(U''-3U^2 - 6V'\right)\psi'\label{eta-KP} \\
        &+ \frac{1}{2}\left(UU'-\frac{1}{6}U'''- UV - 2\lambda U -\frac{5}{3}t_{5} V -\frac{2}{3}t_{2}\right)\psi, \nonumber
    \end{align}
where we have already utilized compatibility of the string equation with the $t_{5}$ flow to reduce the order of the right hand side
from $5$ to $3$. We can similarly write the above two equations in matrix form:
    \begin{prop}
        The equations \eqref{mu-KP}, \eqref{eta-KP} are equivalent to the following pair of matrix equations:
            \begin{align}
                \frac{\partial}{\partial t_{2}} \psi &= Q^{2/3}_+\psi \Longleftrightarrow  
                \frac{\partial \Psi}{\partial t_{2}} = M(\lambda;t_{5},t_{2},x)\Psi,\\
                \frac{\partial}{\partial t_{5}} \psi &= Q^{5/3}_+\psi \Longleftrightarrow  
                \frac{\partial \Psi}{\partial t_{5}} = E(\lambda;t_{5},t_{2},x)\Psi,
            \end{align}
        where the vector $\Psi = \langle \psi'' - \frac{3}{4}U\psi, \psi',\psi\rangle^{T}$, and the matrices 
        $M(\lambda;t_{5},t_{2},x)$, $E(\lambda;t_{5},t_{2},x)$ are defined to be
        \begin{equation}
            M(\lambda;t_{5},t_{2},x) =
                \begin{psmallmatrix}
                    0 & 1 & 0\\
                    0 & 0 & 1\\
                    0 & 0 & 0
                \end{psmallmatrix}\lambda + 
                \begin{psmallmatrix}
                    -\frac{1}{4}U & \frac{1}{4}U - \frac{3}{2}V & \frac{9}{16}U^2 - \frac{1}{4}U''\\
                    0 & \frac{1}{2}U & -\frac{1}{4}U- \frac{3}{2} V\\
                    1 & 0 & -\frac{1}{4}U
                \end{psmallmatrix}.
        \end{equation}
        \begin{align}
            &E(\lambda;t_{5},t_{2},x) = 
                \begin{psmallmatrix}
                    0 & 1 & 0\\
                    0 & 0 & 1\\
                    0 & 0 & 0
                \end{psmallmatrix}\lambda^2
               +\begin{psmallmatrix}
                    -\frac{1}{4}U & \frac{1}{4}U' - \frac{1}{2}V & \frac{5}{16}U^2 - \frac{1}{6}U''\\
                    0 & \frac{1}{2}U & -\frac{1}{4}U' - \frac{1}{2}V\\
                    1 & 0 & -\frac{1}{4}U
                \end{psmallmatrix}\lambda \\
                &+\begin{psmallmatrix}
                    \frac{1}{4}UV - \frac{1}{2}UU'\frac{1}{12}U''' - \frac{5}{6}t_{5} V-\frac{1}{3}t_{2} & e_{12} & e_{13}\\
                    \frac{1}{12}U'' - \frac{1}{4}U^2 + \frac{1}{2}V & \frac{5}{3}t_{5} V -\frac{1}{2}UV + \frac{2}{3}t_{2} & e_{23}\\
                    V & \frac{1}{12}U'' - \frac{1}{4}U^2 -\frac{1}{2}V & \frac{1}{4}UV + \frac{1}{2}UU'-\frac{1}{12}U''' -\frac{5}{6}t_{5} V -\frac{1}{3}t_{2}
                \end{psmallmatrix}, \nonumber
        \end{align}
    where
        \begin{align*}
            e_{12} :&= \frac{1}{8}(U')^2 - \frac{3}{16} U''U -\frac{1}{4}U'V + \frac{1}{8}V'U + \frac{5}{16}U^3 + \frac{5}{12}t_{5}\left(U''-3U^2 + 2V'\right) + x,\\
            e_{23} :&= \frac{1}{8}(U')^2 -\frac{3}{16}U''U +\frac{1}{4}U'V -\frac{1}{8}V'U + \frac{5}{16}U^3+ \frac{5}{12}t_{5}\left(U''-3U^2 - 2V'\right) + x, \\
            e_{13} :&=\frac{1}{8}U''V + \frac{1}{8}U'V' -\frac{9}{16}U^2V + \frac{25}{6}t_{5}^2 V + t_{2} U + \frac{5}{3}t_{5}t_{2}.
         \end{align*}
    \end{prop}
\begin{remark}
    Note that all of the matrices $\mathcal{P}, \mathcal{Q}, M$, and $E$ are traceless; this is a consequence of the fact that the generating operator $Q$ has no term of order $\partial^2$.
    We also comment here that in what follows, the matrices $\mathcal{Q}, M$, and $E$ can in fact be seen to arise on their own by requiring isomonodromy for the connection $\partial_{\lambda} - \mathcal{P}$. We present these matrices
    here for comparison to our results later.
\end{remark}
The requirement that all of the above equations are compatible with one another further determines the derivatives of $U(t_{5},t_{2},x)$,
$V(t_{5},t_{2},x)$ with respect to $t_{2}$ and $t_{5}$; this can either be done at the level of the operators $Q^{k/3}_+$, or can be 
performed using the matrices $\mathcal{P}, \mathcal{Q}, M$, and $E$. The result is the following:
    \begin{prop}
        The compatibility of the the operators $\lambda - Q$, $\frac{\partial}{\partial \lambda} - P$, $\frac{\partial}{\partial t_{2}} - Q^{2/3}_+$, $\frac{\partial}{\partial t_{5}} - Q^{5/3}_+$ (equivalently, the compatibility of the corresponding matrix equations)
        is equivalent to the string equation \eqref{string-equation}, and the following PDEs:
        \begin{align}
            \frac{\partial U}{\partial t_{2}} &= -2V', \label{U_mu}\\
            \frac{\partial V}{\partial t_{2}} &= \frac{1}{6}U''' - UU', \label{V_mu}\\
            \frac{\partial U}{\partial t_{5}} &=\frac{\partial}{\partial x}\left[-\frac{1}{6}UU'' + \frac{1}{8}(U')^2 + \frac{1}{4}U^3 - 
            \frac{1}{2}V^2 - \frac{5}{9}t_{5}\left(3U^2-U''\right) + \frac{4}{3}x\right],\label{U_eta}\\
            \frac{\partial V}{\partial t_{5}} &= \frac{\partial}{\partial x}\left[ \frac{1}{12}U''V - \frac{1}{4}U'V' + \frac{5}{16}U^2V - \left(\frac{5}{3}t_{5} + \frac{1}{4}U\right)^2V - t_{2} U\right]\label{V_eta}.
        \end{align}
    \end{prop}
\begin{proof}
    The proof is a direct calculation. We remark only that the compatibility conditions
        \begin{equation*}
            \left[\lambda - Q, \frac{\partial}{\partial t_{2}} - Q^{2/3}_+\right] = 0, \qquad\qquad \left[\lambda - Q,\frac{\partial}{\partial t_{5}} - Q^{5/3}_+\right] = 0
        \end{equation*}
    are enough to infer \eqref{U_mu}--\eqref{V_eta}; the remaining compatibility conditions are consistent with this calculation, and
    thus redundant.
\end{proof}

The above proposition justifies our notation for $' = \frac{\partial}{\partial x}$: all other derivatives can be rewritten in terms of
$\frac{\partial}{\partial x}$. We will sometimes refer to equations \eqref{U_mu}--\eqref{V_eta} above, along with the string equation,
collectively as the string equation, by a slight abuse of language. Also notice that, in the equations \eqref{U_mu}, \eqref{V_mu}, one can eliminate
$V$ to obtain that $U$ satisfies a scaled version of the \textit{Boussinesq equation}, with $t_{2}$ playing the role of the `time' variable \cite{Zakharov}:
    \begin{equation}
        \frac{\partial^2 U}{\partial t_{2}^2} = \frac{\partial}{\partial x} \left[\frac{1}{6}U''' - UU'\right].
    \end{equation}

\section{Construction of Darboux coordinates for string equations.} \label{AppendixB}
In Proposition \ref{Hamiltonian-Prop}, we introduced a set of Darboux coordinates that simply \textit{worked}, but did not explain their origin. For lower 
rank systems, sets of Darboux coordinates have been constructed using various methods \cite{MM,Takasaki,Marchal2}, but adaptation of these methods to the
present setting is a nontrivial task (this is not to say it is impossible, but it is certainly challenging). Here, we describe an algorithm for constructing
sets of Darboux coordinates for string equations of the type described in Appendix \ref{AppendixA}; this is essentially how we constructed the coordinates
described in Proposition \ref{Hamiltonian-Prop}. We have reason to believe that this algorithm can produce sets of Darboux coordinates for all string equations; nevertheless, we were unable to prove that this in general. This algorithm is essentially based on the fact that the Darboux coordinates for the Painlev\'{e}
I hierarchy are homogeneuous functions of $u,u',...$ and the times $t_1,t_3,...$ under an appropriate rescaling; by identifying the `correct' choice of this
homogeneity for the $(q,p)$ string equation, we can essentially uniquely determine the form of any reasonable set of coordinates. Let us proceed to the description of the algorithm.

Consider the operator
        \begin{equation}
            Q = \partial^{q} + u_{2}\partial^{q-2} + \cdots + u_{q},
        \end{equation}
    and the associated $(q,p)$ string equation:
        \begin{equation}\label{pq-string}
            [Q,P] = 1,
        \end{equation}
    where
        \begin{equation}
            P = Q^{p/q}_+ + \sum_{\substack{\ell = 1\\\ell\mod q \not\equiv 0}}^{p-1} \left(1 + \frac{\ell}{q}\right)t_{\ell+q} Q^{\ell/q}_+.
        \end{equation}
    Introduce the algebra
        \begin{equation}
            \mathbb{A}_{(p,q)} := \text{span}_{\CC}\left\{\left(u_b^{(s)}\right)^n t_a^m\mid b=2,...,q;a=1,...,p+q-1; n,s,m\in \ZZ_+\right\},
        \end{equation}
    where the superscript $s$ denotes $s$ differentiations with respect to $t_1$. Further, define the functional $\rho: \mathbb{A}_{(p,q)}\to \ZZ_+$ on monomials
        \begin{equation}
             \rho\left[\left(u_b^{(s)}\right)^n t_a^m\right] = n\left(b+s\right) + m\left(p+q-a\right).
        \end{equation}
    For instance, if we consider the monomial $t_7^2u_3'u_5$ as an element of the ring of differential polynomials arising from the $(5,6)$ string 
    equation, then
        \begin{equation*}
            \rho(t_7^2u_3'u_5)= 2\rho(t_7) + \rho(u_3') + \rho(u_5) = 2(11-7) + (3+1) + (5) = 17.
        \end{equation*}
    The functional $\rho$ then induces a grading on the algebra $\mathbb{A}_{(p,q)}$:
        \begin{equation}
            \mathbb{A}_{(p,q)} = \bigoplus_{\alpha=0}^{\infty} W_{\alpha},
        \end{equation}
    where $W_{\alpha}$ are finite dimensional subspaces
        \begin{equation}
            W_{\alpha} = \text{span }_{\CC} \left\{ \left(u_b^{(s)}\right)^n t_a^m \mid \rho\left[\left(u_b^{(s)}\right)^n t_a^m\right] = \alpha\right\};
        \end{equation}
    each of these spaces is finite dimensional. Furthermore, if we quotient the algebra $\mathbb{A}_{p,q}$ by the string equation \eqref{pq-string}
    and its compatible flows arising from KP, this grading passes to the quotient algebra $\tilde{\mathbb{A}}_{p,q}$, with 
    $\dim \tilde{W}_{\alpha} <\infty$ for any fixed $\alpha\in \ZZ_+$.
    A suitable set of Darboux coordinates must then satisfy the following constraints:
        \begin{enumerate}
            \item If the $(q,p)$ string equation is of order $d_2$ in the function $u_2$, ...., order $d_p$ in the function $u_p$, the map
            carrying the vector of functions $(u_2,...,u_2^{(d_2-1)}, u_3,...,u_3^{(d_3-1)},\cdots,u_p,...u_p^{(d_p-1)})$ to the Darboux coordinates $\{P_a,Q_a\}$ must be invertible.
            \item Any pair of canonical coordinates $P_a$, $Q_a$ must be homogeneous according to the grading induced by $\rho$; furthermore, if $P_a\in V_{\alpha_a}$, $Q_a\in V_{\beta_a}$, then
                \begin{equation}
                    \alpha_a + \beta_a = p+q.
                \end{equation}
        \end{enumerate}
    The last condition implies that the space of all possible such functions is finite-dimensional for any given $(p,q)$, and thus one can algorithmically produce a set of `good' Darboux coordinates using the procedure of section \S2, provided such coordinates exist satisfying these conditions. The
    benefit of this algorithm is that it reduces the problem of finding suitable Darboux coordinates to solving a finite dimensional system of linear
    equations; however, this system is usually overdetermined, so it is not clear from the outset that a solution exists. Nevertheless, this technique seems
    to work in a fair amount of generality, and it would thus be of interest to prove that this procedure always admits a solution.

    \begin{ex} \textit{The $(3,4)$ string equation.}
    Let us see concretely how this works in the case of the $(3,4)$ string equation. Here, the algebra $\mathbb{A}_{4,3}$ is built from the elements
    $U,V, t_1,t_2,$ and $t_5$, which have weights
        \begin{equation}
            \rho[U] = \rho[t_5] = 2, \quad \rho[V] = 3, \quad \rho[t_2] = 5,\quad \rho[t_1] = 6.
        \end{equation}
    The relevant spaces $W_j$ are then (after quotienting by the string equation)
        \begin{align*}
            W_2 &= \text{span }_{\CC}\left\{ U,t_5\right\},\\
            W_3 &= \text{span }_{\CC}\left\{ V,U'\right\},\\
            W_4 &= \text{span }_{\CC}\left\{V',U'',U^2,t_5U,t_5^2\right\},\\
            W_5 &= \text{span }_{\CC}\left\{U''',UV,UU',t_5V,t_5U',t_2\right\}.
        \end{align*}
    Any pair of canonical coordinates $P_a,Q_a$ must be homogeneous, and their degrees must add to $4+3 = 7$. It must also be the case that the map from
    $(U,U',U'',U''',V,V')$ to the Darboux coordinates must be invertible. These two facts imply that one pair of canonical coordinates have weights $2$ and $5$
    respectively, whereas two pairs of coordinates have weights $3$ and $4$, respectively: thus, we can take as an ansatz for these coordinates
        \begin{align*}
            Q_1&:= U + c_1 t_5,\qquad P_1 :=  c_2U'''+c_3UV+c_4UU'+c_5t_5V + c_6t_5U' + c_7t_2,\\
            Q_2&:= V,\qquad \qquad \quad P_2 := c_8V',\\
            Q_3&:= U',\qquad\qquad\quad P_3:=c_9U''+c_{10}U^2+c_{11}t_5U+c_{12}t_5^2.
        \end{align*}
    Here, the constants $c_1,...,c_{12}$ are to be determined. We further remark that we have already made a choice for how the coordinates of 
    weight $2$ split; in general, one should take the ansatz that $\langle Q_2,Q_3\rangle G = \langle U,V'\rangle$, for some invertible matrix $G$.
    We can now re-express the functions $(U,U',U'',U''',V,V')$ in terms of the variables $(Q_1,Q_2,Q_3,P_1,P_2,P_3)$:
        \begin{align*}
            U  &= Q_1-c_1t_5,\qquad V= Q_2,\qquad U'=Q_3,\qquad V'=P_2/c_8,\\
            U'' &= \left(P_3+2(c_1c_{10}-\frac{1}{2}c_{11})t_5Q_1 - c_{10}Q_1^2 + (c_1c_{11}-c_1^2c_{10}-c_{12})t_5^2\right)/c_9,\\
            U'''&= \left(P_1 -c_3Q_1Q_2-c_4Q_1Q_3+(c_1c_3-c_5)t_5Q_2 + (c_1c_4-c_6)t_5Q_3-c_7t_2\right)/c_2.
        \end{align*}
    The above set of coordinates satisfies properties 1 and 2; our claim is that there exists a choice (in fact, possibly many choices) of the parameters
    $\left\{c_{k}\right\}_{k=1}^{12}$ such that the coordinates defined above are indeed canonical. In order to check this, one proceeds by trying to integrate
    these coordinates to a Hamiltonian $H_1$ using Hamilton's equations
        \begin{equation*}
            \frac{\partial Q_j}{\partial x} = \frac{\partial H_1}{\partial P_j} ,\qquad \frac{\partial P_j}{\partial x} = -\frac{\partial H_1}{\partial Q_j},
        \end{equation*}
    $j=1,2,3$. We can calculate the right hand side of these equations as polynomials in $(P_1,P_2,P_3,Q_1,Q_2,Q_3)$ with coefficients in $\CC[x,t_2,t_5]$. 
    For instance, we have that
        \begin{align*}
            \frac{\partial Q_1}{\partial t_1} &= Q_3,\\
            \frac{\partial Q_2}{\partial t_1} &= P_2/c_8,\\
            \frac{\partial Q_3}{\partial t_1} &= \left((P_3-c_{10}Q_1^2 +2\left(c_1c_{10}-\frac{1}{2}c_{11}\right)t_5Q_1 + (c_1c_{11}-c_{12}-c_1^2c_{10})\right)/c_{9}.
        \end{align*}
    In order for a function $H_1$ to exist, the following differential (which has coefficients which are polynomial in the variables $P_k,Q_k$!)
        \begin{equation}
            \chi_1:= \sum_{k=1}^3\left(\frac{\partial Q_k}{\partial t_1} dP_k - \frac{\partial P_1}{\partial t_1} dQ_k\right)
        \end{equation}
    must be closed: $d\chi_1 = 0$ (here, the differential is taken in the coordinates $P_1,P_2,P_3,Q_1,Q_2,Q_3$). This imposes a set of linear constraints on the constants $\{c_j\}_{j=1}^{12}$; if one can find a choice of $\{c_j\}$ for which the above differential is closed, then one has constructed a set of
    suitable Darboux coordinates, since $\chi_1$ can be integrated to a function $H_1$, the Hamiltonian for the flow along $t_1$. In principle one must also check that the differentials $\chi_2$, $\chi_5$ arising from the other flows $t_2,t_5$ are also closed; in practice, it seems that these constraints come for free once the system $d\chi_1 = 0$ is solved. More constraints on the remaining coefficients can be imposed by requiring the stronger condition that $\frac{d}{dt_k}H_j = \frac{\partial}{\partial t_k} H_j$, and that the Hamiltonians have equal mixed partials. For the present example, one finds that
        \begin{equation}
            c_2 = -c_9,\quad c_3 = 0,\quad c_4 = 9c_9+2c_{10} \quad c_5 = 0, \quad c_8 = 12c_9, \quad c_{11} = c_6 + 5c_9,
        \end{equation}
    with the remaining parameters $c_1,c_6,c_7,c_9,c_{10},c_{12}$ left free. The Darboux coordinates we chose in the introduction correspond to setting
        \begin{equation}
            c_1=-\frac{4}{3}, \quad c_{6}=-\frac{7}{12}, \quad c_{9}=\frac{1}{12}, \quad c_{10}=0, \quad c_{12} = \frac{7}{18}.
        \end{equation}
    \end{ex}

In the above example, one can easily check by hand all of the computations; however, if $(p,q)$ are too large, then pen-to-paper calculations become less feasible. Nevertheless, it is straightforward to implement this algorithm in a symbolic computations language such as Maple or Mathematica. Indeed, we were
able to compute a suitable set of Darboux coordinates for the $(3,5)$ string equation (which appears in a multicritical case of the external source model in random matrix theory) using these ideas in Maple, which we will gladly share with the interested reader upon request.

\printbibliography
\end{document}